\documentclass[a4paper]{amsart}
\usepackage{a4wide}
\usepackage{amsmath,amsfonts,amssymb,amsthm,color}

\usepackage{enumitem}
\setlist{leftmargin=24pt}

\numberwithin{equation}{section}

\theoremstyle{plain}
\newtheorem{theorem}{Theorem}[section]
\newtheorem{proposition}[theorem]{Proposition}
\newtheorem{corollary}[theorem]{Corollary}
\newtheorem{lemma}[theorem]{Lemma}

\theoremstyle{remark}
\newtheorem{remark}[theorem]{Remark}
\newtheorem{definition}[theorem]{Definition}

\newcommand{\xx}{{\underline x}}
\newcommand{\HH}{{\mathcal H}}

\newcommand{\R}{{\mathbb R}}

\newcommand{\kk}{{\bf k}}

\newcommand{\x}{{\bf x}}

\newcommand{\yy}{{\underline y}}

\newcommand{\C}{{\mathbb C}}

\newcommand{\Y}{\mathcal Y}
\newcommand{\tgamma}{\widetilde{\gamma}}
\newcommand{\tx}{\widetilde{x}}
\newcommand{\ep}[2]{\eu^{\iu  \epsilon \phi(#1,#2)}}

\newcommand{\N}{{\mathbb N}}
\newcommand{\Z}{{\mathbb Z}}
\newcommand{\T}{{\mathbb T}}
\newcommand{\Hi}{\mathcal{H}}

\newcommand{\Id}{\mathbf{1}}
\newcommand{\eu}{\mathrm{e}}
\newcommand{\iu}{\mathrm{i}}
\newcommand{\di}{\mathrm{d}}
\newcommand{\set}[1]{\left\{ #1 \right\}}
\newcommand{\norm}[1]{\left\| #1 \right\|}
\newcommand{\scal}[2]{\left\langle #1, #2 \right\rangle}
\newcommand{\bra}[1]{\left\langle #1 \right|}
\newcommand{\ket}[1]{\left| #1 \right\rangle}
\newcommand{\sub}[1]{_{\mathrm{#1}}}
\newcommand{\Ll}{\lambda}    
\DeclareMathOperator{\Ran}{Ran}
\DeclareMathOperator{\Tr}{Tr}
\DeclareMathOperator{\tr}{tr}

\title[Parseval frames of exponentially localized magnetic Wannier functions]{Parseval frames of exponentially localized\\ magnetic Wannier functions}
\author{Horia D. Cornean, Domenico Monaco \and Massimo Moscolari}
\date{\today. Final version for arXiv.org}

\usepackage{hyperref}

\begin{document}

\begin{abstract}
Motivated by the analysis of gapped periodic quantum systems in presence of a uniform magnetic field in dimension $d \le 3$, we study the possibility to construct spanning sets of exponentially localized (generalized) Wannier functions for the space of occupied states.

When the magnetic flux per unit cell satisfies a certain rationality condition, by going to the momentum-space description one can model $m$ occupied energy bands by a real-analytic and $\Z^{d}$-periodic family $\set{P(\kk)}_{\kk \in \R^{d}}$ of orthogonal projections of rank $m$. 
A moving \emph{orthonormal basis} of $\Ran P(\kk)$ consisting of real-analytic and $\Z^d$-periodic Bloch vectors can be constructed if and only if the first Chern number(s) of $P$ vanish(es). Here we are mainly interested in the topologically obstructed case.

First, by dropping the generating condition, we show how to algorithmically construct a collection of $m-1$ \emph{orthonormal}, real-analytic, and $\Z^d$-periodic Bloch vectors. Second, by dropping the linear independence condition, we construct a \emph{Parseval frame} of $m+1$ real-analytic and $\Z^d$-periodic Bloch vectors which generate $\Ran P(\kk)$. Both algorithms are based on a two-step logarithm method which produces a moving orthonormal basis in the topologically trivial case.

A moving Parseval frame of analytic, $\Z^d$-periodic Bloch vectors corresponds to a Parseval frame of exponentially localized composite Wannier functions. We extend this construction to the case of magnetic Hamiltonians with an irrational magnetic flux per unit cell and show how to produce Parseval frames of exponentially localized generalized Wannier functions also in this setting.

Our results are illustrated in crystalline insulators modelled by $2d$ discrete Hofstadter-like Hamiltonians, but apply to certain continuous models of magnetic Schr\"{o}dinger operators as well.

\bigskip


\medskip

\end{abstract}

\maketitle

\setcounter{tocdepth}{1}
\tableofcontents


\section{Introduction} \label{sec:intro}
A large number of problems coming from the condensed matter physics of crystalline insulators can be mathematically described by means of a gapped Hamilton operator, where the gap in the spectrum is a threshold for the occupied states. In this framework it is important to have a suitable set of vectors in the Hilbert space that represents this energy window and encodes all the relevant physical information contained in the gapped spectral island. The reasons for that are multiple: from a theoretical point of view, one needs, for example, to justify the use of effective Hamiltonians, say of tight-binding nature, that simplify the analysis of the model while retaining the relevant features of the underlying physical system; from the computational point of view, the use of a suitable basis allows the efficient computation of physical quantities \cite{Goedecker99, Yates_et_al07, RestaVanderbilt07, Spaldin12, MarzariEtAl12}. 

\medskip 

If the gapped Hamiltonian is \emph{periodic}, then it is possible to pass to the $\kk$-space representation through the \emph{Bloch--Floquet transform} \cite{Kuchment16}, and the space of occupied states is described by a $\Z^d$-periodic family of projections $P(\kk)$, where $d \le 3$ is the dimension of the system. The associated Bloch bundle \cite{Panati07, MonacoPanati15}, the vector bundle over the Brillouin torus $\T^d:=\R^d/\Z^d \simeq (-1/2,1/2)^d$ whose fiber over $\kk$ is the space $\Ran P(\kk) \subset L^2((0,1)^d)$ of occupied states at fixed crystal momentum $\kk$, contains all the physical information pertaining the relevant gapped spectral island. A suitable set of vectors spanning this space consists then of sections $\set{\xi_a(\kk)}_{a \in \{1, \ldots, M\}}$ of the Bloch bundle. The inverse Bloch--Floquet transform
\begin{equation} \label{eqn:WF}
w_a(\xx + \gamma) :=\int_{(-1/2,1/2)^d} \di \kk \, \eu^{\iu  2 \pi \kk \cdot \gamma} \xi_a(\kk,\xx), \quad \xx \in (0,1)^d, \: \gamma \in \Z^d,
\end{equation}
defines then (\emph{composite}) \emph{Wannier functions}, which together with their translates span the gapped spectral island of the Hamiltonian. For the theoretical purposes mentioned above, it is important that these Wannier functions decay at infinity as fast as possible, e.g.~exponentially: by a Paley--Wiener-type argument, this is equivalent to requiring that the Bloch vectors $\xi_a(\kk)$ depend analytically on $\kk$~\cite{Cloizeaux, Kuchment16}. 

The first goal of this paper is to provide a \emph{constructive algorithm} that produces such a spanning set of localized Wannier functions, or rather the corresponding Bloch vectors. Notice that in general the existence of an \emph{orthormal basis} of (continuous, periodic) Bloch vectors is \emph{topologically obstructed} by the geometry of the Bloch bundle \cite{TKNN, Thouless84, Kohmoto85,Panati07,Monaco17} (see Section~\ref{sec:MR_Periodic}). However, if one relaxes the linear independence condition, then this topological obstruction can be circumvented \cite{AucklyKuchment18,Kuchment09,Kuchment16}: we provide a new proof of this result (formulated here as Theorem~\ref{thm:Main}) in the form of an algorithm for the construction of a ``redundant'', non-orthonormal spanning set of Bloch vectors (a \emph{Parseval frame}, to be precise). As we will detail in Section~\ref{sec:Gabor}, this datum is sufficient for example to recover spectral properties and construct effective models associated to the original Hamiltonian, even in lack of orthonormality and linear independence. 

This first result, applies to both continuous and discrete models of 2- and 3-dimensional gapped crystals subject to a constant magnetic field whose flux per unit cell satisfies a certain commensurability condition. In this paper, we choose as a recurring example the model of Hofstadter-like Hamiltonians, which are discrete analogues of magnetic Schr\"{o}dinger operators on the 2-dimensional lattice $\Z^2$ with uniform magnetic field in the orthogonal direction. Hofstadter-like Hamiltonians will be described in detail in Section~\ref{sec:Model}; the application of the general result Theorem~\ref{thm:Main} to these models is spelled out in Theorem~\ref{thm:Main_bis}.

\medskip

When periodicity is lost, there is no underlying vector bundle structure for the occupied states, and so the quest for a suitable set of spanning vectors becomes more complicated. However, the question is still well-defined in the original position-space representation, and one can ask whether a spanning set of localized (generalized) Wannier functions exists. 

To describe this non-periodic setting we specialize, for the sake of a more explicit presentation, to Hofstadter-like Hamiltonians. Up to an explicit unitary ``scaling'' transformation depending on the magnetic field (which corresponds roughly speaking to considering a supercell for the lattice), these Hamiltonians become periodic under the above-mentioned commensurability condition on the magnetic flux per unit cell, see Proposition~\ref{prop1}. Very concretely, in our case this condition requires that the strength of the magnetic field be a rational multiple of $2\pi$. For a value $b_0$ for which the latter condition is satisfied, we can apply the result from the first part, which yields a Parseval frame of localized Wannier functions for every gapped spectral island. 

Perturbing around $b_0 \in 2 \pi \mathbb{Q}$, say for $b=b_0+\epsilon$ with $\epsilon \ll 1$, the spectral island remains gapped but periodicity is lost when $\epsilon/(2\pi)$ is irrational \cite{Nenciu02}. Nevertheless, if $\epsilon$ is small enough, we will show how one can \emph{extend} the construction of a Parseval frame of localized generalized Wannier functions for the spectral island at magnetic field $b_0$ to one at magnetic field~$b$, see Theorem~\ref{thm:Main2} and its Corollary~\ref{CorGWB}. The result essentially depends on Combes--Thomas exponential estimates for the resolvent of the Hamiltonian. Once again, all the results for the type of discrete magnetic Hamiltonians discussed above can be extended to continuous models and magnetic Schr\"{o}dinger operators, see Remark~\ref{ExtensionToContinuous}.

\subsection{Comparison with the literature} \label{sec:literature}

We would like here to compare our results with the existing literature on the subject.

Our most novel contributions are Theorem~\ref{thm:Main2} and its Corollary~\ref{CorGWB} concerning the construction of spanning sets of localized Wannier functions for magnetic Hamiltonians with an \emph{irrational} magnetic flux. These provide in particular a generalization in the discrete setting of the results of \cite{CorneanHerbstNenciu16} which assume \emph{time-reversal symmetry} (hence zero magnetic flux per unit cell), consequently covering only the periodic case with no topological obstruction to the existence of an orthonormal basis of localized Wannier functions, and perturbations theoreof.

The literature on the interplay between topology and existence of localized Wannier functions in gapped \emph{periodic} quantum systems is much more extended. Concerning the topologically trivial case, namely when the existence of an orthonormal basis is not topologically obstructed, our proof of Theorem~\ref{thm:Main}\eqref{Main_3} (correspondingly Theorem~\ref{thm:Main_bis}\eqref{Main_bis_2}) provides a constructive algorithm for the results of \cite{Panati07, MonacoPanati15} (compare also \cite{Kuchment09}), which are instead obtained by abstract bundle-theoretic methods. There, the condition of topological triviality is obtained as a consequence of time-reversal symmetry: constructive algorithms for Bloch bases under this symmetry assumption have been recently investigated in \cite{FiorenzaMonacoPanati16_B, FiorenzaMonacoPanati16, CancesLevittPanatiStoltz17, CorneanHerbstNenciu16, CorneanMonacoTeufel17, CorneanMonaco17}.

Moving to the topologically non-trivial setting, to the best of our knowledge the only previous work which treated the problem of constructing an effective magnetic Hamiltonian starting from a topologically obstructed Fermi projection is \cite{FreundTeufel16}. However, there the authors only allow bounded magnetic potentials as perturbations, thus excluding magnetic fields which do not vanish at infinity.

Even though the use of ``non-orthogonal Wannier functions'' (that is, of redudant spanning sets of Wannier functions) is adopted in several computational schemes for electronic structure, quantum chemistry and density functional theory \cite{GalliParrinello92,Goedecker99,MarzariEtAl12}, in the mathematical literature the study of Parseval (or equivalently $1$-tight) frames of localized Wannier functions in the topologically obstructed case was initiated only in \cite{Kuchment09}, where an upper bound of the form $M \le 2^d \, m$ was given on the number of Bloch vectors needed to span $m$ isolated energy bands in dimension $d$. Improved estimates on $M$ for Bloch bundles in $d \le 3$ were announced in \cite{Kuchment16} and proved in \cite{AucklyKuchment18}, yielding $M=m+1$ as in Theorem~\ref{thm:Main}\eqref{Main_2} (correspondingly Theorem~\ref{thm:Main_bis}\eqref{Main_bis_2}). The results of \cite{AucklyKuchment18} prove the \emph{existence} of a Parseval $(m+1)$-frame of exponentially localized Wannier functions when $d\le 3$, again by means of general bundle-theoretic arguments. However, using powerful results from the theory of functions of several complex variables, their proofs allow to show that the corresponding Bloch vectors are analytic in the \emph{same analyticity domain} of the family of projections $\set{P(\kk)}_{\kk \in \R^d}$. Our techniques, even though more algorithmic and ``explicit'' in nature, allow instead only to exhibit Bloch vectors which are \emph{real-analytic}, that is, analytic in a complex strip around the real axis which could in principle be much smaller than the analyticity domain of the map $\kk \mapsto P(\kk)$. The problem of finding an explicit extension of these Bloch vectors to this domain (again through ``algorithmic'' methods) is an interesting research line, which we postpone to future investigation.

In this respect, it is also interesting to notice that exponential localization of generalized Wannier Parseval frames is somewhat optimal. Indeed, it was recently proved in \cite{DubailRead15} that if one requires the Wannier functions in a Parseval frame to be \emph{compactly supported}, then necessarily the Bloch bundle must be trivial.

Finally, in a very recent paper \cite{LudewigThiang19} the problem of the existence of a localized Wannier basis/Parseval frame (decaying faster than any polynomial) has been translated in a more general framework by using $K$-theoretic techniques, allowing to generalize some of the results in \cite{Kuchment09} to the setting of equivariant vector bundles over certain Riemannian manifolds. The existence of a constructive algorithm in such general framework is still an open problem.

\subsection{Structure of the paper}

The article consists of three main parts.

The first part is devoted to the presentation of the reference physical model, namely Hofstadter-like Hamiltonians, in Section~\ref{sec:Model}, and of our main contributions, whose proofs are postponed to the remaining two parts of the manuscript, in Sections~\ref{sec:MR_Periodic} and~\ref{sec:MR_irrational}. In particular, 
in Section~\ref{sec:MR_Periodic} we present Theorem~\ref{thm:Main} on spanning sets of vectors for periodic families of projections, as well as its application to Hofstadter-like Hamiltonians with rational magnetic flux, that is, Theorem~\ref{thm:Main_bis}. After that, in  Section~\ref{sec:MR_irrational} we describe the extension of these results to Hofstadter-like Hamiltonians with irrational but close-to-rational magnetic flux, namely Theorem~\ref{thm:Main2} and its Corollary~\ref{CorGWB}. A brief outline of the constructive, algorithmic proof for these results can be found in the discussion after Corollary~\ref{CorGWB}.

The second part of the paper focuses on our new proof of Theorem~\ref{thm:Main}, which is spread through Sections~\ref{sec:c1=0} to~\ref{sec:thmMain(ii)}: in Section \ref{sec:c1=0} we prove the third part of Theorem \ref{thm:Main} regarding the periodic topologically trivial case, in Section~\ref{sec:thmMain(i)} we prove the first part of Theorem~\ref{thm:Main} on the maximal number of orthonormal vectors, and finally in Section~\ref{sec:thmMain(ii)} we construct the Parseval frame for projections with non-trivial topology. In the context of Hofstadter-like Hamiltonians, by going back to position-space via the inverse Bloch--Floquet transform~\eqref{eqn:WF} this will prove also Theorem~\ref{thm:Main_bis}.

The third part of the paper is instead concerned with the proofs of Theorem~\ref{thm:Main2} and Corollary~\ref{CorGWB}. In Section~\ref{sec:Hofstadter} we show that the problem of finding Parseval frames of localized Wannier functions for general Hofstadter-like Hamiltonians can be recast in a more abstract problem regarding Fermi-like magnetic projections, see Definition~ \ref{dfn:Fermi-like} and Lemma~\ref{lemmaKN}. In the last Section~\ref{sec:MagneticParseval} we construct the required Parseval frame and conclude the proof of Theorem~\ref{thm:Main2}.

\medskip

\noindent \textbf{Acknowledgments.} The authors would like to thank G.~Panati, G.~Nenciu, G.~De Nittis and P.~Kuchment for inspiring discussions. Financial support from Grant 8021-00084B of the Danish Council for Independent Research $|$ Natural Sciences, from the German Science Foundation (DFG) within the GRK 1838 ``Spectral theory and dynamics of quantum systems'', and from the ERC Consolidator Grant 2016 ``UNICoSM -- Universality in Condensed Matter and Statistical Mechanics'' is gratefully acknowledged.

\section{The reference model: Hofstadter-like Hamiltonians} \label{sec:Model}

This Section presents our motivating physical model of crystalline systems in presence of uniform magnetic fields; we fix here some notation that will be used extensively throughout the paper. The quantum systems of interest are modelled by magnetic Hamiltonians, like e.g.~discrete tight-binding Hamiltonians where hoppings carry Peierls magnetic phases, or continuous Schr\"odinger operators of the form $\frac{1}{2} (-\iu \nabla - A)^2 + V$, where $A$ (respectively $V$) is the magnetic (respectively electrostatic) potential. To fix a reference model, to be used in applications of more general results, we now introduce Hofstadter-like Hamiltonians in the $2$-dimensional discrete setting.

Consider a set of $N$ points $\Y\subset [0,1]^2 \subset \R^2$, and a collection of functions $h_{\kk} \colon \Y \times \Y \to \R$ which are real-analytic and $\Z^2$-periodic in $\kk$, meaning that the maps $\R^2 \ni \kk \mapsto h_{\kk}(\yy,\yy') \in \R$ with fixed $\yy, \yy' \in \Y$ are real-analytic and $\Z^2$-periodic. Let us also introduce the skew-symmetric Peierls magnetic phase $\phi \colon \R^2\times\R^2 \to \R$, defined by
\begin{equation} \label{eqn:Peierls}
\phi(\x,\x'):=(x_1'x_2-x_2'x_1)/2= \{{\bf e}_3\cdot (\x' \times \x)\}/2.
\end{equation}

We define, for $b\in \R$, the bounded operator in $\ell^2(\Z^2\times \Y)\simeq\ell^2(\Z^2)\otimes \ell^2(\Y)$ (note that $\ell^2(\Y) \simeq \C^N$) given by the following matrix elements:
\begin{equation} \label{ianuar1}
H_b(\gamma,\yy;\gamma',\yy') := \eu^{\iu b\phi(\gamma+\yy,\gamma'+\yy')}\mathcal{T}(\gamma-\gamma'; \yy, \yy'), \quad \text{where} \quad \mathcal{T}(\gamma;\yy,\yy') := \int_\Omega \di \kk \, \eu^{\iu 2\pi \kk\cdot \gamma}h_{\kk}(\yy,\yy')
\end{equation}
with $\gamma,\gamma'\in\Z^2$, $\yy,\yy' \in \Y$, and the integral defining $\mathcal{T}$ is performed over $\Omega := (-1/2,1/2)^2$. Then $h_{\kk}$ can be recovered via
\[ h_{\kk}(\yy,\yy') = \sum_{\gamma\in \Z^2}\eu^{-\iu 2\pi \kk\cdot \gamma} \mathcal{T}(\gamma; \yy, \yy'), \quad \yy,\yy' \in \Y, \: \kk \in \R^2. \]
The resulting operator $H_b$ will be called an \emph{Hofstadter-like Hamiltonian}; the original Hofstadter model \cite{Hofstadter76} would correspond to $\Y = \{(0,0)\}$ (i.e. $N=1$),  $h_{\kk}=2\cos(2\pi k_1)+2\cos(2\pi k_2)$, while the Peierls phase would be written in the Landau gauge and equal $\phi_L(\x,\x'):=(x_2-x_2')(x_1+x_1')/2$. The magnetic phase $\eu^{\iu b \phi(\cdot, \cdot)}$ in front of the ``hopping'' $\mathcal{T}$ models the presence of a uniform magnetic field ${\bf B} := b \, {\bf e}_3$ orthogonal to the $2$-dimensional crystal. The quantity $-2 b \phi({\bf e}_1, {\bf e}_2) = {\bf B} \cdot ({\bf e}_1 \times {\bf e}_2)$ is then the magnetic flux per unit cell, to be measured in units of the magnetic flux quantum (which equals $1/2\pi$ in our units).

When $b=0$ the spectrum of $H_0$ is absolutely continuous and it is given by the range of the $N$ eigenvalues $E_j(\kk)$ of $h_\kk$ as functions of $\kk$, that is, 
\[ \sigma(H_0) = \set{E \in \R : E_j(\kk)=E \text{ for some } j \in \set{1,\dots,N}, \; \kk \in \Omega}. \]
The graph of the function $E_j$ is usually called the $j$-th energy (Bloch) band. Several analytic properties of these non-magnetic energy bands are discussed in \cite{AvronSimon78}, see also \cite{ZaidenbergKreinKuchmentPankov75} for the infinite-dimensional generalization from $\kk$-dependent matrices to linear operators on Banach spaces.

\medskip

If the magnetic field strength $b$ is such that $b=b_0$ where $b_0/(2\pi)$ is rational, i.e.\ there exists $q\in\mathbb{N}$ such that $b_0q \in 2\pi \mathbb{Z}$, then $H_{b_0}$ is unitarily equivalent to a periodic operator. Notice that the condition $b_0\in 2\pi \mathbb{Q}$ implies that the magnetic flux per unit cell is a rational multiple of the flux quantum: we will thus call this the ``rational flux'' condition. In order to formulate this unitary equivalence more precisely, and to be able to study also values of the magnetic field strength $b$ which are close to $b_0$, we will use the common technique of enlarging the unit cell in order to have an integer-flux magnetic field: we introduce the new lattice $\Gamma_q := (q\Z) \times \Z \simeq \Z^2$ and denote by $\Y_q$ its fundamental cell, namely
\[ \Y_q = \mathcal{B}_q \times \Y \subset \R^2, \quad \text{where} \quad \mathcal{B}_q := \set{(0,0),\ldots, (q-1,0)} \subset \R^2. \]
Hence, every point in the crystal can be uniquely represented as
\[ \tilde{\eta} + \xx , \quad \tilde{\eta} \in \Gamma_q, \: \xx \in \Y_q , \]
where 
$$\tilde{\eta}=(q\eta_1,\eta_2), \quad \eta_{1},\eta_2\in \Z, \qquad \xx=\yy + \underline{z} , \quad   \underline{z} \in \mathcal{B}_q, \: \yy \in  \Y.$$
In the following, we will often naturally identify $\tilde{\eta}=(q\eta_1,\eta_2) \in \Gamma_q$ with $\eta=(\eta_1,\eta_2) \in \Z^2$, and correspondingly identify e.g.~the magnetic phases $\eu^{\iu \phi(\tilde{\eta},\tilde{\eta}')} = \eu^{\iu q \phi(\eta, \eta')}$ or, with a little abuse of notation, the matrix elements $\mathcal{T}((q\eta_1,\eta_2);\xx)=\mathcal{T}(\tilde{\eta}+\underline{z};\yy)$.

With this notation, we then have following result:
\begin{proposition}\label{prop1}
Assume that $b_0 q\in 2\pi \mathbb{Z}$ as above. Set $Q=Q(q):=qN$. Then, for every $\epsilon\ge 0$, there exists a family of $Q\times Q$ self-adjoint matrices $h_{\kk,b_0+\epsilon}$ which is real-analytic and $\Z^2$-periodic as a function of $\kk$, real-analytic as a function of $\epsilon$, and such that $H_{b_0+\epsilon}$ is unitarily equivalent via a unitary operator $U_{b_0+\epsilon}$ to an operator in $\ell^2(\Z^2)\otimes \C^{Q}$ given by the matrix elements
\begin{equation} \label{apr1'}
\widetilde{\HH}_\epsilon(\gamma,\xx;\gamma',\xx') := \eu^{\iu\epsilon q\phi(\gamma,\gamma')} \mathcal{T}_\epsilon(\gamma-\gamma';\xx,\xx'), \text{ where } \mathcal{T}_\epsilon(\gamma;\xx,\xx') := \int_{\Omega} \di\kk \, \eu^{\iu 2\pi \kk\cdot \gamma} h_{\kk,b_0+\epsilon}(\xx,\xx'),
\end{equation}
with $\gamma,\gamma'\in\Z^2$ and $\xx,\xx' \in  \set{1,\ldots,Q}$.
\end{proposition}
\begin{proof}
Write $b_0=2\pi p/q$ with $p,q\in\mathbb{Z}$ coprime, and let $b=b_0+\epsilon$, where $\epsilon>0$. 
Define the unitary operator from $\ell^2(\Z^2) \otimes \ell^2(\Y)$ to $ \ell^2(\Gamma_q)\otimes \ell^2(\Y_q) $ acting on $ f\in \ell^2(\Z^2)\otimes \ell^2(\Y) $ by
\begin{align}\label{iunie1}
[U_b f](\tilde{\eta},\xx) &:= \eu^{\iu b_0 \tilde{\eta}_1\tilde{\eta}_2/2} \eu^{\iu b\phi(\tilde{\eta},\xx)} f(\tilde{\eta}+\underline{z},\yy) \\
& = \eu^{\iu \pi p \eta_1\eta_2} \eu^{\iu b\phi(\tilde{\eta},\xx)}f(\tilde{\eta}+\underline{z},\yy) \quad \tilde{\eta} \in \Gamma_q, \: \xx \in \Y_q ,\nonumber
\end{align}
where we have used the unique decomposition $\Y_q \ni \xx=\yy + \underline{z}$, $\underline{z} \in \mathcal{B}_q, \: \yy \in  \Y$.
We note the identity
\begin{equation*}
\phi(\tilde{\eta}+\xx,\tilde{\eta}'+\xx') = \phi(\tilde{\eta},\tilde{\eta}') + \phi(\tilde{\eta}-\tilde{\eta}',\xx+\xx') +\phi(\xx,\xx')+\phi(\xx,\tilde{\eta})-\phi(\xx',\tilde{\eta}').
\end{equation*}

By rotating $H_b$ with $U_b$ we have
\begin{multline*}
[U_b H_b U_b^*](\tilde{\eta},\xx;\tilde{\eta}',\xx')=\eu^{\iu \epsilon \phi(\tilde{\eta},\tilde{\eta}')}\cdot \\
\cdot \eu^{\iu b_0 \tilde{\eta}'_1(\tilde{\eta}_2-\tilde{\eta}_2')}\eu^{\iu b_0 (\tilde{\eta}_1-\tilde{\eta}_1')(\tilde{\eta}_2-\tilde{\eta}_2')/2}\eu^{\iu b\phi(\tilde{\eta}-\tilde{\eta}',\xx+\xx')}\eu^{\iu b\phi(\xx,\xx')}\mathcal{T}(\tilde{\eta}-\tilde{\eta}';\xx,\xx').
\end{multline*}
We observe that $b_0 \tilde{\eta}'_1(\tilde{\eta}_2-\tilde{\eta}_2')=2\pi p\eta_1'(\eta_2-\eta_2')\in 2\pi\mathbb{Z}$, thus $\eu^{\iu b_0 \tilde{\eta}'_1(\tilde{\eta}_2-\tilde{\eta}_2')} = 1$ and
\begin{equation*}
[U_bH_bU_b^*](\tilde{\eta},\xx;\tilde{\eta}',\xx')=\eu^{\iu \epsilon \phi(\tilde{\eta},\tilde{\eta}')} (-1)^{p (\eta_1-\eta_1')(\eta_2-\eta_2')}
\cdot \eu^{\iu b\phi(\tilde{\eta}-\tilde{\eta}',\xx+\xx')}\eu^{\iu b\phi(\xx,\xx')}\mathcal{T}(\tilde{\eta}-\tilde{\eta}';\xx,\xx').
\end{equation*}
Upon the identification of $\tilde{\gamma} = (q \gamma_1, \gamma_2) \in \Gamma_q$ with $\gamma = (\gamma_1, \gamma_2) \in \Z^2$ as in the comments before the statement, every operator on $\ell^2(\Gamma_q) \otimes \ell^2(\Y_q)$ is identified with an operator on $\ell^2(\Z^2) \otimes \ell^2(\Y_q)$. In particular, the above unitary conjugation of the Hofstadter-like Hamiltonian can be seen as acting in $\ell^2(\Z^2)\otimes \ell^2(\Y_q)$ with matrix elements
\begin{equation} \label{febru5}
\begin{aligned}
{\widetilde{\mathcal{H}}}_\epsilon(\gamma,\xx;\gamma',\xx')& :=\eu^{\iu \epsilon q \phi(\gamma,\gamma')}
(-1)^{p (\gamma_1-\gamma_1')(\gamma_2-\gamma_2')}\eu^{\iu (b_0+\epsilon )(\gamma_2-\gamma_2')(\xx_1+\xx_1')/2}  \\ 
& \qquad \cdot \eu^{-\iu q(b_0+\epsilon )(\gamma_1-\gamma_1')(\xx_2+\xx_2')/2}\eu^{\iu (b_0+\epsilon)\phi(\xx,\xx')} \mathcal{T}((q(\gamma_1-\gamma_1'),\gamma_2-\gamma_2');\xx,\xx').
\end{aligned}
\end{equation}
Notice that the whole expression on the right-hand side of the above, with the exception of the phase $\eu^{\iu \epsilon q \phi(\gamma,\gamma')}$, depends only on the difference $\gamma-\gamma'$. Identifying $\ell^2(\Y_q) \simeq \ell^2(\mathcal{B}_q) \otimes \ell^2(\Y) \simeq \C^q \otimes \C^N \simeq \C^{qN}$, we can thus determine a new Bloch fiber for ${\widetilde{\mathcal{H}}}_\epsilon$, which will be a matrix of size $qN\times qN$ equal to
\begin{align*}
h_{\kk,b_0+\epsilon }(\xx;\xx'):=\sum_{\gamma\in \Z^2}&  \eu^{-\iu 2\pi \kk\cdot \gamma} (-1)^{p \gamma_1\gamma_2} \eu^{\iu (\pi p/q+\epsilon/2 )\gamma_2(\xx_1+\xx_1')} \\ 
&\cdot \eu^{-\iu (\pi p+q\epsilon/2 )\gamma_1(\xx_2+\xx_2')}\eu^{\iu (b_0+\epsilon)\phi(\xx,\xx')} \mathcal{T}((q\gamma_1,\gamma_2);\xx,\xx').
\end{align*}
Moreover, the family of matrices $h_{\kk,b_0+\epsilon}$ is real-analytic in both $\kk$ and $\epsilon$. This is a simple direct consequence of the exponential decay of $\mathcal{T}(\gamma)$ as a function of $\gamma$, which in turn is a consequence of the real-analyticity in $\kk$ of the original $h_\kk$.
\end{proof}

Let us now use the above result to show how the original rational flux Hamiltonian $H_{b_0}$ is unitarily equivalent to a fibered operator. Indeed, if $\epsilon=0$, the Hamiltonian $\widetilde{\HH}_0$ is periodic, that is, it commutes with the usual translation operators by shifts in $\Z^2$. Therefore, it is possible to diagonalize it by using the usual  Bloch--Floquet theory \cite{ReedSimon4, Kuchment93, FreundTeufel16, Kuchment16}.  Consider the Bloch--Floquet transform $\mathcal{U}\sub{BF}$ defined, for every $ f \in C_0^\infty(\Z^2)\otimes \ell^2(\Y_q) $, as  
	\[ (\mathcal{U}\sub{BF} f)_{\kk}(\xx):= \sum_{{\eta} \in \Z^2} \eu^{-\iu 2 \pi \kk \cdot {\eta}} f({\eta},\xx) \,, \quad  \kk \in \Omega, \: \xx \in \Y_q, \]
	and then extended by continuity to a unitary operator $\mathcal{U}\sub{BF} \colon \ell^2(\Z^2) \otimes \ell^2(\Y_q) \to L^2(\Omega) \otimes \ell^2(\Y_q)$.
	
	Let us introduce the group of (\emph{modified}) \emph{magnetic translations} $\widehat\tau_{b_0,{\tilde{\eta}}}$ defined for every ${\tilde{\eta}} \in \Gamma_q$ by
	\begin{equation}
	\label{MagneticTranslation} \left[\widehat\tau_{b_0,{\tilde{\eta}}}f\right](\tilde{\gamma},\xx):= \eu^{\iu b_0 {\tilde{\eta}}_1{\tilde{\eta}}_2/2}\eu^{\iu b_0 \phi(\tilde{\gamma}+\xx,\tilde{\eta})} f({\tilde{\gamma}}-{\tilde{\eta}},\xx), \quad f \in  \ell^2(\Gamma_q)\otimes \ell^2(\Y_q)\, . 
	\end{equation}
	We stress that the phase factor $\eu^{\iu b_0 {\tilde{\eta}}_1{\tilde{\eta}}_2/2}$ is crucial in order to have a \emph{unitary representation} of the group $\Z^2$ (that is, $\widehat\tau_{b_0,{\tilde{\eta}}} \widehat\tau_{b_0,{\tilde{\gamma}}}= \widehat\tau_{b_0,{\tilde{\eta}}+{\tilde{\gamma}}}$ for ${\tilde{\gamma}}, {\tilde{\eta}} \in \Gamma_q$), instead of just a projective one (compare Remark~\ref{rmk:MatrixCommuteMagn} below), when $b_0 \in 2\pi\mathbb{Q}$.
	
	Define the following operator: 
	\begin{equation} \label{MagneticBF}
	\left(\mathcal{U}\sub{mBF}g\right)_{\kk}(\xx) :=\left(\mathcal{U}\sub{BF} U_{b_0}g\right)_{\kk}(\xx) = \sum_{{\eta} \in \Z^2} \eu^{-\iu 2 \pi \kk \cdot {\eta}} (\widehat\tau_{b_0,{-(q\eta_1,\eta_2)}}g)(0,\xx),  \quad g \in  C^\infty_0(\Z^2)\otimes \ell^2(\Y) \,.
	\end{equation}
	The unitary operator $\mathcal{U}\sub{mBF}$ is a modified Bloch--Floquet transform. Then we have the identity
	\begin{equation} \label{eqn:FiberedHb0}
	\mathcal{U}\sub{BF} \widetilde{\HH}_0 \mathcal{U}\sub{BF}^* = 	\mathcal{U}\sub{mBF} H_{b_0} \mathcal{U}\sub{mBF}^* = \int_{\Omega}^{\oplus} d\kk\, h(\kk) \,,
	\end{equation}
	where the fiber Hamiltonian $h(\kk) \equiv h_{\kk,b_0}$ is periodic in $\kk$ with respect to shifts in the dual lattice $\set{\kk \in \R^2 : \kk \cdot {\eta} \in 2 \pi \Z \: \forall \, {\eta} \in \Z^2} \simeq \Z^2$ and acts on functions with fixed crystal momentum $\kk$, so only on the degrees of freedom in the supercell $\Y_q$. 

If the Hamiltonian $H_{b_0}$ is gapped, then its Fermi projection $\Pi_{b_0}$ onto the (finite) gapped spectral island is also unitarily equivalent to a direct integral $\int_{\T^2}^{\oplus} \di \kk \, P(\kk)$ of (finite-rank) projections, and the same properties of regularity and periodicity in $\kk$ claimed in Proposition~\ref{prop1} for the fibers $h(\kk)$ hold for $P(\kk)$ as well.

While the original Hamiltonian $H_{b_0}$ commutes with the modified magnetic translations defined in \eqref{MagneticTranslation}, we stress, once again, that the Hamiltonian $\widetilde{\HH}_0$ is truly a periodic operator. Indeed, $\widetilde{\HH}_0$ commutes with the usual translations defined, for every $f \in \ell(\Z^2) \otimes \C^Q$, by 
	$$
	\left[\tau_{0,\eta} f\right](\gamma,\xx):=f(\gamma-\eta,\xx)\, , \qquad \eta \in \Z^2 \, .
	$$
This is reflected by \eqref{eqn:FiberedHb0}, where $\widetilde{\HH}_0$ is fibered by the Bloch--Floquet transform $\mathcal{U}\sub{BF}$.
The situation is different when we consider an irrational magnetic flux. Indeed, the original perturbed Hamiltonian $H_{b_0+\epsilon}$ commutes with the magnetic translations associated to $b_0 + \epsilon$, that is the unitary operator defined, for every $f \in \ell(\Z^2) \otimes \ell(\Y)$, by 
\begin{equation}
\left[\tau^{(b_0+\epsilon)}_{\eta}f\right](\gamma,\xx):= \eu^{\iu (b_0+\epsilon) \phi(\gamma+\xx,\eta)} f(\gamma-\eta,\xx), \quad \eta \in \Z^2 \, . 
\end{equation}
Instead, because of the action of the unitary operator $U_{b_0+\epsilon}$, the Hamiltonian ${\widetilde{\mathcal{H}}}_\epsilon$ commutes with the unitary operator defined, for every $f \in \ell(\Z^2) \otimes \C^Q$, by 
\begin{equation}
\label{eqn:tau}
\left[\tau_{\epsilon,\eta}f\right](\gamma,\xx):= \eu^{\iu \epsilon q \phi(\gamma,\eta)} f(\gamma-\eta,\xx), \quad \eta \in \Z^2 \, . 
\end{equation}  

\begin{remark} \label{rmk:MatrixCommuteMagn}
More generally, notice that any operator $A$ on $\ell^2(\Z^2) \otimes \C^Q$ whose matrix elements are of the form
\[ A(\gamma,\xx;\gamma',\xx'):= \eu^{\iu \epsilon q \phi(\gamma,\gamma')} a_{\epsilon}(\gamma-\gamma';\xx,\xx'), \quad \gamma,\gamma' \in \Z^2, \: \xx, \xx' \in \set{1, \ldots, Q}, \]
	commutes with the magnetic translations $\tau_{\epsilon,\eta}$, $\eta \in \Z^2$, defined in~\eqref{eqn:tau}. Indeed for $f \in \ell^2(\Z^2) \otimes \C^Q$
	\begin{equation} \label{luglio2}
	\begin{aligned}
	{[A \tau_{\epsilon,\eta} f]}(\gamma,\xx) & = \sum_{\gamma' \in \Z^2} \sum_{\xx=1}^{Q} \eu^{\iu \epsilon q \phi(\gamma,\gamma')} a_{\epsilon}(\gamma-\gamma';\xx,\xx') \, \eu^{\iu \epsilon q \phi(\gamma',\eta)} f(\gamma'-\eta,\xx') \\
	& = \sum_{\gamma''=\gamma'-\eta \in \Z^2} \sum_{\xx=1}^{Q} \eu^{\iu \epsilon q \phi(\gamma,\gamma''+\eta)} a_{\epsilon}(\gamma-\gamma''+\eta;\xx,\xx') \, \eu^{\iu \epsilon q \phi(\gamma'' + \eta ,\eta)} f(\gamma'',\xx') \\
	& = \eu^{\iu \epsilon q \phi(\gamma,\eta)} \sum_{\gamma'' \in \Z^2} \sum_{\xx=1}^{Q} \eu^{\iu \epsilon q \phi(\gamma-\eta,\gamma'')} a_{\epsilon}(\gamma-\eta-\gamma'';\xx,\xx') \, f(\gamma'',\xx') \\
	& = [\tau_{\epsilon, \eta} A f](\gamma,\xx),
	\end{aligned}
	\end{equation}
	where we repeatedly used the skew-symmetry of the Peierls magnetic phase $\phi(\cdot,\cdot)$.
	
	Contrary to the modified magnetic translations $\widehat\tau_{b_0,{\eta}}$ defined in \eqref{MagneticTranslation}, the translation operators $\tau_{\epsilon,\eta}$ do not form a unitary representation of the group $\Z^2$, but rather a projective one. Indeed
	\begin{equation} \label{eqn:PropMagTransl}
	\tau_{\epsilon,\eta}^* = \tau_{\epsilon,-\eta} \quad \text{and} \quad \tau_{\epsilon,\eta} \tau_{\epsilon,\eta'} = \eu^{\iu \epsilon q \phi(\eta',\eta)} \tau_{\epsilon,\eta+\eta'}, \qquad \eta,\eta' \in \Z^2.
	\end{equation}
\end{remark}

The main achievement of Proposition~\ref{prop1} is to reduce the original Hamiltonian $H_{b_0+\epsilon}$ to the product of a phase factor times a fibered operator, that is $\mathcal{T}_\epsilon$,  
whose fiber $h_{\kk,b_0+\epsilon}$ acts in a fiber space whose dimension $Q$ is independent of $\epsilon$ and which only depends on $b_0$ via $q$. This is crucial in order to control the perturbation induced by $\epsilon$, because the $\epsilon$-dependent fiber operators act in the \emph{same} space even as $\epsilon$ is varying. Even though the representation in Proposition~\ref{prop1} is valid for all values of $\epsilon$, it will be used for $\epsilon$ sufficiently small, for which the spectral properties of $H_{b_0}$ and $H_{b_0+\epsilon}$ are ``comparable'' (i.e., for which the spectral gap of the rational-flux Hamiltonian persists also at $\epsilon \ne 0$).

\section{Periodic setting: results for rational flux Hamiltonians} \label{sec:MR_Periodic}

Having established a clear reference model, we now abstract from the periodic setting of Hofstadter-like Hamiltonians satisfying a rational flux condition, and consider families of rank-$m$ orthogonal projections $\set{P(\kk)}_{\kk \in \R^d}$, $P(\kk) = P(\kk)^2 = P(\kk)^*$, acting on some Hilbert space $\Hi$, which are subject to the following conditions:
\begin{enumerate}
 \item the map $P \colon \R^d \to \mathcal{B}(\Hi)$, $\kk \mapsto P(\kk)$, is smooth (at least of class $C^1$);
 \item the map $P \colon \R^d \to \mathcal{B}(\Hi)$, $\kk \mapsto P(\kk)$, is $\Z^d$-periodic, that is, $P(\kk) = P(\kk + \mathbf{n})$ for all $\mathbf{n} \in \Z^d$.
\end{enumerate}
The rank $m$ corresponds to the number of occupied energy bands in physical applications. As discussed e.g.~in~\cite{CorneanHerbstNenciu16,MonacoPanatiPisanteTeufel18}, the same setting arises also from continuous models (described by a magnetic Schr\"{o}dinger operator as the Hamiltonian) of gapped periodic quantum systems subject to a magnetic field satisfying the rational flux property: we note that, in this case, some technical modifications are required to define the Bloch--Floquet representation, and one is led to use in this case the so-called (\emph{magnetic}) \emph{Bloch--Floquet--Zak transform} (see also \cite{FreundTeufel16}).

\begin{definition} \label{def:BlochFrame}
A \emph{Bloch vector} for the family of projections $\set{P(\kk)}_{\kk \in \R^d}$ is a map $\xi \colon \R^d \to \Hi$ such that 
\[ P(\kk) \xi(\kk) = \xi(\kk) \quad \text{for all } \kk \in \R^d. \]

A Bloch vector $\xi$ is called
\begin{enumerate}
 \item \label{item:Bloch_a} \emph{continuous} if the map $\xi \colon \R^d \to \Hi$ is continuous;
 \item \label{item:Bloch_b} \emph{periodic} if the map $\xi \colon \R^d \to \Hi$ is $\Z^d$-periodic, that is, $\xi(\kk) = \xi(\kk + \mathbf{n})$ for all $\mathbf{n} \in \Z^d$;
 \item \label{item:Bloch_c} \emph{normalized} if $\norm{\xi(\kk)} = 1$ for all $\kk \in \R^d$.
\end{enumerate}

A collection of $M$ Bloch vectors $\set{\xi_a}_{a=1}^{M}$ is said to be
\begin{enumerate}
 \item \emph{independent} (respectively \emph{orthonormal}) if the vectors $\set{\xi_a(\kk)}_{a=1}^{M} \subset \Hi$ are linearly independent (respectively orthonormal) for all $\kk \in \R^d$;
 \item a \emph{moving Parseval $M$-frame} (or $M$-frame in short) if $M \ge m$ and for every $\psi\in \Ran P(\kk)$ we have  
\begin{equation}\label{april1} 
 \psi=\sum_{a=1}^M \langle \xi_a(\kk),\psi \rangle \;\xi_a(\kk)\quad \text{or equivalently} \quad \norm{\psi}^2=\sum_{a=1}^M \left| \langle \xi_a(\kk),\psi \rangle \right|^2.
 \end{equation}
 If $M=m$, we call $\set{\xi_a}_{a=1}^{m}$ a \emph{Bloch basis}.
\end{enumerate}
\end{definition}

In general, all the above conditions on a collection of Bloch vectors compete against each other, and one has to give up some of them in order to enforce the others. As was recalled in the Introduction, this is well-known in differential geometry. Indeed, given a smooth, periodic family of projections, one can construct the associated \emph{Bloch bundle} $\mathcal{E} \to \T^d$ \cite{Panati07}, which is an Hermitian vector bundle over the (Brillouin) $d$-torus $\T^d = \R^d / \Z^d$, and Bloch vectors are nothing but \emph{sections} for this vector bundle. The \emph{topological obstruction} to construct sections of a vector bundle reflects in the impossibility to construct collections of Bloch vectors with the required properties. For example:
\begin{itemize}
 \item in general, a Bloch vector can be continuous but \emph{not} periodic, or viceversa periodic but \emph{not} continuous: in the latter case, one then speaks of \emph{local sections} of the associated Bloch bundle, defined in the patches where they are continuous;
 \item global (continuous and periodic) sections may exist, but they may \emph{vanish} in $\T^d$, thus violating the normalization condition for a Bloch vector;
 \item when $d \le 3$, the topological obstruction to construct a (possibly orthonormal) Bloch basis consisting of continuous, periodic Bloch vectors is encoded in the \emph{Chern numbers} \cite{AvisIsham79, Panati07, Monaco17}
 \begin{equation} \label{Chern}
 c_1(P)_{ij} = \frac{1}{2 \pi \iu} \int_{\T^2_{ij}} \di k_i \, \di k_j \, \Tr_\Hi \left( P(\kk) \, \left[ \partial_i P(\kk), \partial_j P(\kk) \right] \right) \quad \in \Z, \quad 1 \le i < j \le d,
 \end{equation}
 where $\T^2_{ij} \subset \T^d$ is the $2$-torus where the coordinates different from $k_i$ and $k_j$ are set equal to zero. Only when the Chern numbers vanish does a Bloch basis exist, in which case the Bloch bundle is \emph{trivial}, i.e.\ isomorphic to $\T^d \times \C^m$.
\end{itemize}

In the first part of this paper, we discuss the possibility of relaxing the condition to be a continuous, periodic, and orthonormal Bloch basis in two possible ways, by considering instead collections of $M$ Bloch vectors such that
\begin{enumerate}
 \item $M < m$, and the continuous, periodic Bloch vectors are still \emph{orthonormal};
 \item $M > m$, and the continuous, periodic Bloch vectors are still \emph{generating} (hence constitute an $M$-frame).
\end{enumerate}
In the present context of families of projections arising from gapped crystalline Hamiltonians, optimal \emph{existence} results on orthonormal sets and Parseval frames of Bloch vectors were first proved in \cite{Kuchment09,AucklyKuchment18} via general bundle-theoretic argument, as already mentioned in Section~\ref{sec:literature}. Here ``optimal'' refers to finding the optimal value $M$ in each of the two situations (the maximal $M$ in the first, and the minimal $M$ in the second). The results are summarized in the following
\begin{theorem}[\cite{AucklyKuchment18,Kuchment09}] \label{thm:Main}
Let $d \le 3$, and let $\set{P(\kk)}_{\kk \in \R^d}$ be a smooth, $\Z^d$-periodic family of orthogonal projections of rank $m$.
\begin{enumerate}[label=(\roman*),ref=\roman*]
 \item \label{Main_1} There exist at least $m-1$ independent Bloch vectors which are continuous and $\Z^d$-periodic.
 \item \label{Main_2} There exists a Parseval $(m+1)$-frame of continuous and $\Z^d$-periodic Bloch vectors (see \eqref{april1}).
 \item \label{Main_3} Assume furthermore that $c_1(P)_{ij} = 0 \in \Z$ for all $1 \le i < j \le d$, where $c_1(P)_{ij}$ is defined in \eqref{Chern}. Then, there exists an \emph{orthonormal Bloch basis} of continuous and $\Z^d$-periodic Bloch vectors.
\end{enumerate}
\end{theorem}

\begin{remark}
By standard arguments, which we reproduce in Appendix~\ref{app:Smoothing} for the reader's convenience, it is possible to improve the regularity of Bloch vectors if the family of projections is more regular: the only obstruction is to continuity. In other words, if for example the map $\kk \mapsto P(\kk)$ is smooth or analytic, then a continuous Bloch vector yields a smooth or real-analytic one by convolution with a sufficiently regular kernel. Moreover, one can always make sure that all the other properties (periodicity, orthogonality, \ldots) are preserved by this smoothing procedure.

As was already remarked in Section~\ref{sec:literature}, in the case of an analytic family of projections the techniques of \cite{AucklyKuchment18,Kuchment09} allow to show the existence of Bloch vectors which are analytic in the same analyticity strip. The explicit smoothing procedure mentioned above, instead, only gives a weaker real analyticity (i.e.\ analyticity of the Bloch vectors in a complex strip around the real $\kk$'s of \textit{a priori} smaller width than the one of the analyticity domain of the projections).
\end{remark}

Abstract results concerning the existence of such collections of Bloch vectors can be also found in the literature on vector bundles. For example:
\begin{enumerate}
 \item by \cite[Chap.~9, Thm.~1.2]{Husemoller94}, there exist $m-\ell_d$ continuous and periodic \emph{independent} sections of the Bloch bundle, where
 \footnote{We denote by $\lceil x \rceil$ the smallest integer $n$ such that $x \le n$.} 
 $\ell_d = \lceil (d-1)/2 \rceil$;
 \item by \cite[Chap.~8, Thm.~7.2]{Husemoller94}, there exists an \emph{$(m+r_d)$-frame} for the Bloch bundle, where $r_d = \lceil d/2 \rceil$.
\end{enumerate}
The second of the above statements can be rephrased by saying that there exists a \emph{trivial} vector bundle $\mathcal{F}$ of rank $m+r_d$ that contains $\mathcal{E}$ as a subbundle. Indeed, if $\set{\psi_a}_{a=1}^{m+r_d}$ is a moving basis for $\mathcal{F}$, then setting $\xi_a(\kk) := P(\kk) \, \psi_a(\kk)$, $a \in \set{ 1, \ldots, m+r_d}$, defines an $(m+r_d)$-frame for $\mathcal{E}$ (see also \cite{FreemanPooreWeiWyse14}). Notice that the above Theorem~\ref{thm:Main} for $d=3$ yields an optimal number ($M=m+1$) of vectors in a Parseval frame, which is actually smaller than the number $M=m+r_{d=3} = m+2$ predicted by the general, bundle-theoretic result quoted above \cite[Chap.~8, Thm.~7.2]{Husemoller94}.

This kind of results have a much broader range of applicability and hold for a large class of base manifolds (of which the base space of the Bloch bundle, namely the $d$-torus for $d \le 3$, is only a very specific case). However, their proofs rely on techniques from algebraic topology, specifically on homotopy and obstruction theory, which may not be particularly suited to numerical implementations, because for example they allow to construct the required objects only up to homotopies which are often difficult to describe analytically
\footnote{A noteworthy exception for the present discussion is provided by~\cite{GontierLevittSiraj-Dine18}, where certain homotopies in the unitary groups are computed to numerically construct Wannier functions in time-reversal symmetric topological insulators.}. 

Aiming at this type of applications in computational condensed matter physics, as already mentioned in the Introduction (see also Section~\ref{sec:Gabor} below), our first contribution in this direction is to provide an \emph{alternative, algorithmic proof} of Theorem~\ref{thm:Main}, which explicitly exhibits the optimal number of orthonormal (respectively generating) Bloch vectors via an algorithm, in a finite number of steps and working out all analytical details (mostly in the Appendices and in references therein). The algorithm we propose is sketched at the end of the next Section, and the details of our proof of Theorem~\ref{thm:Main} are presented in Sections~\ref{sec:c1=0} to~\ref{sec:thmMain(ii)}.

\subsection{Applications to Wannier functions} 

Concerning the specific case of Bloch bundles arising from condensed matter systems, the construction of (real-analytic) Bloch vectors translates to the construction of \emph{localized} (\emph{composite}) \emph{Wannier functions} for the occupied states of the magnetic Hamiltonian describing the crystal, by transforming the Bloch vectors back from the $\kk$-space representation to the position representation via the Bloch--Floquet transform~\eqref{eqn:WF} \cite{Kuchment16}. Our proof of the second part of Theorem~\ref{thm:Main} can then be rephrased as the possibility to \emph{algorithmically construct} Parseval frames for the spectral island onto $m$ gapped energy bands consisting of $m+1$ exponentially localized Wannier functions, together with their (magnetic) translates. Although, as mentioned above, the range of applicability of Theorem~\ref{thm:Main} on Bloch vectors includes also spectral projections of certain magnetic Schr\"{o}dinger operators, for simplicity, and in order to avoid too many technical conditions, we only formulate the result for Hofstadter-like Hamiltonians:
\begin{theorem} \label{thm:Main_bis}
Let $H_{b_0}$ be an Hofstadter-like Hamiltonian on $\ell^2(\Z^2) \otimes \ell^2(\Y)$ corresponding to a magnetic field $b_0 \in 2 \pi \mathbb{Q}$. Let $\Pi=\Pi_{b_0}$ be the spectral projection onto an isolated spectral island of $H_{b_0}$ consisting of $m$ energy bands, and let $\mathcal{U}\sub{mBF}\,\Pi\,\mathcal{U}\sub{mBF}^* =\int_{\T^2}^{\oplus} \di \kk \, P(\kk)$. Then:
\begin{enumerate}[label=(\roman*),ref=\roman*]
\item \label{Main_bis_1} there exists an \emph{exponentially localized Wannier Parseval frame} for the subspace $\Ran \Pi \subset \ell^2(\Z^2) \otimes \ell^2(\Y)$, i.e.~there exist $m+1$ exponentially localized vectors $w_a$, $a \in \set{1,\ldots,m+1}$, such that
\[	w = \sum_{\gamma\in \Z^2} \sum_{a=1}^{m+1} \scal{ \hat{\tau}_{b_0,\gamma} w_a}{w} \left(\hat{\tau}_{b_0,\gamma} w_a \right) \quad \text{for all } w \in \Ran \Pi; \]
 \item \label{Main_bis_2} if moreover $c_1(P) = 0 \in \Z$, where $c_1(P) \equiv c_1(P)_{12}$ is defined in \eqref{Chern}, then there exist $m$ exponentially localized vectors $w_a$, $a \in \set{1,\ldots,m}$, such that $\{\hat{\tau}_{b_0,\gamma} w_a\}_{a\in\set{1,\ldots,m}, \:\gamma \in \Z^2} $ is an \emph{orthonormal basis} of $\Ran \Pi \subset \ell^2(\Z^2) \otimes \ell^2(\Y)$.
\end{enumerate}
\end{theorem}

We stress again that, using our proof of Theorem~\ref{thm:Main}, the objects whose existence is claimed in Theorem~\ref{thm:Main_bis} can be constructed with a finite-step algorithm that in principle can be numerically implemented.

\subsection{Why are Parseval frames useful in solid state physics?} \label{sec:Gabor}

Inspired by \cite{Kuchment09}, we advocate the use of Parseval frames of localized Wannier functions as an efficient tool to derive tight-binding models for magnetic Hamiltonians, much in the same way as orthonormal bases are used in topologically unobstructed cases, e.g.~under a time-reversal symmetry assumption \cite{Panati07, MonacoPanati15}.

To substantiate this claim let us start by some general considerations, and recall the definition of a classical Parseval Gabor frame \cite{HanLarson00,Simon15}. For every pair $(\lambda, \gamma) \in \Z^d \times \Z^d = \Z^{2d}$ we consider the functions $\psi_{\lambda\gamma}(\x):=\eu^{2\pi\iu \lambda \cdot \x} g(\x-\gamma)$ where $g$ is a smooth function, compactly supported in $[-1,1]^d$ and such that $\sum_{\gamma\in \Z^d}|g(\x-\gamma)|^2=1$ for all $\x\in \R^d$.  It is well-known that the set $\set{\psi_{\lambda\gamma}}_{\lambda, \gamma \in \Z^d}$ forms an overcomplete Parseval frame in $L^2(\R^d)$, in the sense that any $f\in L^2(\R^d)$ can be written~as 
\[ f=\sum_{\lambda, \gamma\in \Z^d} \scal{\psi_{\lambda\gamma}}{f} \psi_{\lambda\gamma},\quad \text{with} \quad  \norm{f}^2 = \sum_{\lambda, \gamma\in \Z^d} \left| \scal{\psi_{\lambda\gamma}}{f}\right|^2. \]
Although a Parseval frame is not an orthonormal basis, one can represent any reasonable linear (pseudo-differential) operator $A$ on $L^2(\R^d)$ as an ``infinite double matrix'' acting in $\ell^2(\Z^{2d})$, where the matrix elements are given by $A(\lambda,\gamma; \lambda',\gamma') := \scal{\psi_{\lambda\gamma}}{A\psi_{\lambda'\gamma'}}$ \cite{Grochenig,FeichtingerStrohmer98}.

In applications one is typically interested in finding a generating set of vectors for the subspace which is the range of an orthogonal Fermi projection $\Pi$ onto an isolated group of $m$ bands of an Hamiltonian $H$ unitarily conjugated to a fibered operator $\int_{\T^d}^{\oplus} \di \kk \, h(\kk)$. Our proof of Theorem~\ref{thm:Main_bis} provides a way to construct the \emph{smallest} finite set of \emph{exponentially} localized functions $\set{w_a}_{1 \le a \le M}$ with $m\leq M$ such that 
\[ \Pi = \sum_{\gamma\in \Z^d} \sum_{a=1}^{M} \ket{T_\gamma \, w_a} \bra{T_\gamma \, w_a} \] 
where $\gamma \mapsto T_\gamma$ is a suitable representation of the translation group $\Z^d$ (e.g.\ $T_\gamma = \hat{\tau}_{b_0,\gamma}$ for rational-flux Hofstadter-like Hamiltonians in $d=2$). In particular, the existence of a Parseval frame allows to isometrically identify $\Ran \Pi$ with the space $\ell^2(\Z^d)\otimes \C^M$. The number $M$ is either $m$ or $m+1$, depending on the vanishing or not of the Chern numbers of the Bloch bundle associated to the fibers $\set{P(\kk)}_{\kk \in \T^d}$ of $\Pi$.

The analytic and periodic Bloch frame corresponding to the Parseval frame $\set{T_\gamma \, w_a}_{1 \le a \le M, \, \gamma \in \Z^d}$, see Theorem~\ref{thm:Main}, can be also used to construct an effective model to study, for example, the band structure of the fiber Bloch Hamiltonian $h(\kk)$ numerically through Fourier interpolation \cite{Yates_et_al07, MarzariEtAl12}. By the gap condition and a shift of the (Fermi) energy, we can assume that $h(\kk)$ has a gap at zero energy and hence has non-zero eigenvalues. Denote by $E_j(\kk)\ne 0$ the Bloch energy bands, labelled in increasing order, and by $\psi_j(\kk)$ the corresponding (normalized) Bloch eigenfunctions, $h(\kk) \psi_j(\kk) = E_j(\kk) \psi_j(\kk)$. Assume that the negative eigenvalues (below the spectral gap and the Fermi energy) are labelled by $j \in \set{1, \ldots, m}$. We have
\[ P(\kk) = \sum_{j=1}^{m} \ket{\psi_j(\kk)} \bra{\psi_j(\kk)}, \quad   h(\kk)P(\kk) = \sum_{j=1}^{m} E_j(\kk) \, \ket{\psi_j(\kk)}\bra{\psi_j(\kk)}. \]
The eigenvectors $\psi_j(\kk)$ are not necessarily smooth in $\kk$ even though $h(\kk) P(\kk)$ is smooth and periodic. Using our Parseval frame $\set{\xi_a(\kk)}_{1 \le a \le M}$ as in \eqref{april1}, we can introduce an $M \times M$ matrix $h\sub{eff}(\kk)$ acting on $\C^{M}$ and given by
\[ h\sub{eff}(\kk)_{a a'} := \scal{\xi_a(\kk)}{h(\kk) \xi_{a'}(\kk)}_{\mathcal{H}}, \quad 1\leq a,a'\leq M. \]
This matrix is both smooth and periodic; we show further that its non-zero spectrum coincides with the relevant Bloch eigenvalues $E_j(\kk)$. Define the vectors $\Psi_j(\kk)\in \C^{M}$, $j \in \set{1,\ldots,m}$, with components given by $(\Psi_j(\kk))_a = \scal{\xi_a(\kk)}{\psi_j(\kk)}_{\mathcal{H}}$, $a \in \set{1,\ldots,M}$. Then, by the Parseval property \eqref{april1},
\[
\scal{\Psi_j(\kk)}{\Psi_{j'}(\kk)}_{\C^M} = \sum_{a=1}^{M} \scal{\psi_j(\kk)}{\xi_a(\kk)}_{\mathcal{H}} \, \scal{\xi_a(\kk)}{\psi_{j'}(\kk)}_{\mathcal{H}} = \scal{\psi_j(\kk)}{\psi_{j'}(\kk)}_{\mathcal{H}} = \delta_{jj'}
\]
and furthermore, by definition
\[ h\sub{eff}(\kk) =\sum_{j=1}^{m} E_j(\kk) \ket{\Psi_j(\kk)} \bra{\Psi_j(\kk)}. \]
From the above we see that $h\sub{eff}(\kk)$ has $\Psi_j(\kk)$ as an eigenvector with corresponding eigenvalue~$E_j(\kk)$.

Even though $h\sub{eff}(\kk)$ has a (redundant) constant zero eigenvalue, no information about the non-zero spectrum is lost. In particular, the $m$ non-zero eigenvalues of the $M\times M$ matrix $h\sub{eff}(\kk)$, coinciding with the relevant Bloch bands, are periodic functions of $\kk$ and can be sampled at a few points $\kk$ in a mesh for $(-1/2,1/2)^d$. Interpolating these few points with Fourier multipliers allows to approximate the energy bands with great accuracy: this is guaranteed by the smoothness of the constructed Parseval frame $\set{\xi_a(\kk)}$, which implies a very fast decay of their Fourier coefficients (namely of the corresponding Wannier functions) and hence a fast convergence of their Fourier series, see for example \cite{Yates_et_al07,MarzariEtAl12} and references therein.

\section{Non-periodic setting: results for irrational flux Hamiltonians}
\label{sec:MR_irrational}

Once the construction of Parseval frames is established for periodic projections, it is a legitimate question to ask whether it is possible to extend this result to systems that are not periodic. Our second novel result goes in this direction. As was explained in Section~\ref{sec:Model}, one such situation is provided by Hofstadter-like Hamiltonians on 2-dimensional crystals subject to a magnetic field which has irrational flux through the fundamental cell, in units of the magnetic flux quantum. As soon as the rationality condition is not satisfied, the Bloch bundle construction fails. This is due to the fact that, despite the Hamiltonian is still commuting with the set of magnetic translations, they are not a unitary representation of the translation group $\Z^2$, but only a projective one. Therefore, since there is no $\kk$-space description, one is forced to build spanning sets of localized vectors for the Fermi projection onto an isolated spectral island directly in position-space. 

We approach the problem of an irrational magnetic flux perturbatively and set $b = b_0 + \epsilon$ with $b_0 q \in 2\pi \mathbb{Z}$ for $q \in \mathbb{N}$ and $0\le \epsilon \ll 1$.  We assume that the periodic Hamiltonian $\widetilde{\mathcal{H}}_{0}$ in~\eqref{apr1'}, which is unitarily equivalent to the Hofstadter-like Hamiltonian $H_{b_0}$, has an isolated spectral island consisting of $m$ bands which are associated to a Fermi projection $\widetilde{\mathcal{P}}_{0}$ unitarily equivalent to the fibered operator $\int_{\T^2}^{\oplus} \di \kk \, P_0(\kk)$. Notice that, in the periodic Hamiltonian $\widetilde{\mathcal{H}}_{0}$, the information about the magnetic field $b_0$ is encoded in the translation invariant matrix elements and the Peierls phase is absent, therefore $\widetilde{\mathcal{H}}_{0}$ is a sort of effective reference non-magnetic Hamiltonian. If $\epsilon$ is small enough, then $\widetilde{\mathcal{H}}_{\epsilon}$ will also have an isolated spectral island \cite{Cornean} associated to a Fermi projection $\widetilde{\mathcal{P}}^{(\epsilon)}$, with $\widetilde{\mathcal{P}}^{(\epsilon=0)}=\widetilde{\mathcal{P}}_{0}$; notice that the number of magnetic mini-bands may change. Note that, as it was explained in Section \ref{sec:Model}, both $\widetilde{\mathcal{H}}_{\epsilon}$ and $\widetilde{\mathcal{P}}_{\epsilon}$ commute with the unitary operator defined in \eqref{eqn:tau}.  Then our second main result is the following. 
\begin{theorem} \label{thm:Main2}
For $\eta \in \Z^2$, let $\tau_{\epsilon,\eta}$ be the unitary given defined in \eqref{eqn:tau}.
Then there exists $\epsilon_0>0$ such that for all $0\leq \epsilon\leq \epsilon_0$ the following hold:
\begin{enumerate}[label=(\roman*),ref=\roman*]
\item there exist $m+1$ exponentially localized vectors $\big\{w_a^{(\epsilon)}\big\}_{1 \le a \le m+1}$ such that 
\begin{equation} \label{apr7}
\widetilde{\mathcal{P}}^{(\epsilon)} := \sum_{\eta \in \Z^2} \sum_{a=1}^{m+1} \ket{\tau_{\epsilon,\eta}\, w_a^{(\epsilon)}} \bra{\tau_{\epsilon,\eta}\, w_a^{(\epsilon)}};
\end{equation}
 \item if moreover $c_1(P_{0}) = 0 \in \Z$, where $c_1(P_{0}) \equiv c_1(P_{0})_{12}$ is defined in \eqref{Chern}, then there exist $m$ exponentially localized vectors $\big\{w_a^{(\epsilon)}\big\}_{1 \le a \le m}$ such that $\big\{\tau_{\epsilon,\eta} \, w_a^{(\epsilon)}\big\}_{1 \le a \le m, \:\eta \in \Z^2}$ is an \emph{orthonormal basis} of $\Ran \widetilde{\mathcal{P}}^{(\epsilon)}$.
\end{enumerate}
\end{theorem}

In the same spirit of the proof of Theorem~\ref{thm:Main}, our argument for the above result provides a constructive algorithm consisting of finitely many steps which exhibits the required Wannier functions.

The above result can be rephrased in terms of the original Hofstadter-like Hamiltonian $H_b= H_{b_0+\epsilon}$. In order to do so, we refer to the notion, introduced in \cite{NenciuNenciu93,NenciuNenciu98} (see also \cite{CorneanHerbstNenciu16}), of a generalized Wannier basis or Parseval frame for Hamiltonians which do not commute with a unitary representation of the group $\Z^2$.

\begin{definition}
An \emph{exponentially localized generalized Wannier basis} (respectively \emph{Parseval frame}) for the projection $\Pi$ acting in $\ell^2(\Z^2)\otimes\C^Q$ is a couple $\left(\Gamma,\mathcal{W}\right)$, where $\Gamma$ is a discrete subset of $\R^2$, and $\mathcal{W}=\{\psi_{\gamma,a}\}_{\gamma \in \Gamma, \: 1\leq a\leq m(\gamma)<m^*}$, with $m^*>0$ and independent of $\gamma$, is an orthonormal basis (respectively Parseval frame) for the range of $\Pi$ such that
\begin{equation*}
 \sum_{\eta \in \Z^2} \sum_{\xx=1}^{Q} |\psi_{\gamma,a} (\eta,\xx)|^2 \, \eu^{\beta\|\eta-\gamma\|} \leq M, \quad a \in \set{1,\ldots,m(\gamma)},
\end{equation*}
for some positive constants $\beta,M>0$ uniform in $\gamma$.
\end{definition}

Consider now the projection $\Pi_{b}$ of the original Hofstadter-like Hamiltonian. As was explained in Section~\ref{sec:Model}, at $b = b_0$ the Hamiltonian $H_{b_0}$ is fibered by the magnetic Bloch--Floquet transform, $\mathcal{U}\sub{mBF} H_{b_0} \mathcal{U}\sub{mBF}^* = \int_{\T^2}^{\oplus} \di \kk \, h(\kk)$ (compare~\eqref{eqn:FiberedHb0}), and correspondingly $\mathcal{U}\sub{mBF} \Pi_{b_0} \mathcal{U}\sub{mBF}^* = \int_{\T^2}^{\oplus} \di \kk \, P(\kk)$. Then the following result easily follows from Theorem \ref{thm:Main2} and Proposition \ref{prop1}.
 
\begin{corollary} \label{CorGWB}
There exists $\epsilon_0>0$ such that for all $0\leq \epsilon\leq \epsilon_0$ the following hold:
\begin{enumerate}[label=(\roman*),ref=\roman*]
	\item there exists an exponentially localized generalized Wannier Parseval frame for the projection $\Pi_{b=b_0+\epsilon}$ that is given by the couple $	\big(\Z^2,\{U^*_b \tau_{\epsilon,\eta} \, w_a^{(\epsilon)}\}_{\eta \in \Z^2, \: 1 \le a \le m+1}\big)$ and satisfies 
	\begin{equation*}
	\Pi_{b}=  \sum_{\eta \in \Z^2} \sum_{a=1}^{m+1} \ket{U^*_b \tau_{\epsilon,\eta}\, w_a^{(\epsilon)}} \bra{U^*_b \tau_{\epsilon,\eta}\, w_a^{(\epsilon)}}.
	\end{equation*}
	\item if moreover $c_1(P) = 0 \in \Z$, where $c_1(P) \equiv c_1(P)_{12}$ is defined in \eqref{Chern}, then there exists an exponentially localized generalized Wannier \emph{basis} for the projection $\Pi_{b=b_0+\epsilon}$, given by	$\big(\Z^2,\{U^*_b \tau_{\epsilon,\eta} \, w_a^{(\epsilon)}\}_{\eta \in \Z^2, \: 1 \le a \le m}\big)$.
\end{enumerate}
\end{corollary}

Let us stress that the unitary $U_b$ consists just of multiplication by a local phase (compare~\eqref{iunie1}), hence it does not spoil the localization properties of the function on which it is applied.

The proofs of Theorem \ref{thm:Main2} and of its Corollary \ref{CorGWB}, crucially rely on Combes--Thomas estimates and on \emph{gauge covariant magnetic perturbation theory} for discrete magnetic Hamiltonians, which we briefly review in Appendix~\ref{app:KernelEstimates}. These techniques are available also for continuous magnetic Schr\"{o}dinger operators: see \cite{CorneanNenciu} for the Combes--Thomas estimates and \cite{Nenciu02,Cornean,CorneanMonacoMoscolari18} for magnetic perturbation theory. Thus, our proofs can be generalized to the continuous setting with only minor efforts (compare Remark~\ref{ExtensionToContinuous}). 

The generalized Wannier Parseval frame $\big\{\psi_{\eta,a} := U^*_b \tau_{\epsilon,\eta}\, w_a^{(\epsilon)}\big\}_{\eta \in \Z^2, 1 \le a \le m}$ provided by the above Corollary allows to construct effective Hamiltonians $h_{b, \text{eff}}(\eta,a; \eta',a') := \scal{\psi_{\eta,a}}{H_b \psi_{\eta',a'}}$ on $\ell^2(\Z^2) \otimes \C^m$, from which spectral properties of the restriction to the isolated spectral island of the original Hofstadter-like Hamiltonian $H_b$ can be investigated, compare Section~\ref{sec:Gabor} above.

Since we consider Corollary~\ref{CorGWB} as the most important and novel contribution of our paper, we briefly sketch here the steps of its proof, serving also as an outline for the expert reader of the algorithmic construction of exponentially localized (generalized) Wannier functions claimed above, at least in 2-dimensions.
\begin{description}[leftmargin=4em]
 \item[Step 1] First look at $b=b_0 \in 2\pi\mathbb{Q}$, or, said otherwise, at $\epsilon=0$. The projection $\Pi_0$ is unitarily equivalent, via the modified Bloch--Floquet transform $\mathcal{U}\sub{mBF}$, to an analytic and $\Z^2$-periodic family of rank-$m$ projections $\set{P(\kk)}_{\kk\in\R^2}$. Via a modified \emph{parallel transport} in the second direction, one can extend any orthonormal basis for $P(\mathbf{0})$ to a smooth, $\Z$-periodic orthonormal Bloch basis for $\set{P(0,k_2)}_{k_2\in\R}$. Again parallel transport in the first direction will lead to a smooth orthonormal Bloch basis $\set{\psi_a(\kk)}_{\kk\in\R^2}$ for $\set{P(\kk)}_{\kk\in\R^2}$, which fails however to be periodic in $k_1$:
 \[ \psi_b(k_1+1,k_2) = \sum_{a=1}^{m} \psi_a(k_1,k_2)\, \alpha(k_2)_{ab}, \quad b \in \set{1,\ldots,m}. \]
 The unitary matrix $\alpha(k_2)$ is called the \emph{matching matrix}.
 \item[Step 2] Via a \emph{two-step logarithm}, the matching matrix can be deformed continuously to a diagonal matrix having all $1$'s as the first $m-1$ diagonal entries and a $k_2$-dependent phase as the last entry. The latter can also be ``unwinded'', that is, made equal to $1$ for all $k_2 \in \R$, exactly when the Chern number of $P$ vanishes. Deforming the matching matrix allows in turn to modify the $\psi_a$'s to a new orthonormal Bloch basis where $m-1$ vectors are also $\Z^2$-periodic, while the last one picks up a (topological) phase when looping in the first direction over the Brillouin torus. A smoothing procedure further allows to choose this Bloch basis as a regular function of $\kk$.
 \item[Step 3] Define $P_1(\kk)$ and $P_2(\kk)$ to be the subprojections of $P(\kk)$ onto the space spanned by the first $m-1$ Bloch vectors constructed before and onto the orthogonal complement of their span, respectively. Then $\set{P_2(\kk)}_{\kk\in\R^2}$ is a smooth and $\Z^2$-periodic family of rank-1 projections. We \emph{double the space dimension}, and consider the projection $P_2(\kk) \oplus (C P_2(\kk) C^{-1})$, where $C$ is a complex conjugation operator. This family of projections is topologically trivial and, by what was explained in Step 1, it admits a smooth and $\Z^2$-periodic Bloch basis consisting of two vectors. By projecting these two Bloch vectors back to the original space, we obtain a Parseval $2$-frame for the rank-1 projection $P_2(\kk)$, and consequently also a Parseval $(m+1)$-frame for $P(\kk)$, having all the desired properties.
 \item[Step 4] When applied to a rational-flux Hofstadter-like Hamiltonian $H_{b_0}$, the previous Steps produce an exponentially localized Wannier Parseval frame $\set{\hat{\tau}_{b_0,\gamma} w_a}_{1\le a \le m, \gamma \in \Z^2}$ for $\Pi_{b_0}$ by using the inverse magnetic Bloch--Floquet transform. We now perturb around $b_0$, passing to $b=b_0+\epsilon$. Recall from Proposition~\ref{prop1} that $H_{b_0+\epsilon} = U_{b_0+\epsilon}^* \, \widetilde{\mathcal{H}}_{\epsilon} \, U_{b_0+\epsilon}$, where the matrix elements of $\widetilde{\mathcal{H}}_{\epsilon}$ have the form \eqref{apr1'}. Inspired by gauge covariant magnetic perturbation theory \cite{CorneanNenciu98,Nenciu02,Cornean}, we consider the auxiliary Hamiltonian $\mathcal{H}_{\epsilon}$ defined through its matrix elements as in \eqref{apr1'}, but with the hopping $\mathcal{T}_{\epsilon}$ replaced by $\mathcal{T}_{0}$. We prove that the spectra of $\widetilde{\mathcal{H}}_{\epsilon}$ and of $\mathcal{H}_{\epsilon}$ are close (in the Hausdorff distance), and hence in particular that to every spectral projection $\widetilde{\mathcal{P}}_{\epsilon}$ onto a gapped spectral island of $\widetilde{\mathcal{H}}_{\epsilon}$ there corresponds a spectral projection $\mathcal{P}_{\epsilon}$ of $\mathcal{H}_{\epsilon}$; even more, for $|\epsilon|$ small enough the two projections are unitarily conjugated via a \emph{Kato--Nagy unitary} $K_\epsilon$. Thus a Wannier Parseval frame can be constructed for $\Pi_b$ if and only if it can be constructed for $\mathcal{P}_{\epsilon}$, since the two are unitarily conjugated via $U_{b_0+\epsilon} K_\epsilon$: the decay properties of the matrix elements of the latter unitary imply that localization is preserved under this unitary map.
 \item[Step 5] The projection $\mathcal{P}_{\epsilon}$ enjoys a number of properties, which we summarize in the Definition~\ref{dfn:Fermi-like} of a \emph{Fermi-like magnetic projection}. In particular, it is $\epsilon$-close to the operator whose matrix elements are equal to the ones of $\mathcal{P}_0$ multiplied by the $\epsilon$-dependent Peierls magnetic phase $\eu^{\iu \epsilon \phi(\cdot,\cdot)}$. The latter is not a projection anymore (it squares to itself only up to errors of order $\epsilon$), but it is much better-behaved as a function of $\epsilon$. Exploiting gauge covariant magnetic perturbation theory coupled with the procedure in Step 3, we find localized vectors close to the magnetic translates via $\tau_{\epsilon,\gamma}$ of the ones constructed previously at $\epsilon=0$, which give the required Parseval frame for $\mathcal{P}_{\epsilon}$ (and hence for $\Pi_b$ by Step 4). See the discussion after Proposition \ref{thm:Main3} for a more detailed description of this procedure.
\end{description}


\section{Proof of Theorem~\ref{thm:Main}\eqref{Main_3}: the topologically trivial case} \label{sec:c1=0}

We begin by proving Theorem~\ref{thm:Main}\eqref{Main_3} since elements of this proof will be essential for the other two parts of Theorem~\ref{thm:Main}. Thus we assume throughout this Section that $\set{P(\kk)}_{\kk \in \R^d}$, $d \le 3$, is a smooth and $\Z^d$-periodic family of rank-$m$ projections on the Hilbert space $\Hi$ with vanishing Chern numbers. We will construct an orthonormal Bloch basis (so, a $m$-tuple of orthogonal Bloch vectors) which is continuous and $\Z^d$-periodic. To stress that our proofs are algorithmic and explicit in nature, we use the phrase ``one can construct\ldots'' in many of the following statements.

\subsection{The 1D case}

We start from the case $d=1$. Notice that any $1$-dimensional family of projections $\set{P(k)}_{k \in \R}$ is topologically trivial, that is, it has vanishing Chern numbers (as there are no non-zero differential $2$-forms on the circle $\T$). 

Let $T(k,0)$ denote the parallel transport unitary along the segment from the point $0$ to the point $k$ associated to $\set{P(k)}_{k \in \R}$ (see Appendix~\ref{app:Parallel} for more details). At $k=1$, write $T(1,0) = \eu^{\iu  M}$, where $M = M^* \in \mathcal{B}(\Hi)$ is self-adjoint. 

Pick an orthonormal basis $\set{\xi_a(0)}_{a=1}^{m}$ in $\Ran P(0) \simeq \C^m \subset \Hi$, and define for $a \in \set{1, \ldots, m}$ and $k \in \R$
\[ \xi_a(k) := W(k) \, \xi_a(0), \quad W(k) := T(k,0) \eu^{- \iu k M}. \]
Then $\set{\xi_a}_{a=1}^{m}$ gives a continuous, $\Z^2$-periodic, and orthonormal Bloch basis for the $1$-dimensional family of projections $\set{P(k)}_{k \in \R}$ (compare \cite{CorneanHerbstNenciu16, CorneanMonacoTeufel17}). This proves Theorem~\ref{thm:Main} in $d=1$ (where the only non-trivial statement is part~\eqref{Main_3}).

\subsection{The induction argument in the dimension} \label{sec:induction} 

Consider a smooth and periodic family of projections $\set{P(k_1, \kk)}_{(k_1,\kk) \in \R^d}$, and let $D := d-1$. Assume that the $D$-dimensional restriction $\set{P(0,\kk)}_{\kk \in \R^{D}}$ admits a continuous and $\Z^{D}$-periodic orthonormal Bloch basis $\set{\xi_a(0,\cdot)}_{a=1}^{m}$. Consider now the parallel transport unitary $T_{\kk}(k_1,0)$ along the straight line from the point $(0,\kk)$ to the point $(k_1,\kk)$. At $k_1=1$, denote $\mathcal{T}(\kk) := T_{\kk}(1,0)$. Define
\begin{equation} \label{psi}
\psi_a(k_1,\kk) := T_{\kk}(k_1,0) \, \xi_a(0,\kk), \quad a \in \set{1, \ldots, m}, \: (k_1,\kk) \in \R^d.
\end{equation}
The above defines a collection of $m$ Bloch vectors for $\set{P(k_1,\kk)}_{(k_1,\kk) \in \R^d}$ which are continuous, orthonormal, and $\Z^{D}$-periodic in the variable $\kk$, but in general fail to be $\Z$-periodic in the variable $k_1$. Indeed, one can check that
\begin{equation} \label{alpha}
\psi_b(k_1+1, \kk) = \sum_{a=1}^{m} \psi_a(k_1, \kk) \, \alpha(\kk)_{ab}, \quad \text{where} \quad \alpha(\kk)_{ab} := \scal{\xi_a(0,\kk)}{\mathcal{T}(\kk) \, \xi_b(0,\kk)}
\end{equation}
(compare \cite[Eqn.s~(3.4) and (3.5)]{CorneanMonacoTeufel17}). The family $\set{\alpha(\kk)}_{\kk \in \R^{D}}$ defined above is a continuous and $\Z^{D}$-periodic family of $m \times m$ unitary matrices.

The possibility of ``rotating'' $\alpha(\kk)$ to the identity entails thus the construction of a Bloch basis which is also periodic in $k_1$. Formally, we have the following statement (compare also \cite[Thm.s~2.4 and 2.6]{CorneanMonaco17}).

\begin{proposition} \label{thm:alpha}
For the continuous and periodic family $\set{\alpha(\kk)}_{\kk \in \R^D}$ defined in \eqref{alpha}, the following are equivalent:
\begin{enumerate}
 \item  \label{alpha_1} the family is \emph{null-homotopic}, namely there exists a collection of continuous and $\Z^D$-periodic family of unitary matrices $\set{\alpha_t(\kk)}_{\kk \in \R^D}$, depending continuously on $t \in [0,1]$, and such that $\alpha_{t=0}(\kk) \equiv \Id$ while $\alpha_{t=1}(\kk) = \alpha(\kk)$ for all $\kk \in \R$;
 \item \label{alpha_2} assuming $D \le 2$, we have $\deg_j(\det \alpha)=0$ for all  $j \in \set {1, \ldots, D}$. In the smooth case, this is the same as:
 \begin{equation} \label{deg}
 \deg_j(\det \alpha) = \frac{1}{2 \pi \iu} \int_{0}^{1} \di k_j \, \tr_{\C^m} \left( \alpha(\kk)^* \frac{\partial \alpha}{\partial k_j}(\kk) \right) = 0 \quad \text{for all } j \in \set {1, \ldots, D};
 \end{equation}
 \item \label{alpha_3} the family admits a continuous and $\Z^D$-periodic \emph{$N$-step logarithm}, namely there exist $N$ continuous and $\Z^D$-periodic families of self-adjoint matrices $\set{h_i(\kk)}_{\kk \in \R^D}$, $i \in \set{1, \ldots, N}$, such that
 \begin{equation} \label{multilog}
 \alpha(\kk) = \eu^{\iu  h_1(\kk)} \cdots \eu^{\iu  h_N(\kk)}, \quad \kk \in \R^D;
 \end{equation}
 \item  \label{alpha_4} there exists a continuous family of unitary matrices $\set{\beta(k_1, \kk)}_{(k_1, \kk) \in \R^d}$, $d = D+1$, which is $\Z^D$-periodic in $\kk$, with $\beta(0,\kk) \equiv \Id$ for all $\kk \in \R^D$, and such that
 \[ \alpha(\kk) = \beta(k_1, \kk) \, \beta(k_1 + 1, \kk)^{-1}, \quad (k_1, \kk) \in \R^d; \]
 \item \label{alpha_5} there exists a continuous and $\Z^d$-periodic Bloch basis $\set{\xi_a}_{a=1}^{m}$ for $\set{P(\kk)}_{\kk \in \R^d}$.
\end{enumerate}
\end{proposition}
\begin{proof}
\noindent \fbox{\ref{alpha_1} $\Longleftrightarrow$ \ref{alpha_2}.} The integer $\deg_j(\det \alpha)$ defined in \eqref{deg} computes the winding number of the continuous and periodic function $k_j \mapsto \det \alpha(\cdots,k_j,\cdots) \colon \R \to U(1)$, $j \in \set{1, \ldots, D}$. It is a well-known fact in topology that $\pi_1(U(m)) \simeq \pi_1(U(1)) \simeq \Z$, with the first isomorphism implemented by the map $[\alpha] \mapsto [\det \alpha]$ and the second one implemented by the map $[\varphi] \mapsto \deg(\varphi) := (2 \pi \iu)^{-1} \int_{0}^{1} \varphi^{-1} \, \di \varphi$. It can be then argued that these winding numbers constitute complete homotopy invariants for continuous, periodic maps $\alpha \colon \R^D \to U(m)$ when $D \le 2$ (see {\sl e.g}~\cite[App.~A]{MonacoTauber17}).

\noindent \fbox{\ref{alpha_1} $\Longleftrightarrow$ \ref{alpha_3}.} Let $\set{\alpha_t(\kk)}_{\kk \in \R^D}$ be an homotopy between $\Id$ and $\alpha$, as in the statement. Since $[0,1]$ is a compact interval and $\alpha_t$ is $\Z^D$-periodic, by uniform continuity there exists $\delta > 0$ such that
\begin{equation} \label{close}
\sup_{\kk \in \R^D} \norm{\alpha_s(\kk) - \alpha_t(\kk)} < 2 \quad \text{whenever} \quad |s-t|< \delta. 
\end{equation}
Let $N \in \N$ be such that $1/N < \delta$. Then in particular
\[ \sup_{\kk \in \R^D} \norm{\alpha_{1/N}(\kk) - \Id} < 2 \]
so that the Cayley transform (see Appendix~\ref{app:Cayley}) provides a ``good'' logarithm for $\alpha_{1/N}(\kk)$, i.e. $\alpha_{1/N}(\kk) = \eu^{\iu  h_N(\kk)}$, with $h_N(\kk) = h_N(\kk)^*$ continuous and $\Z^D$-periodic.

Using again \eqref{close} we have that
\[ \sup_{\kk \in \R^D} \norm{\alpha_{2/N}(\kk) \, \eu^{-\iu  h_N(\kk)} - \Id} = \sup_{\kk \in \R^D} \norm{\alpha_{2/N}(\kk) - \alpha_{1/N}(\kk)} < 2 \]
so that by the same argument 
\[ \alpha_{2/N}(\kk) \, \eu^{-\iu  h_N(\kk)} = \eu^{\iu  h_{N-1}(\kk)}, \quad \text{or} \quad \alpha_{2/N}(\kk) = \eu^{\iu  h_{N-1}(\kk)} \, \eu^{\iu  h_N(\kk)}. \]
Repeating the same line of reasoning $N$ times, we end up exactly with \eqref{multilog}.

Conversely, if $\alpha(\kk)$ is as in \eqref{multilog}, then
\[ \alpha_t(\kk) := \eu^{\iu  \, t h_1(\kk)} \ldots \eu^{\iu  \, t h_N(\kk)}, \quad t \in [0,1], \: \kk \in \R^D, \]
defines the required homotopy between $\alpha_0(\kk) \equiv \Id$ and $\alpha_1(\kk) = \alpha(\kk)$.

\noindent \fbox{\ref{alpha_1} $\Longleftrightarrow$ \ref{alpha_4}.} Let $\set{\alpha_t(\kk)}_{\kk \in \R^D}$ be an homotopy between $\Id$ and $\alpha$. We set
\[ \beta(k_1,\kk) := \alpha_{k_1}(\kk)^{-1}, \quad k_1 \in [0,1],\;\kk \in \R^D, \]
and extend this definition to positive $k_1 > 0$ via
\[ \beta(k_1+1,\kk) := \alpha(\kk)^{-1} \, \beta(k_1,\kk) \]
and to negative $k_1 < 0$ via
\[ \beta(k_1,\kk) := \alpha(\kk) \, \beta(k_1+1,\kk). \]
We just need to show that this definition yields a continuous function of $k_1$. We have $\beta(0^+,\kk)=\Id$ and $\beta(1^-,\kk)=\alpha(\kk)^{-1}$ by definition. Let $\epsilon > 0$. If $k_1=-\epsilon$ is negative but close to zero, we have due to the definition
\[ \beta(-\epsilon,\kk) = \alpha(\kk) \, \beta(1-\epsilon,\kk) \to \alpha(\kk) \, \beta(1^-,\kk) = \Id \quad \text{as } \epsilon \to 0. \]
Hence $\beta$ is continuous at $k_1=0$. At $k_1=1$ we have instead
\[ \beta(1+\epsilon,\kk) = \alpha(\kk)^{-1} \, \beta(\epsilon,\kk) \to \alpha(\kk)^{-1} \, \beta(0^+,\kk) = \alpha(\kk)^{-1} \quad \text{as } \epsilon \to 0 \]
and $\beta$ is also continuous there. In a similar way one can prove continuity at every integer, thus on $\R$.

Conversely, if $\set{\beta(k_1, \kk)}_{(k_1,\kk) \in \R^2}$ is as in the statement, then the required homotopy $\alpha_t$ between $\Id$ and $\alpha$ is provided by setting
\[ \alpha_t(\kk) := \beta(-t/2,\kk) \, \beta(t/2,\kk)^{-1}, \quad t \in [0,1], \: \kk \in \R^D. \]

\noindent \fbox{\ref{alpha_4} $\Longleftrightarrow$ \ref{alpha_5}.} It suffices to set
\[ \xi_a(k_1,\kk) := \sum_{b=1}^{m} \psi_b(k_1,\kk) \, \beta(k_1,\kk)_{ba}, \quad a \in \set{1, \ldots, m}, \]
or equivalently
\[ \beta(k_1,\kk)_{ba} := \scal{\psi_b(k_1,\kk)}{\xi_a(k_1,\kk)}, \quad a,b \in \set{1, \ldots, m}, \]
for $\set{\psi_b}_{b=1}^{m}$ as in \eqref{psi} and $(k_1,\kk) \in \R^d$.
\end{proof}

To turn the above proof into a constructive argument, we need to construct the ``good'' logarithms in \eqref{multilog}.

\begin{proposition} \label{thm:2-step}
For $D \le 2$, let $\set{\alpha(\kk)}_{\kk \in \R^D}$ be a continuous and $\Z^D$-periodic family of unitary matrices. Assume that $\alpha$ is null-homotopic. Then it is possible to \emph{construct} a two-step ``good'' logarithm for $\alpha$, i.e.\ $N=2$ in Proposition~\ref{thm:alpha}\ref{alpha_3}.
\end{proposition}
\begin{proof}
\noindent \textsl{Step 1 : the generic form.} We first need to know that one can construct a sequence of continuous, $\Z^D$-periodic families of unitary matrices $\set{\alpha_n(\kk)}_{\kk \in \R^D}$, $n \in \N$, such that
\begin{itemize}
 \item $\sup_{\kk \in \R^D} \norm{\alpha_n(\kk) - \alpha(\kk)} \to 0$ as $n \to \infty$, and
 \item the spectrum of $\alpha_n(\kk)$ is completely non-degenerate for all $n \in \N$ and $\kk \in \R^D$.
\end{itemize}

The proof of this fact is rather technical, and is deferred to Appendix~\ref{app:NonDegenerate}. In the following, we denote $\alpha'(\kk) := \alpha_n(\kk)$ where $n \in \N$ is large enough so that
\[ \sup_{\kk \in \R} \norm{\alpha'(\kk) - \alpha(\kk)} < 2. \]

\noindent \textsl{Step 2 : $\alpha'$ is homotopic to $\alpha$.} Since
\[ \sup_{\kk \in \R^D} \norm{\alpha'(\kk) \, \alpha(\kk)^{-1} - \Id} = \sup_{\kk \in \R^D} \norm{\alpha'(\kk) - \alpha(\kk)} < 2 \]
we have that $-1$ always lies in the resolvent set of $\alpha'(\kk) \, \alpha(\kk)^{-1}$, which then admits a continuous and $\Z^D$-periodic logarithm defined via the Cayley transform:
\begin{equation} \label{a'a-1}
\alpha'(\kk) \, \alpha(\kk)^{-1} = \eu^{\iu  h''(\kk)}, \quad h''(\kk)^* = h''(\kk) = h''(\kk+\mathbf{n}) \text{ for } \mathbf{n} \in \Z^D.
\end{equation}
Therefore
\[ \alpha_t(\kk) := \alpha'(\kk) \, \eu^{\iu  \, t h''(\kk)}, \quad t \in [0,1], \: \kk \in \R^D, \]
gives a continuous homotopy between $\alpha_0(\kk) = \alpha'(\kk)$ and $\alpha_1(\kk) = \alpha(\kk)$. As a consequence, we have that $\alpha'$ is null-homotopic, since $\alpha$ is by assumption.

\noindent \textsl{Step 3 : a logarithm for $\alpha'$.} Denote by $\set{\lambda_1(\kk), \ldots, \lambda_m(\kk)}$ a continuous labelling of the periodic, non-degenerate eigenvalues of $\alpha'(\kk)$. 

If $m=1$, then $\alpha'(\kk) \equiv \det (\alpha'(\kk)) \equiv \lambda_1(\kk)$ cannot wind around the circle, due to the hypothesis that $\alpha'$ is null-homotopic. This implies that one can choose a continuous \emph{and periodic} argument for $\lambda_1$, namely $\lambda_1(\kk) = \eu^{\iu  \phi_1(\kk)}$ with $\phi_1 \colon \R^D \to \R$ continuous and $\Z^D$-periodic (compare e.g.\ \cite[Lemma~2.13]{CorneanHerbstNenciu16}).

If $m \ge 2$, then the same is true for each of the eigenvalues $\lambda_j(\kk)$, $j \in \set{1, \ldots, m}$. Indeed, let $\phi_j \colon \R^D \to \R$ be a continuous argument of the eigenvalue $\lambda_j$. The function $\phi_j$ will satisfy
\[ \phi_j(\kk + \mathbf{e}_l) = \phi_j(\kk) + 2 \pi n_j^{(l)}, \quad l \in \set{1,\ldots,D}, \quad n_j^{(l)} \in \Z, \]
where $\mathbf{e}_l = (0, \ldots, 1, \ldots, 0)$ is the $l$-th vector in the standard basis of $\R^D$ and the integer $n_j^{(l)}$ is the winding number of the periodic function $\R \to U(1)$, $k_l \mapsto \lambda_j(\cdots,k_l,\cdots)$. Fix $l \in \set{1, \ldots, D}$, and assume that there exist $i, j \in \set{1, \ldots, m}$ for which $n_i^{(l)} \ne n_j^{(l)}$. Define $\phi(\kk) := \phi_j(\kk) - \phi_i(\kk)$; then
\[ \phi(\kk+\mathbf{e}_j) = \phi(\kk) + 2 \pi \left(n_j^{(l)} - n_i^{(l)}\right). \]
Since $n_j^{(l)} - n_i^{(l)} \ne 0$, the periodic function $\lambda(\kk) := \eu^{\iu  \phi(\kk)}$ winds around the circle $U(1)$ at least once as a function of the $l$-th component, and in particular covers the whole circle. So there must exist $\kk_0 \in \R^D$ such that $\lambda(\kk_0) = 1$, or equivalently $\lambda_i(\kk_0) = \eu^{\iu  \phi_i(\kk_0)} = \eu^{\iu  \phi_j(\kk_0)} = \lambda_j(\kk_0)$, in contradiction with the non-degeneracy of the eigenvalues of $\alpha'(\kk)$.

We deduce then that $n_i^{(l)} = n_j^{(l)} \equiv n^{(l)}$ for all $i,j \in \set{1, \ldots, m}$. Set now $\det(\alpha'(\kk)) = \eu^{\iu  \Phi(\kk)}$ for $\Phi(\kk) = \phi_1(\kk) + \cdots + \phi_m(\kk)$. Then the equality
\[ \Phi(\kk+\mathbf{e}_l) = \Phi(\kk) + 2 \pi \sum_{j=1}^{m} n_j^{(l)} = \Phi(\kk) + 2 \pi m n^{(l)} \]
shows that necessarily $n^{(l)}=0$ for all $l \in \set{1,\ldots,D}$, as otherwise the determinant of $\alpha'$ would wind around the circle contrary to the hypothesis of null-homotopy of $\alpha'$.

Finally, denote by $0<g\leq 2$ the minimal distance between any two eigenvalues of $\alpha'(\kk)$, and define the continuous and periodic function $\rho(\kk) := \phi_1(\kk) + g/ 100$. Then $\eu^{\iu  \, \rho(\kk)}$ lies in the resolvent set of $\alpha'(\kk)$ for all $\kk \in \R$. As a consequence, $-1$ is always in the resolvent set of the continuous and periodic family of unitary matrices $\widetilde{\alpha}(\kk) := \eu^{-\iu  \, (\rho(\kk)+\pi)} \, \alpha'(\kk)$, which then admits a continuous and periodic logarithm via the Cayley transform: $\widetilde{\alpha}(\kk) = \eu^{\iu  \, \widetilde{h}(\kk)}$. We conclude that
\begin{equation} \label{alpha'}
\alpha'(\kk) = \eu^{\iu  h'(\kk)} \quad \text{with} \quad h'(\kk) := \widetilde{h}(\kk) + (\rho(\kk)+\pi) \Id.
\end{equation}
The family of self-adjoint matrices $\set{h'(\kk)}_{\kk \in \R^D}$ is still continuous and periodic by definition.

\noindent \textsl{Step 4 : a two-step logarithm for $\alpha$.} In view of \eqref{a'a-1} and \eqref{alpha'} we have $\eu^{\iu  h'(\kk)} \, \alpha(\kk)^{-1} = \eu^{\iu  h''(\kk)}$ for continuous and periodic families of self-adjoint matrices $\set{h'(\kk)}_{\kk \in \R^D}$ and $\set{h''(\kk)}_{\kk \in \R^D}$. This can be rewritten as $\alpha(\kk) = \eu^{-\iu  h''(\kk)} \, \eu^{\iu  h'(\kk)}$, which is \eqref{multilog} for $N=2$.
\end{proof}

\subsection{\texorpdfstring{The link between the topology of $\alpha$ and that of $P$}{The link between the topology of alpha and that of P}} 

We now come back to Theorem~\ref{thm:Main}\eqref{Main_3}. First  we consider the case $d=2$ (so that $D=d-1=1$). We have constructed in \eqref{alpha} a continuous and $\Z$-periodic family of unitary matrices $\set{\alpha(k_2)}_{k_2 \in \R}$, starting from a smooth, periodic family of projections $\set{P(k_1,k_2)}_{(k_1,k_2) \in \R^2}$ and an orthonormal Bloch basis for the restriction $\set{P(0,k_2)}_{k_2 \in \R}$. The next result links the topology of $\alpha$ with the one of $P$.

\begin{proposition} \label{thm:alphaChern}
Let $\set{\alpha(k_2)}_{k_2 \in \R}$ and $\set{P(\kk)}_{\kk \in \R^2}$ be as above. Then
\[ \deg(\det \alpha) = c_1(P). \]
\end{proposition}
\begin{proof}
The equality in the statement follows at once from the following chain of equalities:
\[ \tr_{\C^m} \left( \alpha(k_2)^* \partial_{k_2} \alpha(k_2) \right) = \Tr_{\Hi} \left( P(0,k_2) \, \mathcal{T}(k_2)^* \partial_{k_2} \mathcal{T}(k_2) \right) = \int_{0}^{1} \di k_1 \, \Tr_{\Hi} \left( P(\kk) \, \left[ \partial_1 P(\kk), \partial_2 P(\kk) \right] \right). \]
Their proof can be found in Appendix~\ref{app:Parallel} (compare \cite[Sec.~6.3]{CorneanMonacoTeufel17}).
\end{proof}

We are finally able to conclude the proof of Theorem~\ref{thm:Main}\eqref{Main_3}.

\begin{proof}[Proof of Theorem~\ref{thm:Main}\eqref{Main_3}, $d=2$]
Given our initial hypothesis that $c_1(P)=0$, the combination of Propositions~\ref{thm:alpha} and~\ref{thm:alphaChern} gives that $\alpha$ is null-homotopic, and hence admits a two-step logarithm which can be constructed via Proposition~\ref{thm:2-step}. This construction then yields the desired continuous and periodic Bloch basis, again via Proposition~\ref{thm:alpha}.
\end{proof}

\begin{proof}[Proof of Theorem~\ref{thm:Main}\eqref{Main_3}, $d=3$]
Let $\set{P(k_1,k_2,k_3)}_{(k_1,k_2,k_3) \in \R^3}$ be a smooth and periodic family of projections. Under the assumption that $c_1(P)_{23}=0$, the $2$-dimensional result we just proved provides an orthonormal Bloch basis for the restriction $\set{P(0,k_2,k_3)}_{(k_2,k_3) \in \R^2}$, which can be parallel-transported to $\set{k_1=1}$ and hence defines $\set{\alpha(k_2,k_3)}_{(k_2,k_3) \in \R^2}$, as in \eqref{alpha}. We now apply Proposition~\ref{thm:alphaChern}, to the $2$-dimensional restrictions $\set{P(k_1,0,k_3)}_{(k_1,k_3) \in \R^2}$ and $\set{P(k_1,k_2,0)}_{(k_1,k_2) \in \R^2}$ instead, and obtain that
\begin{equation} \label{alphaChern_ij}
\deg_2(\det \alpha) = \deg(\det \alpha(\cdot, 0)) = c_1(P)_{12} = 0, \quad \deg_3(\det \alpha) = \deg(\det \alpha(0, \cdot)) = c_1(P)_{13} = 0
\end{equation}
(compare Appendix~\ref{app:Parallel}). Again by Proposition~\ref{thm:alpha} the family $\alpha$ is then null-homotopic, and one can construct its two-step logarithm via Proposition~\ref{thm:2-step}. Proposition~\ref{thm:alpha} illustrates how to produce the required continuous and $\Z^3$-periodic Bloch basis.
\end{proof}


\section{Proof of Theorem~\ref{thm:Main}\eqref{Main_1}: maximal number of orthonormal Bloch vectors} \label{sec:thmMain(i)}

We come to the proof of Theorem~\ref{thm:Main}\eqref{Main_1}, concerning the existence of $m-1$ orthonormal Bloch vectors for a smooth and $\Z^d$-periodic family of projections $\set{P(\kk)}_{\kk \in \R^d}$ with $2 \le d \le 3$. As usual, we have denoted by $m$ the rank of $P(\kk)$.

\subsection{Pseudo-periodic families of matrices} 

Before giving the proof of Theorem~\ref{thm:Main}\eqref{Main_1}, we need some generalizations of the results in Section~\ref{sec:induction}. 

\begin{definition} \label{def:gamma}
Let $\set{\gamma(k_3)}_{k_3 \in \R}$ be a continuous and $\Z$-periodic family of unitary matrices. We say that a continuous family of matrices $\set{\mu(k_2,k_3)}_{(k_2, k_3) \in \R^2}$ is \emph{$\gamma$-periodic} if it satisfies the following conditions:
\[ \mu(k_2+1,k_3) = \gamma(k_3) \, \mu(k_2, k_3) \, \gamma(k_3)^{-1}, \quad \mu(k_2, k_3+1) = \mu(k_2,k_3), \quad (k_2, k_3) \in \R^2. \]

We say that two continuous and $\gamma$-periodic families $\set{\mu_0(\kk)}_{\kk \in \R^2}$ and $\set{\mu_1(\kk)}_{\kk \in \R^2}$ are \emph{$\gamma$-homotopic} if there exists a collection of continuous and $\gamma$-periodic families $\set{\mu_t(\kk)}_{\kk \in \R^2}$, depending continuously on $t \in [0,1]$, such that $\mu_{t=0}(\kk) = \mu_0(\kk)$ and $\mu_{t=1}(\kk) = \mu_1(\kk)$ for all $\kk \in \R^2$.
\end{definition}

Notice that a $\gamma$-periodic family of matrices is periodic in $k_3$ and only pseudo-periodic in $k_2$: the family $\gamma$ encodes the failure of $k_2$-periodicity.

\begin{proposition} \label{thm:gamma-null}
Let $\set{\alpha(k_2,k_3)}_{(k_2,k_3) \in \R^2}$ be a continuous and $\gamma$-periodic family of unitary matrices, and assume that $\deg_2(\det \alpha) = \deg_3(\det \alpha) = 0$. Then one can \emph{construct} a continuous and $\gamma$-periodic two-step logarithm for $\alpha$, namely there exist continuous and $\gamma$-periodic families of self-adjoint matrices $\set{h_i(\kk)}_{\kk \in \R^2}$, $i \in \set{1, 2}$, such that
\[ \alpha(k_2,k_3) = \eu^{\iu  h_1(k_2,k_3)} \, \eu^{\iu  h_2(k_2,k_3)}, \quad (k_2,k_3) \in \R^2. \]
\end{proposition}
\begin{proof}
The argument goes as in the proof of Proposition~\ref{thm:2-step}. One just needs to modify Step 1 there, where the approximants of $\alpha$ with completely non-degenerate spectrum are constructed obeying $\gamma$-periodicity rather than mere periodicity (compare Appendix~\ref{app:NonDegenerate}). It is also worth noting that both the spectrum and the norm of $\mu(k_2+1,k_3)$ coincide with the spectrum and the norm of $\mu(k_2,k_3)$ for any $\gamma$-periodic family of matrices $\mu$, and that the Cayley transform of a $\gamma$-periodic family of unitary matrices $\set{\alpha(k_2,k_3)}_{(k_2,k_3) \in \R^2}$ is also $\gamma$-periodic. Hence, logarithms constructed via functional calculus on the Cayley transform are automatically $\gamma$-periodic (see Appendix~\ref{app:Cayley}). Finally, observing that the spectrum of a $\gamma$-periodic family of matrices is $\Z^2$-periodic, the rest of the argument for Proposition~\ref{thm:2-step} goes through unchanged.
\end{proof}

The next result generalizes Proposition~\ref{thm:alpha} considerably.

\begin{proposition} \label{thm:gamma}
Assume that $D \le 2$. Let $\set{\alpha_0(\kk)}_{\kk \in \R^D}$ and  $\set{\alpha_1(\kk)}_{\kk \in \R^D}$ be continuous and periodic families of unitary matrices. Then the following are equivalent:
\begin{enumerate}
 \item \label{gamma_1} the families are homotopic;
 \item \label{gamma_2} $\deg_j(\det \alpha_0) = \deg_j(\det \alpha_1)$ for all $j \in \set{1,\ldots, D}$, where $\deg_j(\det \cdot)$ is defined in \eqref{deg};
 \item \label{gamma_3} one can \emph{construct} a continuous family of unitary matrices $\set{\beta(k_1,\kk)}_{(k_1,\kk) \in \R^d}$, $d=D+1$, which is $\Z^D$-periodic in $\kk$, with $\beta(0,\kk) \equiv \Id$ for all $\kk \in \R^D$, and such that
 \[ \alpha_1(\kk) = \beta(k_1,\kk) \, \alpha_0(\kk) \, \beta(k_1+1,\kk)^{-1}, \quad (k_1,\kk) \in \R^d. \]
\end{enumerate}

If $D=2$, then the above three statements remain equivalent even if one replaces periodicity by $\gamma$-periodicity and homotopy by $\gamma$-homotopy.
\end{proposition}
\begin{proof}
Since periodicity is a particular case of $\gamma$-periodicity, we give the proof in the $\gamma$-periodic framework. Set
\[ \alpha'(\kk) := \alpha_1(\kk)^{-1} \, \alpha_0(\kk), \quad \kk \in \R^D. \]
Then $\set{\alpha'(\kk)}_{\kk \in \R^D}$ is a continuous and $\gamma$-periodic family of unitary matrices, satisfying moreover $\deg_j(\det \alpha') = \deg_j(\det \alpha_0) - \deg_j(\det \alpha_1) = 0$ for $j \in \set{1, \ldots, D}$. In view of Proposition~\ref{thm:gamma-null}, one can construct a continuous and $\gamma$-periodic two-step logarithm for $\alpha'$:
\[ \alpha'(\kk) = \eu^{\iu  h_2(\kk)} \, \eu^{\iu  h_1(\kk)}. \]

Define
\[ \beta(k_1,\kk) := \eu^{\iu  \, k_1 \, h_2(\kk)} \, \eu^{\iu  \, k_1 \, h_1(\kk)}, \quad k_1 \in [0,1], \: \kk \in \R^D, \]
and extend this definition to positive $k_1>0$ by
\[ \beta(k_1+1,\kk) := \alpha_1(\kk)^{-1} \, \beta(k_1,\kk) \,\alpha_0(\kk) \]
and to negative $k_1 < 0$ by
\[ \beta(k_1,\kk) := \alpha_1(\kk) \, \beta(k_1+1,\kk) \, \alpha_0(\kk)^{-1}. \]
Notice first that the above defines a family of unitary matrices which is $\gamma$-periodic in $\kk$. We just need to show that this definition yields also a continuous function of $k_1$. We have $\beta(0^+,\kk)=\Id$ and $\beta(1^-,k_2)=\alpha_1(\kk)^{-1} \, \alpha_0(\kk)$ by definition. Let $\epsilon > 0$. If $k_1=-\epsilon$ is negative but close to zero, we have due to the definition
\[ \beta(-\epsilon,\kk) = \alpha_1(\kk) \, \beta(1-\epsilon,\kk) \, \alpha_0(\kk)^{-1} \to \alpha_1(\kk) \, \beta(1^-,\kk) \, \alpha_0(\kk)^{-1} = \Id \quad \text{as } \epsilon \to 0. \]
Hence $\beta$ is continuous at $k_1=0$. At $k_1=1$ we have instead
\[ \beta(1+\epsilon,\kk) = \alpha_1(\kk)^{-1} \, \beta(\epsilon,\kk) \, \alpha_0(\kk) \to \alpha_1(\kk)^{-1} \, \beta(0^+,\kk) \, \alpha_0(\kk) = \alpha_1(\kk)^{-1} \, \alpha_0(\kk) \quad \text{as } \epsilon \to 0 \]
and $\beta$ is also continuous there. A similar argument shows continuity of $k_1 \mapsto \beta(k_1,\kk)$ at every other integer value of $k_1$.

\medskip

Conversely, if we are given $\set{\beta(k_1,\kk)}_{(k_1,\kk) \in \R^3}$ as in the statement, then
 \[ \alpha_t(\kk) := \beta(-t/2,\kk)\alpha_0(\kk) \, \beta(t/2,\kk)^{-1}, \quad t \in [0,1], \: \kk \in \R^2, \]
gives the desired $\gamma$-homotopy between $\alpha_0$ and $\alpha_1$.
\end{proof}

\subsection{Orthonormal Bloch vectors} 

We now come back to the proof of Theorem~\ref{thm:Main}\eqref{Main_1}.

\begin{proof}[Proof of Theorem~\ref{thm:Main}\eqref{Main_1}]
Let us start from a $2$-dimensional smooth and $\Z^2$-periodic family of rank-$m$ projections $\set{P(\kk)}_{\kk \in \R^2}$. We replicate the construction at the beginning of Section~\ref{sec:c1=0} (see Equation~\eqref{psi}) to obtain an orthonormal collection of $m$ Bloch vectors $\set{\psi_a}_{a=1}^{m}$ for $\set{P(\kk)}_{\kk \in \R^2}$ which are continuous and $\Z$-periodic in the variable $k_2$. The continuous and periodic family of unitary matrices $\set{\alpha\sub{2D}(k_2)}_{k_2 \in \R}$, defined as in \eqref{alpha}, measures the failure of $\set{\psi_a}_{a=1}^{m}$ to be periodic in $k_1$.

Define
\begin{equation} \label{eqn:tildealpha} 
\widetilde{\alpha}\sub{2D}(k_2) := \begin{pmatrix}
\det \alpha\sub{2D}(k_2) & 0 & \cdots & 0 \\
0 & 1 & \cdots & 0 \\
\vdots & \vdots & \ddots & \vdots \\
0 & 0 & \cdots & 1
\end{pmatrix}.
\end{equation}
Clearly $\det \alpha\sub{2D}(k_2) = \det \widetilde{\alpha}\sub{2D}(k_2)$, so that in particular $\alpha\sub{2D}$ and $\widetilde{\alpha}\sub{2D}$ are homotopic. Proposition~\ref{thm:gamma} applies and produces a family of unitary matrices $\set{\beta\sub{2D}(k_1, k_2)}_{(k_1, k_2) \in \R^2}$ which is periodic in $k_2$ and such that
\[ \alpha\sub{2D}(k_2) = \beta\sub{2D}(k_1, k_2) \, \widetilde{\alpha}\sub{2D}(k_2) \, \beta\sub{2D}(k_1 + 1, k_2)^{-1} \]
holds for all $(k_1,k_2) \in \R^2$.

With $\set{\psi_a}_{a=1}^{m}$ as in~\eqref{psi} and $\set{\beta\sub{2D}(\kk)}_{\kk \in \R^2}$ as above, define
\[ \xi_a(\kk) := \sum_{b=1}^{m} \psi_b(\kk)  \, \beta\sub{2D}(\kk)_{ba}, \quad a \in \set{1, \ldots, m}, \: \kk \in \R^2. \]
Then we see that for $a \in \set{1, \ldots, m}$ and $(k_1,k_2) \in \R^2$
\begin{equation} \label{newmatch}
\begin{aligned}
\xi_a(k_1 + 1, k_2) & = \sum_{b=1}^{m} \psi_b(k_1+1, k_2)  \, \beta\sub{2D}(k_1+1,k_2)_{ba} = \sum_{b,c=1}^{m} \psi_c(k_1, k_2) \, \alpha\sub{2D}(k_2)_{cb} \, \beta\sub{2D}(k_1+1,k_2)_{ba} \\
& = \sum_{c=1}^{m} \psi_c(k_1, k_2) \, \left[ \alpha\sub{2D}(k_2) \, \beta\sub{2D}(k_1+1,k_2)\right]_{ca} = \sum_{c=1}^{m} \psi_c(k_1, k_2) \, \left[ \beta\sub{2D}(k_1,k_2) \, \widetilde{\alpha}\sub{2D}(k_2) \right]_{ca} \\
& = \sum_{b=1}^{m} \sum_{c=1}^{m} \psi_c(k_1, k_2) \, \beta\sub{2D}(k_1,k_2)_{cb} \, \widetilde{\alpha}\sub{2D}(k_2)_{ba} = \sum_{b=1}^{m} \xi_b(k_1, k_2) \, \widetilde{\alpha}\sub{2D}(k_2)_{ba}.
\end{aligned}
\end{equation}

Since $\widetilde{\alpha}\sub{2D}(k_2)$ is in the form \eqref{eqn:tildealpha}, when we set $a \in \set{2, \ldots, m}$ in the above equation this reads $\xi_a(k_1+1,k_2) = \xi_a(k_1,k_2)$, that is, $\set{\xi_a}_{a=2}^{m}$ is an orthonormal collection of $(m-1)$ Bloch vectors which are continuous \emph{and $\Z^2$-periodic}. This concludes the proof of Theorem~\ref{thm:Main}\eqref{Main_1} in the $2$-dimensional case.

\bigskip

We now move to the case $d=3$. Let $\set{P(\kk)}_{\kk \in \R^3}$ be a family of rank-$m$ projections which is smooth and $\Z^3$-periodic. In view of what we have just proved, the $2$-dimensional restriction $\set{P(0,k_2,k_3)}_{(k_2,k_3) \in \R^2}$ admits a collection of $m$ orthonormal Bloch vectors $\set{\xi_a(0,\cdot,\cdot)}_{a=1}^{m}$ satisfying
\begin{equation} \label{pseudo}
\begin{aligned}
\xi_1(0,k_2+1,k_3)& = \det \alpha\sub{2D}(k_3) \, \xi_1(0,k_2,k_3), \\
\xi_b(0,k_2+1,k_3)& = \xi_b(0,k_2,k_3) \text{ for all } b \in \set{2, \ldots, m}, \\
\xi_a(0,k_2,k_3+1)& = \xi_a(0,k_2,k_3) \text{ for all } a \in \set{1, \ldots, m}.
\end{aligned}
\end{equation}

Parallel-transport these Bloch vectors along the $k_1$-direction, and define $\set{\psi_a}_{a=1}^{m}$ as in \eqref{psi} and $\set{\alpha(k_2,k_3)}_{(k_2,k_3) \in \R^2}$ as in \eqref{alpha}. The latter matrices are still unitary, depend continuously on $(k_2,k_3)$, are periodic in $k_3$, but
\[ \alpha(k_2+1,k_3) = \widetilde{\alpha}\sub{2D}(k_3) \, \alpha(k_2,k_3) \, \widetilde{\alpha}\sub{2D}(k_3)^{-1}, \]
as can be checked from \eqref{pseudo}. Thus, the family $\set{\alpha(k_2,k_3)}_{(k_2,k_3) \in \R^2}$ is $\widetilde{\alpha}\sub{2D}$-periodic, and consequently so is the family defined by
\[ \widetilde{\alpha}(k_2,k_3) := \begin{pmatrix}
\det \alpha(k_2,k_3) & 0 & \cdots & 0 \\
0 & 1 & \cdots & 0 \\
\vdots & \vdots & \ddots & \vdots \\
0 & 0 & \cdots & 1
\end{pmatrix} \]
(actually, since $\widetilde{\alpha}$ and $\widetilde{\alpha}\sub{2D}$ commute, in this case $\widetilde{\alpha}\sub{2D}$-periodicity reduces to mere periodicity). Since $\alpha$ and $\widetilde{\alpha}$ share the same determinant, Proposition~\ref{thm:gamma} again produces a continuous, $\widetilde{\alpha}\sub{2D}$-periodic family of unitary matrices $\set{\beta(k_1, \kk)}_{(k_1, \kk) \in \R^3}$ such that for all $(k_1,\kk) \in \R^3$
\[ \alpha(\kk) = \beta(k_1, \kk) \, \widetilde{\alpha}(\kk) \, \beta(k_1 + 1, \kk)^{-1}. \]

Arguing as above (compare \eqref{newmatch}), the collection of Bloch vectors defined by
\[ \xi_a(k_1,\kk) := \sum_{b=1}^{m} \psi_b(k_1,\kk)  \, \beta(k_1,\kk)_{ba}, \quad a \in \set{1, \ldots, m}, \: (k_1,\kk) \in \R^3, \]
satisfies
\[ \xi_a(k_1 + 1, \kk) = \sum_{b=1}^{m} \xi_b(k_1, \kk) \, \widetilde{\alpha}(\kk)_{ba}, \quad a \in \set{1,\ldots,m}. \]
Due to the form of $\widetilde{\alpha}$, this implies again that $\set{\xi_a}_{a=2}^{m}$ are continuous, orthonormal, \emph{and $\Z^3$-periodic} Bloch vectors for $\set{P(\kk)}_{\kk \in \R^3}$, thus concluding the proof.
\end{proof}


\section{Proof of Theorem~\ref{thm:Main}\eqref{Main_2}: moving Parseval frames of Bloch vectors} \label{sec:thmMain(ii)}

In this Section, we finally prove Theorem~\ref{thm:Main}\eqref{Main_2}, and complete the proof of the first main result. The central step consists in proving the result for families of rank $1$, which we will do first.

\subsection{The rank-1 case}

\begin{proof}[Proof of Theorem~\ref{thm:Main}\eqref{Main_2} (rank-$1$ case)]
Let $d \le 3$. We consider first a smooth and $\Z^d$-periodic family of projections $\set{P_1(\kk)}_{\kk \in \R^d}$ of rank $m = 1$. We want to show that there exists two Bloch vectors $\set{\xi_1, \xi_2}$ which are continuous, $\Z^d$-periodic, and generate the $1$-dimensional space $\Ran P_1(\kk) \subset \Hi$ at each $\kk \in \R^d$.

To do so, fix a complex conjugation $C$ on the Hilbert space $\Hi$ (which is tantamount to the choice of an orthonormal basis). Define
\begin{equation} \label{Q=CPC}
Q(\kk) := C \, P_1(-\kk) \, C^{-1}.
\end{equation}
Using the fact that $C$ is an antiunitary operator such that $C^2 = \Id$, one can check that $\set{Q(\kk)}_{\kk \in \R^d}$ defines a smooth and $\Z^d$-periodic family of orthogonal projectors. Moreover, one also has that $c_1(Q)_{ij} = - c_1(P)_{ij}$ for all $1 \le i<j \le d$, as can be seen by integrating the identity
\[ \Tr_{\Hi} \left( Q(\kk) \, \left[ \partial_i Q(\kk), \partial_j Q(\kk) \right] \right) = - \Tr_{\Hi} \left( P_1(-\kk) \, \left[ \partial_i P_1(-\kk), \partial_j P_1(-\kk) \right] \right) \]
over $\T^2_{ij}$ (compare \cite{Panati07, AucklyKuchment18}).

Set now $P(\kk) := P_1(\kk) \oplus Q(\kk)$ for $\kk \in \R^d$. The rank-$2$ family of projections $\set{P(\kk)}_{\kk \in \R^d}$ on $\Hi \oplus \Hi$ satisfies then
\[ c_1(P)_{ij} = c_1(P_1)_{ij} + c_1(Q)_{ij} = 0 \quad \text{for all } 1 \le i<j\le d. \]
Hence, in view of the results of Section~\ref{sec:c1=0}, it admits a Bloch basis $\set{\psi_1, \psi_2}$. Let $\pi_j \colon \Hi \oplus \Hi \to \Hi$ be the projection on the $j$-th factor, $j \in \set{1,2}$. Set finally
\[ \xi_a(\kk) := \pi_1 \left( \left(P_1(\kk) \oplus 0\right) \, \psi_a(\kk) \right), \quad a \in \set{1,2}, \: \kk \in \R^d. \]

Let us show that $\set{\xi_a(\kk)}_{a=1}^{2}$ gives a (continuous and $\Z^d$-periodic) Parseval frame in $\Ran P_1(\kk)$. Indeed, let $\psi \in \Ran P_1(\kk)$: then automatically $\psi \oplus 0 \in  \Ran P(\kk)$. Since $\set{\psi_a(\kk)}_{a=1}^{2}$ is an orthonormal basis for $\Ran P(\kk)$, we obtain that
\[ \psi \oplus 0 = \sum_{a=1}^{2} \scal{\psi_a(\kk)}{\psi \oplus 0}_{\Hi \oplus \Hi} \psi_a(\kk)= \sum_{a=1}^{2} \scal{\xi_a(\kk)}{\psi}_{\Hi} \psi_a(\kk). \]
Finally, we apply $\pi_1 \circ \left(P_1(\kk) \oplus 0\right)$ on both sides and obtain
\[ \psi = \sum_{a=1}^{2} \scal{\xi_a(\kk)}{\psi}_{\Hi} \xi_a(\kk) \]
which is the defining condition~\eqref{april1} for $\set{\xi_a(\kk)}_{a=1}^{2}$ to be a Parseval frame in $\Ran P_1(\kk)$.
\end{proof}

\subsection{\texorpdfstring{The higher rank case: $m>1$}{The higher rank case: m>1}}

\begin{proof}[Proof of Theorem~\ref{thm:Main}\eqref{Main_2} (rank-$m$ case, $m>1$)]
Let $d \le 3$, and consider a smooth and $\Z^d$-periodic family of rank-$m$ projections $\set{P(\kk)}_{\kk \in \R^d}$. In view of Theorem~\ref{thm:Main}\eqref{Main_1}, which we already proved, it admits $m-1$ orthonormal Bloch vectors $\set{\xi_a}_{a=1}^{m-1}$: they are $\Z^d$-periodic, and without loss of generality (see Appendix~\ref{app:Smoothing}) we assume them to be smooth. Denote by 
\[ P_{m-1}(\kk) := \sum_{a=1}^{m-1} \ket{\xi_a(\kk)} \bra{\xi_a(\kk)}, \quad \kk \in \R^2, \]
the rank-$(m-1)$ projection onto the space spanned by $\set{\xi_a(\kk)}_{a=1}^{m-1}$. Since the latter are smooth and periodic Bloch vectors for $\set{P(\kk)}_{\kk \in \R^d}$, the family $\set{P_{m-1}(\kk)}_{\kk \in \R^d}$ is smooth, $\Z^d$-periodic, and satisfies $P_{m-1}(\kk) \, P(\kk) = P(\kk) \, P_{m-1}(\kk) = P_{m-1}(\kk)$.

Denote by $P_1(\kk)$ the orthogonal projection onto the orthogonal complement of $\Ran P_{m-1}(\kk)$ inside $\Ran P(\kk)$. Then $\set{P_1(\kk)}_{\kk \in \R^d}$ is a smooth and $\Z^d$-periodic family of rank-$1$ projections, and furthermore $P_1(\kk) \, P(\kk) = P(\kk) \, P_1(\kk) = P_1(\kk)$. In view of the results of the previous Subsections, we can construct two continuous and $\Z^d$-periodic Bloch vectors $\set{\xi_m, \xi_{m+1}}$ which generate $\Ran P_1(\kk)$ at all $\kk \in \R^d$. Since $P_1(\kk)$ is a sub-projection of $P(\kk)$, it then follows that 
\[ P(\kk) \, \xi_a(\kk) = P(\kk) \, P_1(\kk) \, \xi_a(\kk) = P_1(\kk) \, \xi_a(\kk) = \xi_a(\kk) \quad \text{for all } a \in \set{m, m+1} . \]
Besides, by construction $\set{\xi_m(\kk), \xi_{m+1}(\kk)}$ generate the orthogonal complement in $\Ran P(\kk)$ to the span of $\set{\xi_a(\kk)}_{a=1}^{m-1}$, and hence the full collection of $m+1$ Bloch vectors $\set{\xi_a}_{a=1}^{m+1}$ give an $(m+1)$-frame for $\set{P(\kk)}_{\kk \in \R^d}$ consisting of continuous and $\Z^d$-periodic vectors, as desired.
\end{proof}
 
\section{Proof of Theorem~\ref{thm:Main2}: from Hofstadter-like Hamiltonians to Fermi-like magnetic projections} \label{sec:Hofstadter}

In this Section we leave the periodic setting and start the proof of Theorem~\ref{thm:Main2}, which applies to a 2-dimensional discrete magnetic Hamiltonian $H_b$, as described in Section~\ref{sec:Model} (from which we borrow much of the notation). We show that, in order to exhibit an exponentially localized Wannier Parseval frame for the the Fermi projection of the original Hamiltonian, it is sufficient to solve the same problem for a suitable family of \emph{Fermi-like magnetic projections}, in the sense of Definition~\ref{dfn:Fermi-like} and Lemma~\ref{lemmaKN} below. 

\subsection{Fermi-like magnetic projections}

In the following, we drop the dependence on $q$ for notational convenience, effectively setting $q=1$ in the magnetic phase $\eu^{\iu \epsilon q \phi(\cdot,\cdot)}$.

We assume that $H_{b_0}$, and hence $\widetilde{\HH}_0 \equiv U_{b_0} H_{b_0} U_{b_0}^*$, has an isolated spectral island. Then we know that there exists $\epsilon^*$ such that for every $\epsilon \le \epsilon^*$ also $\widetilde{\HH}_\epsilon$ has an isolated spectral island \cite{Cornean}. This allows us to define the family of spectral projections onto this spectral island of $\widetilde{\HH}_\epsilon$, that we denote by $\big\{\widetilde{\mathcal{P}}_\epsilon\big\}_{0 \le \epsilon \le \epsilon^*}$. In the following, we will show that there exists an $\epsilon_0<\epsilon^*$ such that $\widetilde{\mathcal{P}}_\epsilon$ admits a Parseval frames in the sense of Theorem~\ref{thm:Main2}, for every $\epsilon<\epsilon_0$. Notice that this implies also the existence of a Parseval frame for the Fermi projection $\Pi_{b=b_0+\epsilon}$ of the original Hamiltonian $H_{b=b_0+\epsilon}$ consisting of exponentially localized generalized Wannier functions, since the two projections $\Pi_{b=b_0+\epsilon}$ and $\widetilde{\mathcal{P}}_\epsilon$ are unitarily conjugated by the explicit unitary multiplication operator $U_{b=b_0+\epsilon}$ in \eqref{iunie1}, see Corollary \ref{CorGWB}.

From \eqref{apr1'} we see that the $\epsilon$-dependence of (the matrix elements of) the Hamiltonian $\widetilde{\HH}_\epsilon$ is both in the phase factor and in $h_{\kk,\epsilon}$. Nevertheless, it is easier to remove the latter dependence; we will see that this does not spoil the validity of Theorem~\ref{thm:Main2}. We thus consider the operator ${\HH}_\epsilon$ acting on $\ell^2(\Z^2)\otimes \C^{Q}$ defined by the matrix elements
\begin{equation} \label{HHepsilon}
{\HH}_\epsilon(\gamma,\xx;\gamma',\xx') := \eu^{\iu \epsilon \phi(\gamma,\gamma')} \mathcal{T}_0(\gamma-\gamma';\xx,\xx'), \quad \text{where} \quad \mathcal{T}_0(\gamma;\xx,\xx') := \int_\Omega \di \kk \, \eu^{\iu 2\pi \kk\cdot \gamma} h_{\kk,0}(\xx,\xx')
\end{equation}
for $\gamma,\gamma'\in\Z^2$ and $\xx,\xx'\in \set{1,\ldots,Q}$. Notice that, in view of Remark~\ref{rmk:MatrixCommuteMagn}, ${\HH}_\epsilon$ still commutes with the magnetic translations $\tau_{\epsilon,\eta}$, $\eta \in \Z^2$, defined in~\eqref{eqn:tau}.

The following Lemma ensures that if $\widetilde{\HH}_{\epsilon}$ is gapped, also $\HH_{\epsilon}$ is gapped.

\begin{lemma}
There exists $\widetilde{\epsilon}_0$ such that, for every $0\le\epsilon\le\widetilde{\epsilon}_0$, the spectral island $\sigma_{b_0+\epsilon}$ of ${\HH}_{\epsilon}$ is $\epsilon$-close in the Hausdorff distance to a spectral island $\widetilde{\sigma}_{b_0+\epsilon}$ of $\widetilde{\HH}_\epsilon$.
\end{lemma}
\begin{proof}
From \eqref{febru5} we see that 
\[ \left| {\HH}_\epsilon(\gamma,\xx;\gamma',\xx') - \widetilde{\HH}_\epsilon(\gamma,\xx;\gamma',\xx') \right| \leq C \, \epsilon \, \left( \|\gamma-\gamma'\| +1 \right) \left| \mathcal{T}(\gamma-\gamma';\xx,\xx') \right|, \]
which implies the estimate
\begin{equation}\label{hc20}
\left| {\HH}_\epsilon(\gamma,\xx;\gamma',\xx') - \widetilde{\HH}_\epsilon(\gamma,\xx;\gamma',\xx') \right| \leq C \, \epsilon \,  \eu^{-\beta \|\gamma-\gamma'\|}
\end{equation}
for some positive constants $C, \beta > 0$ uniformly in $\xx, \xx' \in \set{1, \ldots, Q}$. From the above we conclude via a Schur--Holmgren estimate that $\norm{{\HH}_\epsilon-\widetilde{\HH}_\epsilon} \leq C' \, \epsilon$ for some constant $C'>0$. An elementary argument based on Neumann series shows that if ${\rm dist}(z,\sigma(\HH_\epsilon))> C' \,\epsilon$ then $z$ must also be in the resolvent set of $\widetilde{\HH}_\epsilon$. The argument is symmetric in the two operators, hence the spectra are at Hausdorff distance $\epsilon$. 
\end{proof}

In view of the above Lemma, the family of projections $\{\mathcal{P}_\epsilon\}_{\epsilon \in \left[0,\widetilde{\epsilon}_0\right]}$ onto the spectral island $\sigma_{b_0+\epsilon}$ of ${\HH}_{\epsilon}$ is well defined, for example by the Riesz formula
\begin{equation} \label{RieszFormula}
\mathcal{P}_\epsilon = \frac{\iu}{2 \pi} \oint_{\mathcal{C}} \di z \, ({\HH}_{\epsilon} - z)^{-1},
\end{equation}
where $\mathcal{C}$ is a positively-oriented contour in the complex energy plane which encloses only the spectral island $\sigma_{b_0+\epsilon}$. This family of projections satisfies a number of properties (see Proposition~\ref{ApplicationMPF} below), which for later convenience we collect in the following Definition.

\begin{definition}[Fermi-like magnetic projections] \label{dfn:Fermi-like}
A family of projections $\set{\Pi^{(\epsilon)}}_{\epsilon \in [0,\epsilon_0)}$ acting on $\ell^2(\Z^2) \otimes \C^Q$ is called a family of \emph{Fermi-like magnetic projections} if the following properties are satisfied:
\begin{enumerate}[label=(\roman*),ref=\roman*]
 \item {\sl perturbation of a periodic projection}: the projection $\Pi_0\equiv\Pi^{(0)}$ is such that there exists a family of rank-$m$ orthogonal projection $P_0(\kk)$ acting on $\C^Q$ which is smooth and $\Z^2$-periodic as a function of $\kk$, and such that 
\begin{equation} \label{ZeroPeriodic}
\Pi_0(\gamma,\xx;\gamma',\xx')= \int_{\Omega} \di\kk \, \eu^{\iu 2\pi \kk \cdot (\gamma-\gamma')} P_0(\kk)(\xx,\xx') \quad \text{for all } \gamma,\gamma' \in \Z^2, \: \xx,\xx' \in \set{1,\ldots,Q};
\end{equation}
moreover for some positive constants $C$ and $\alpha$
\begin{equation} \label{DiagLoc}
\left|\Pi^{(\epsilon)}(\gamma,\xx;\gamma',\xx') - \eu^{\iu \epsilon \phi(\gamma,\gamma')} \Pi_{0}(\gamma,\xx;\gamma',\xx') \right| \leq C \,\epsilon \, \eu^{-\alpha \|\gamma-\gamma'\|}
\end{equation}
for all $\gamma,\gamma' \in \Z^2$ and $\xx,\xx' \in \set{1,\ldots,Q}$;
 \item {\sl exponential localization of the matrix elements}: for some positive constants $C$ and $\Ll$ we have 
\begin{equation} \label{ExpLocIntKernel}
\left| \Pi^{(\epsilon)}(\gamma,\xx;\gamma',\xx') \right| \leq  C \eu^{-\Ll \|\gamma-\gamma'\|}
\end{equation} 
for all $\gamma, \gamma' \in \Z^2$, $\xx, \xx' \in \set{1, \ldots, Q}$ and $\epsilon \in [0,\epsilon_0)$;
 \item {\sl intertwining with the magnetic translations}: for all $\eta\in \Z^2$
\begin{equation} \label{InterMagTransl}
\eu^{\iu \epsilon \phi(\gamma,\eta)} \Pi^{(\epsilon)}(\gamma-\eta,\xx;\gamma'-\eta,\xx') \eu^{-\iu \epsilon \phi(\gamma',\eta)} = \Pi^{(\epsilon)}(\gamma,\xx;\gamma',\xx') \quad \forall \: \gamma,\gamma' \in \Z^2, \: \xx,\xx' \in \set{1,\ldots,Q}.
\end{equation}
\end{enumerate}
\end{definition}

Notice that the above relation~\eqref{InterMagTransl} is clearly equivalent to
\begin{equation} \label{MagCommute}
\tau_{\epsilon,\eta} \Pi^{(\epsilon)} \tau_{\epsilon,\eta}^* = \Pi^{(\epsilon)}\, , \quad \forall \: \eta \in \Z^2
\end{equation}
by a computation analogous to~\eqref{luglio2}.

\begin{proposition} \label{ApplicationMPF}
The family of projections $\{\mathcal{P}_\epsilon\}_{\epsilon \in \left[0,\widetilde{\epsilon}_0\right]}$ is a family of magnetic Fermi-like projections in the sense of Definition \ref{dfn:Fermi-like}.
\end{proposition}
\begin{proof}
Since $\HH_\epsilon$ commutes with the magnetic translations, every $\mathcal{P}_\epsilon$ satisfies~\eqref{InterMagTransl} by construction. Moreover, they also satisfy~\eqref{ExpLocIntKernel}. This is a direct application of the Combes--Thomas estimates on the resolvent of $\HH_\epsilon$, because the gap in the spectrum of the Hamiltonian persists for every $\epsilon$. For completeness we report the proof of the Combes--Thomas estimates in Appendix~\ref{app:KernelEstimates}, see Proposition~\ref{Combes-Thomas}.

By hypothesis, $\HH_0=\widetilde{\HH}_0$ has a (possibly) non-trivial spectral island $\sigma_{b_0}$ which must come from the ranges of $m<Q$ bands of $h_{\kk,0}$. In other words, $h_{\kk,0}$ has $m<Q$ eigenvalues whose ranges remain separated from the other $Q-m$, and these ranges together build up the spectral island $\sigma_{b_0}$.  Denote by $P_0(\kk)$ the spectral projection of $h_{\kk,0}$ corresponding to these $m$ eigenvalues. From the properties of $h_{\kk,0}$ we can deduce that $P_0(\kk)$ is smooth and periodic in $\kk$, and the spectral projection of $\HH_0$ onto $\sigma_{b_0}$ is simply given by 
\begin{equation*}
\mathcal{P}_0(\gamma,\xx;\gamma',\xx')= \int_{\Omega} \di\kk\, \eu^{\iu  2\pi \kk \cdot (\gamma-\gamma')} P_0(\kk)(\xx,\xx'), \quad \gamma,\gamma' \in \Z^2, \: \xx,\xx' \in \set{1,\dots,Q}.
\end{equation*}
Hence also the property~\eqref{ZeroPeriodic} is proved.

Now it remains to prove~\eqref{DiagLoc}. This result is  essentially known, see~\cite{Cornean} and~\cite{BrynildsenCornean13}; for completeness, we present a proof adapted to our setting in~Proposition \ref{prop:DiagLoc}.
\end{proof}

Coming back to the problem of constructing a Parseval frame spanning the range of the Fermi projection $\widetilde{\mathcal{P}}_\epsilon$ of $\widetilde{\HH}_\epsilon$, we show that it is in fact sufficient to construct it for the Fermi-like magnetic projections $\mathcal{P}_\epsilon$.

\begin{lemma} \label{lemmaKN}
There exists $\epsilon_0 \le \widetilde{\epsilon}_0$ such that, for every $0\le\epsilon\le\epsilon_0$, the Fermi-like magnetic projection $\mathcal{P}_\epsilon$ admits a magnetic Parseval frame in the sense of Theorem~\ref{thm:Main2} if and only if the Fermi projection $\widetilde{\mathcal{P}}_\epsilon$ admits it. 
\end{lemma}
\begin{proof}
The two families $\{\mathcal{P}_\epsilon\}_{\epsilon \in \left[0,\widetilde{\epsilon}_0\right]}$ and $\{\widetilde{\mathcal{P}}_\epsilon\}_{\epsilon \in \left[0,\widetilde{\epsilon}_0\right]}$ are projections onto isolated spectral island of Hamiltonians with exponentially localized matrix elements (see Proposition~\ref{Combes-Thomas}), and hence have themselves exponentially localized matrix elements in view of the Riesz formula~\eqref{RieszFormula} (compare~\eqref{ExpLocIntKernel}). Using the resolvent identity and \eqref{hc20}, it follows that
\begin{equation} \label{giugno27}
\left| {\mathcal{P}}_{\epsilon}(\gamma,\xx;\gamma',\xx') -\widetilde{\mathcal{P}}_{\epsilon}(\gamma,\xx;\gamma',\xx') \right| \leq C \, \epsilon \, \eu^{-\beta \|\gamma-\gamma'\|}.
\end{equation}
Because of this, there exists an $\epsilon_0>0$ such that $\|{\mathcal{P}}_{\epsilon}-\widetilde{\mathcal{P}}_{\epsilon}\|\leq 1/2$ for all $\epsilon\le\epsilon_0$. Then via the associated Kato--Nagy unitary $K_\epsilon$ \cite{Kato66} we have $\widetilde{\mathcal{P}}_\epsilon = K_\epsilon \mathcal{P}_\epsilon K_\epsilon^{-1}$. The formula
\begin{equation} \label{eqn:KatoNagy}
\begin{aligned}
K_\epsilon&=\left[{\bf 1}-\left(\widetilde{\mathcal{P}}_\epsilon-{\mathcal{P}}_\epsilon\right)^{2}\right]^{-1/2} \left(\widetilde{\mathcal{P}}_\epsilon {\mathcal{P}}_\epsilon+ \left({\bf 1} - \widetilde{\mathcal{P}}_\epsilon  \right) \left({\bf 1} - {\mathcal{P}}_\epsilon \right) \right)\\
&=\left[{\bf 1}+\sum_{k=1}^{\infty}\frac{\left(2k-1\right)!!}{k! 2^k}\left(\widetilde{\mathcal{P}}_\epsilon-{\mathcal{P}}_\epsilon\right)^{2k}\right] \left(\widetilde{\mathcal{P}}_\epsilon {\mathcal{P}}_\epsilon+ \left({\bf 1} - \widetilde{\mathcal{P}}_\epsilon  \right) \left({\bf 1} - {\mathcal{P}}_\epsilon \right) \right)
\end{aligned}
\end{equation}
for the Kato--Nagy unitary clearly shows that also $K_\epsilon$ commutes with the magnetic translations. Thus we see that if $\big\{w_a^{(\epsilon)}\big\}_{1 \le a \le M}$ is such that~\eqref{apr7} holds then the vectors $\big\{K_\epsilon^{*} w_a^{(\epsilon)}\big\}_{1\leq a\leq M}$ will span (together with all their magnetic translates) the range of ${\mathcal{P}}_{\epsilon}$, and viceversa. In order to see that the vectors $K_\epsilon^{*} w_a^{(\epsilon)} = w_a^{(\epsilon)} + (K_\epsilon^* - {\bf 1}) w_a^{(\epsilon)}$ are also exponentially localized, it suffices to show that $K_\epsilon-{\bf 1}$ has exponentially localized matrix elements in the sense of~\eqref{ExpLocIntKernel}, which we will do below in Lemma~\ref{Nagy}. This will conclude the proof.
\end{proof}

\begin{lemma} \label{Nagy}
For $K_\epsilon$ in \eqref{eqn:KatoNagy}, there exist constants $C, \lambda>0$ such that for all $\gamma, \gamma' \in \Z^2$ and $\xx, \xx' \in \set{1, \ldots, Q}$ and if $\epsilon$ is small enough
\[ \left| K_\epsilon(\gamma,\xx;\gamma',\xx') - \delta_{\gamma,\gamma'} \, \delta_{\xx,\xx'} \right| \leq  C \eu^{-\Ll \|\gamma-\gamma'\|}. \]
\end{lemma}
\begin{proof}
Using $\widetilde{\mathcal{P}}_\epsilon = \widetilde{\mathcal{P}}_\epsilon^2$ and ${\mathcal{P}}_\epsilon = {\mathcal{P}}_\epsilon^2$, let us first rewrite~\eqref{eqn:KatoNagy} as
\begin{equation} \label{giugno27'}
\begin{aligned}
K_\epsilon-{\bf 1}&=2\widetilde{\mathcal{P}}_\epsilon {\mathcal{P}}_\epsilon - \widetilde{\mathcal{P}}_\epsilon   - {\mathcal{P}}_\epsilon +\left[\sum_{k=1}^{\infty}\frac{\left(2k-1\right)!!}{k! 2^k}\left(\widetilde{\mathcal{P}}_\epsilon-{\mathcal{P}}_\epsilon\right)^{2k}\right] \left(\widetilde{\mathcal{P}}_\epsilon {\mathcal{P}}_\epsilon+ \left({\bf 1} - \widetilde{\mathcal{P}}_\epsilon  \right) \left({\bf 1} - {\mathcal{P}}_\epsilon \right) \right) \\
&=\widetilde{\mathcal{P}}_\epsilon \left({\mathcal{P}}_\epsilon - \widetilde{\mathcal{P}}_\epsilon \right) + \left({\mathcal{P}}_\epsilon - \widetilde{\mathcal{P}}_\epsilon \right) {\mathcal{P}}_\epsilon +\left[\sum_{k=1}^{\infty}\frac{\left(2k-1\right)!!}{k! 2^k}\left(\widetilde{\mathcal{P}}_\epsilon-{\mathcal{P}}_\epsilon\right)^{2k}\right] \left(\widetilde{\mathcal{P}}_\epsilon {\mathcal{P}}_\epsilon+ \left({\bf 1} - \widetilde{\mathcal{P}}_\epsilon  \right) \left({\bf 1} - {\mathcal{P}}_\epsilon \right) \right).
\end{aligned}
\end{equation}
The right-hand side of the above is a sum of operators which are all expressed in terms of $D_\epsilon := \widetilde{\mathcal{P}}_\epsilon-{\mathcal{P}}_\epsilon$ times the bounded operators ${\mathcal{P}}_\epsilon$, $\widetilde{\mathcal{P}}_\epsilon$ and $\widetilde{\mathcal{P}}_\epsilon {\mathcal{P}}_\epsilon+({\bf 1} - \widetilde{\mathcal{P}}_\epsilon) ({\bf 1} - {\mathcal{P}}_\epsilon)$. Notice that, in view of~\eqref{giugno27}, we have
\begin{equation} \label{Est1}
\left|D_\epsilon(\gamma,\xx;\gamma',\xx') \right| \leq C \, \epsilon \,  \eu^{-\beta\|\gamma-\gamma'\|} \, , \quad \gamma, \gamma' \in \Z^2, \: \xx,\xx' \in \set{1,\ldots,Q}.
\end{equation}
Thus, to show that the matrix elements of $K_\epsilon - {\bf 1}$ satisfy the estimate in the statement, it suffices to show that the series in square brackets appearing in~\eqref{giugno27'} defines an operator $O_\epsilon$ with exponentially localized matrix elements. 

Consider the matrix elements of $D_\epsilon^{n}$ with $n\in \N$ and $n\geq2$:
\[ D_\epsilon^n(\gamma,\xx;\gamma',\xx')=\sum_{\gamma_1\in \Z^2}\sum_{\xx_1=1}^{Q} \cdots \sum_{\gamma_{n-1}\in \Z^2}\sum_{\xx_{n-1}=1}^{Q} D_\epsilon(\gamma,\xx;\gamma_1,\xx_1) \cdots D_\epsilon(\gamma_{n-1},\xx_{n-1};\gamma',\xx')\, . \]
In view of \eqref{Est1} we have that, for $0<\beta'<\beta$,
\begin{align*}
\eu^{\beta' \|\gamma-\gamma'\|} \left|D^n_\epsilon(\gamma,\xx;\gamma',\xx')\right| &\leq C \, \epsilon \,  Q^{n-1} \left( \sup_{\gamma'' \in \Z^2}\sup_{\xx'',\yy\in \set{1,\ldots,Q}} \sum_{\eta \in \Z^2} |D_\epsilon(\gamma'',\xx'';\eta,\yy)| \eu^{\beta' \|\gamma''-\eta\|} \right)^{n-1}\\
&\leq \epsilon^{n} (C')^{n}
\end{align*}
for some $C'>0$. With this estimate, it follows that 
\[ |O_\epsilon(\gamma,\xx;\gamma',\xx')| \le [(1-(\epsilon C'))^{-1/2}-1] \, \eu^{-\beta' \|\gamma-\gamma'\|} \le C'' \, \epsilon  \, \eu^{-\beta' \|\gamma-\gamma'\|} \]
for some $C''>0$ uniform in $\gamma, \gamma' \in \Z^2$, $\xx, \xx' \in \set{1,\ldots,Q}$, and $\epsilon$ sufficiently small.
\end{proof}


\section{Proof of Theorem~\ref{thm:Main2}: Parseval frames for Fermi-like magnetic projections}  \label{sec:MagneticParseval}

In view of the discussion in the previous Section, we have reduced the statement of Theorem~\ref{thm:Main2} to the following equivalent result, formulated in terms of Fermi-like magnetic projections.

\begin{proposition} \label{thm:Main3}
Let $\set{\Pi^{(\epsilon)}}_{\epsilon \in [0,\epsilon_0)}$ be a family of Fermi-like magnetic projections as in Definition~\ref{dfn:Fermi-like}. Then there exists $\epsilon_0'<\epsilon_0$ such that for all $0\leq \epsilon\leq \epsilon_0'$ the following hold:
\begin{enumerate}[label=(\roman*),ref=\roman*]
\item \label{item:Main3_i} one can \emph{construct} $m-1$ orthonormal exponentially localized vectors $\big\{w_s^{(\epsilon)}\}_{1 \le s \le m-1}$ and two other exponentially localized vectors $\big\{W_r^{(\epsilon)}\big\}_{1 \le r \le 2}$ such that 
\[ \Pi_1^{(\epsilon)} := \sum_{\gamma \in \Z^2} \sum_{s=1}^{m-1} \ket{\tau_{\epsilon,\gamma}\, w_s^{(\epsilon)}} \bra{\tau_{\epsilon,\gamma}\, w_s^{(\epsilon)}} \quad \text{and} \quad \Pi_2^{(\epsilon)} := \sum_{\gamma \in \Z^2} \sum_{r=1}^{2} \ket{\tau_{\epsilon,\gamma}\, W_r^{(\epsilon)}} \bra{\tau_{\epsilon,\gamma}\, W_r^{(\epsilon)}} \]
are two orthogonal projections commuting with all magnetic translations $\tau_{\epsilon,\gamma}$ defined in~\eqref{eqn:tau} and such that $\Pi^{(\epsilon)} = \Pi_1^{(\epsilon)} + \Pi_2^{(\epsilon)}$ and $\Pi_1^{(\epsilon)}\Pi_2^{(\epsilon)}=0$;
 \item \label{item:Main3_ii} if moreover $c_1(P_0)=0 \in \Z$, where $P_0$ is as in~\eqref{ZeroPeriodic}, then one can \emph{construct} $m$ \emph{orthonormal} and \emph{exponentially localized} vectors $\big\{w_a^{(\epsilon)}\big\}_{1 \le a \le m}$ such that
\[ \Pi^{(\epsilon)} = \sum_{\gamma \in \Z^2} \sum_{a=1}^{m} \ket{\tau_{\epsilon,\gamma}\, w_a^{(\epsilon)}} \bra{\tau_{\epsilon,\gamma}\, w_a^{(\epsilon)}}. \]
\end{enumerate}
\end{proposition}

The rest of this Section will be devoted to the proof of Proposition~\ref{thm:Main3}. Before diving into the mathematical details, we briefly sketch here the strategy of the case $c_1(P_0)\neq 0$, namely Proposition \ref{thm:Main3}\eqref{item:Main3_i}.
\begin{description}[leftmargin=4em]
	\item[Step 1] By Definition \ref{dfn:Fermi-like} the projection $\Pi^{(0)}$ is a fibered operator in the Bloch--Floquet representation. Therefore, applying Theorem \ref{thm:Main} we can construct an exponentially localized Wannier Parseval frame $\{\tau_{\epsilon=0,\gamma} w_a\}_{a\in\set{1,\ldots,m+1}, \:\gamma \in \Z^2}$ for $\Pi^{(0)}$ (compare Theorem \ref{thm:Main_bis}). Consider now the $m-1$ orthonormal exponentially localized Wannier vectors $\left\{w_s\right\}_{1\leq s \leq m-1}$ together with all their magnetic translates via $\tau_{\epsilon,\gamma}$.  Using gauge covariant magnetic perturbation theory by means of hypothesis \eqref{DiagLoc}, we can prove that the Gram matrix associated to the set $\{\Pi^{(\epsilon)} \tau_{\epsilon, \gamma} w_a\}_{a\in\set{1,\ldots,m-1}, \:\gamma \in \Z^2} $ has matrix elements which decay exponentially away from the diagonal and is close to the identity. Hence we can construct $m-1$ orthogonal exponential localized vectors $\big\{ w_s^{(\epsilon)}\big\}_{1\leq s\leq m-1}$ in the range of $\Pi^{\epsilon}$, such that $\Pi_1^{(\epsilon)} := \sum_{\gamma \in \Z^2} \sum_{s=1}^{m-1} \ket{\tau_{\epsilon,\gamma}\, w_s^{(\epsilon)}} \bra{\tau_{\epsilon,\gamma}\, w_s^{(\epsilon)}}$. 
	\item[Step 2] The second step makes use of the space-dimension-doubling procedure of Section \ref{sec:thmMain(ii)} coupled with gauge covariant magnetic perturbation theory. We consider the operator $T^{(\epsilon)}=T_1^{(\epsilon)}\oplus T_2^{(\epsilon)}$ whose matrix elements are given by the matrix elements of the position-space representation of the projection $P_2(\kk) \oplus (C P_2(\kk) C^{-1})$ times the $\epsilon$-dependent Peierls magnetic phase $\eu^{\iu \epsilon \phi(\cdot,\cdot)}$. $T^{(\epsilon)}$ is not a projection (for $\epsilon\neq0$) but it is an almost projection, namely it is $\epsilon$-close to an actual projection $\mathfrak{P}^{(\epsilon)}=\mathfrak{P}_{1}^{(\epsilon)}
	\oplus \mathfrak{P}_{2}^{(\epsilon)}$. Since the projection $T^{(\epsilon=0)}$ is trivial, we can repeat the procedure described in Step 1 and obtain a basis of exponentially localized Wannier-type functions for the projection $\mathfrak{P}^{(\epsilon)}$. Projecting onto one component of the doubled space and using a Kato--Nagy unitary $K_\epsilon$, we obtain $\Pi_2^{(\epsilon)}=\sum_{\gamma\in \Z^2}\sum_{r=1}^2\ket{\tau_{\epsilon,\gamma}W_{r}^{(\epsilon)}}\bra{\tau_{\epsilon,\gamma}W_{r}^{(\epsilon)}}$, where $\Pi_2^{(\epsilon)}= K_\epsilon \mathfrak{P}_2^{(\epsilon)}K_\epsilon^{-1}$ and $\big\{W_{r}^{(\epsilon)}\big\}_{r \in \{1,2\}}$ are two exponentially localized vectors.
\end{description}

\subsection{Parseval frame at $\epsilon = 0$}

We first look at the projection $\Pi_0 \equiv \Pi^{(\epsilon=0)}$, which can be fibered through the projections $\set{P_0(\kk)}_{\kk \in \R^2}$ as in~\eqref{ZeroPeriodic}. We know from Theorem \ref{thm:Main} that we may find $m-1$ orthonormal vectors $\set{\xi_s(\kk)}_{1 \le s \le m-1}$ in the range of $P_0(\kk)$ which are both $\Z^2$-periodic and real-analytic in $\kk$. Also, we may find two other vectors $\Xi_1(\kk)$ and $\Xi_2(\kk)$ in the range of $P_0(\kk)$ which are $\Z^2$-periodic and real-analytic in $\kk$, so that $\set{\xi_1(\kk),\ldots,\xi_{m-1}(\kk),\Xi_1(\kk),\Xi_2(\kk)}$ forms a Parseval frame for the range of $P_0(\kk)$. This means that we have the following orthogonal decomposition for $P_0(\kk)$ :
\begin{equation}\label{apr2}
P_0(\kk) = P_1(\kk) + P_2(\kk), \quad P_1(\kk) := \sum_{s=1}^{m-1} \ket{\xi_s(\kk)} \bra{\xi_s(\kk)}, \quad P_2(\kk) := \sum_{r=1}^{2} \ket{\Xi_r(\kk)} \bra{\Xi_r(\kk)}. 
\end{equation}
Note that $P_1(\kk) P_2(\kk)=0$ and $P_2(\kk)$ has rank 1. Going back from $\kk$-space to position-space we define the operators $\Pi_j$ acting in $\ell^2(\Z^2) \otimes \C^Q$ having the following matrix elements:
\[ \Pi_j(\gamma,\xx;\gamma',\xx') := \int_\Omega \di\kk \, \eu^{\iu 2\pi \kk \cdot (\gamma-\gamma')} P_j(\kk)(\xx,\xx'), \quad j \in \set{0,1,2} \, , \gamma, \gamma'\in \Z^2\, , \: \xx, \xx' \in \set{1,\ldots,Q} \, . \]

Since $\Pi^{(\epsilon)}$ is a family of Fermi-like magnetic projections, from \eqref{ExpLocIntKernel} we have
\begin{equation} \label{CTP0}
\left|\Pi_0(\gamma,\xx;\gamma',\xx') \right| \leq C \exp^{-\Ll \|\gamma-\gamma'\|} \, .
\end{equation}
Define now the exponentially localized Wannier-type functions
\[ w_s(\gamma,\xx) := \int_\Omega \di\kk \, \eu^{\iu 2\pi \kk\cdot \gamma}\xi_s(\kk,\xx),\quad W_r(\gamma,\xx) := \int_\Omega \di\kk \, \eu^{\iu 2\pi \kk\cdot \gamma}\Xi_r(\kk,\xx), \]
for $s \in \set{1, \ldots, m-1}$ and $r \in \set{1,2}$. Due to the analyticity of the Bloch-type vectors $\xi_s$ and $\Xi_r$, we obtain that \cite{Cloizeaux, Kuchment16}
\begin{equation} \label{ExpoWF0}
\begin{gathered}
\max_{s \in \set{1,\dots,m-1}} \max_{\xx \in \set{1,\ldots,Q}} \left| \eu^{\beta \|\gamma\|}  w_s(\gamma,\xx) \right| \leq C  \, , \\
 \max_{r \in \set{1,2}} \max_{\xx \in \set{1,\ldots,Q}} \left| \eu^{\beta \|\gamma\|}  W_r(\gamma,\xx) \right| \leq C   \, ;
\end{gathered}
\end{equation}
where $\beta$ is less than the width of the strip of analiticity of the Bloch-type frame.
Moreover we have the following identities:
\[ \Pi_1 = \sum_{\gamma\in\Z^2} \sum_{s=1}^{m-1} \ket{\tau_{0,\gamma}\, w_s} \bra{\tau_{0,\gamma} \, w_s}, \quad \Pi_2 = \sum_{\gamma\in\Z^2} \sum_{r=1}^{2} \ket{\tau_{0,\gamma} \, W_r} \bra{\tau_{0,\gamma} \, W_r}, \quad \Pi_0 = \Pi_1 + \Pi_2. \]

\subsection{Construction of $\Pi_1^{(\epsilon)}$}

The next Lemma provides the construction of $\Pi_1^{(\epsilon)}$ as in the statement of Proposition~\ref{thm:Main3}.

\begin{lemma}[Construction of $\Pi_1^{(\epsilon)}$] \label{lemaianuar1}
Define the self-adjoint operator $M_\epsilon$ acting on the space $\ell^2(\Z^2) \otimes \C^{m-1}$ defined by the matrix elements
\begin{equation} \label{eqn:Mepsilon}
M_\epsilon(\gamma,s;\gamma',s') := \scal{\tau_{\epsilon,\gamma} w_s}{\Pi^{(\epsilon)} \tau_{\epsilon,\gamma'} w_{s'}}_{\ell^2(\Z^2)\otimes \C^Q}.
\end{equation}
If $\epsilon_0$ is small enough, then there exists some $C,\alpha>0$ such that uniformly in $\epsilon\leq \epsilon_0$ we have
\begin{equation} \label{ExpoM}
\left| \delta_{ss'} \delta_{\gamma \gamma'} - M_\epsilon(\gamma,s;\gamma',s') \right| \leq C \, \epsilon \, \eu^{-\alpha\|\gamma-\gamma'\|}.
\end{equation}
The vectors
\begin{equation}\label{ianuar21}
V_{\gamma'',s,\epsilon}(\gamma,\xx) := \sum_{\gamma'\in\Z^2} \sum_{s'=1}^{m-1} [M_\epsilon^{-1/2}](\gamma',s';\gamma'',s) [\Pi^{(\epsilon)}\tau_{\epsilon,\gamma'} w_{s'}](\gamma,\xx)
\end{equation}
are orthonormal. Moreover, there exist $m-1$ exponentially localized vectors $w_s^{(\epsilon)}$, $s \in \set{1,\ldots,m-1}$, such that 
\begin{equation} \label{Vepsilon}
V_{\gamma'',s,\epsilon} = \tau_{\epsilon,\gamma''} w_s^{(\epsilon)} \quad \text{for all } \gamma'' \in \Z^2, \: \epsilon \in [0,\epsilon_0).
\end{equation}
Finally
\[ \Pi_1^{(\epsilon)} := \sum_{\gamma \in\Z^2} \sum_{s=1}^{m-1} \ket{V_{\gamma,s,\epsilon}} \bra{V_{\gamma,s,\epsilon}} = \sum_{\gamma \in\Z^2} \sum_{s=1}^{m-1} \ket{\tau_{\epsilon,\gamma} \, w_s^{(\epsilon)}} \bra{\tau_{\epsilon,\gamma} \, w_s^{(\epsilon)}} \]
is an orthogonal projection such that $\Pi^{(\epsilon)} \,  \Pi_1^{(\epsilon)} = \Pi_1^{(\epsilon)}$. 
\end{lemma}

\begin{remark}
When the Chern number of $\set{P_0(\kk)}_{\kk \in \R^2}$ vanishes, then the construction of $\Pi_1^{(\epsilon)}$ provided by Lemma~\ref{lemaianuar1} above can be applied to the whole $\Pi^{(\epsilon)}$, thus proving also Proposition~\ref{thm:Main3}(\ref{item:Main3_ii}).
\end{remark}

\begin{proof}[Proof of Lemma~\ref{lemaianuar1}.]
The proof follows the same ideas as in \cite{CorneanHerbstNenciu16}. In the following the magnetic phase composition rule
\begin{equation}\label{ianuar20}
\phi(\gamma,\gamma')+\phi(\gamma',\gamma'')=\phi(\gamma,\gamma'')+\phi(\gamma-\gamma',\gamma'-\gamma'')
\end{equation}
will be used repeatedly. During the proof we will denote inessential constants by~$K$.

Consider the operator $\widehat{\Pi}^{(\epsilon)}$ defined by the following matrix elements:
\begin{equation} \label{hatPi}
\begin{aligned}
\widehat{\Pi}^{(\epsilon)} (\gamma,\xx;\gamma',\xx')&:= \eu^{\iu \epsilon \phi(\gamma,\gamma')} \Pi_0(\gamma,\xx;\gamma',\xx') =  \eu^{\iu \epsilon \phi(\gamma,\gamma')} \Pi_1(\gamma,\xx;\gamma',\xx')+  \eu^{\iu \epsilon \phi(\gamma,\gamma')} \Pi_2(\gamma,\xx;\gamma',\xx')  \\
&=: \widehat{\Pi}^{(\epsilon)}_1 (\gamma,\xx;\gamma',\xx') + \widehat{\Pi}^{(\epsilon)}_2 (\gamma,\xx;\gamma',\xx') \, ,
\end{aligned}
\end{equation}
and the operator $\widetilde{\Pi}^{(\epsilon)}$ defined by
\begin{equation} \label{tildePi}
\widetilde{\Pi}^{(\epsilon)}:= \sum_{\gamma'' \in \Z^2} \sum_{s=1}^{m-1} | \tau_{\epsilon,\gamma''}w_{s} \rangle \langle  \tau_{\epsilon,\gamma''}w_{s}| + \sum_{\gamma''\in\Z^2} \sum_{r=1}^2  | \tau_{\epsilon,\gamma''}W_{r} \rangle \langle  \tau_{\epsilon,\gamma''}W_{r}| =: \widetilde{\Pi}_1^{(\epsilon)} + \widetilde{\Pi}_2^{(\epsilon)} \,.
\end{equation}
Then we have the following estimate:
\begin{equation} \label{giugno28}
\begin{aligned}
|M_\epsilon(\gamma,s;\gamma',s')- \delta_{ss'}\delta_{\gamma \gamma'}| &\leq \left| \scal{\tau_{\epsilon,\gamma} w_{s}}{\left(\Pi^{(\epsilon)}-\widehat{\Pi}^{(\epsilon)}\right)\tau_{\epsilon,\gamma'} w_{s'}} \right| \\
& \quad+ \left| \scal{\tau_{\epsilon,\gamma} w_{s}}{\left(\widehat{\Pi}^{(\epsilon)}-\widetilde{\Pi}^{(\epsilon)}\right)\tau_{\epsilon,\gamma'} w_{s'}} \right| + \left| \scal{\tau_{\epsilon,\gamma} w_{s}}{\widetilde{\Pi}_2^{(\epsilon)}\tau_{\epsilon,\gamma'} w_{s'}} \right|  \\
& \quad + \left| \scal{\tau_{\epsilon,\gamma} w_{s}}{\widetilde{\Pi}_1^{(\epsilon)}\tau_{\epsilon,\gamma'} w_{s'}} - \delta_{ss'}\delta_{\gamma \gamma'} \right| \, .
\end{aligned}
\end{equation}

The first term of the right-hand side is exponentially localized due to \eqref{DiagLoc} and \eqref{ExpoWF0}. In order to prove the exponential localization of the second term we prove an estimate analogue to \eqref{DiagLoc} for the matrix elements of $\widehat{\Pi}^{(\epsilon)}-\widetilde{\Pi}^{(\epsilon)}$:
\begin{equation} \label{dm28}
\begin{aligned}
\Big| \widehat{\Pi}^{(\epsilon)}_1(\gamma,\xx;\gamma',\xx') & - \widetilde{\Pi}^{(\epsilon)}_1 (\gamma,\xx;\gamma',\xx')\Big| \\
&\leq \sum_{\gamma''\in\Z^2} \sum_{s'=1}^{m-1} \left| \left(1-\eu^{\iu\epsilon \phi(\gamma-\gamma'',\gamma''-\gamma')}\right)w_{s''}(\gamma-\gamma'',\xx) \overline{w_{s''}(\gamma'-\gamma'',\xx')} \right| \\
&\leq \frac{\epsilon}{2} \sum_{\gamma''\in\Z^2} \sum_{s'=1}^{m-1}  \|\gamma-\gamma''\| \|\gamma''-\gamma'\|\left| w_{s''}(\gamma-\gamma'',\xx) \overline{w_{s''}(\gamma'-\gamma'',\xx')} \right| \\
&\leq K \, \epsilon \, \eu^{-\alpha'\|\gamma-\gamma'\|} \, ,
\end{aligned}
\end{equation}
where $\alpha'<\beta$, since we have used \eqref{ExpoWF0}. The same argument works also for the matrix elements of $\widehat{\Pi}^{(\epsilon)}_2-\widetilde{\Pi}^{(\epsilon)}_2$; hence we can conclude that 
\[ \left| \left( \widehat{\Pi}^{(\epsilon)}-\widetilde{\Pi}^{(\epsilon)} \right) (\gamma,\xx;\gamma',\xx')\right| \leq K \, \epsilon \, \eu^{-\alpha' \|\gamma-\gamma'\|} \, . \]
Then the exponential localization also of the second term on the right-hand side of~\eqref{giugno28} follows.
	
Consider now the scalar product $\scal{\tau_{\epsilon,\gamma} w_{s}}{\tau_{\epsilon,\gamma'} W_{r}}$. Since $\tau_{0,\gamma} w_{s}$ and $\tau_{0,\gamma'} W_{r}$ belong to orthogonal subspaces for every $\gamma,\gamma' \in \Z^2$, $s \in \set{1, \ldots, m-1}$ and $r \in \set{1,2}$, we have that
\begin{equation} \label{ImpoEstimates}
\begin{aligned}
| \langle \tau_{\epsilon,\gamma} w_{s}, \, & \tau_{\epsilon,\gamma'}W_{r} \rangle | = \left| \scal{\tau_{\epsilon,\gamma} w_{s}}{\tau_{\epsilon,\gamma'} W_{r}} - \scal{\tau_{0,\gamma} w_{s}}{\tau_{0,\gamma'} W_{r}} \right| \\
&\leq \sum_{\gamma'' \in \Z^2} \sum_{\yy =1}^{Q} \left| \left(\ep{\gamma-\gamma''}{\gamma''-\gamma'}-1\right) \right| \left|\overline{(\tau_{0,\gamma} w_{s})(\gamma'',\yy)} (\tau_{0,\gamma'}W_{r})(\gamma'',\yy)\right| \\
&\leq \eu^{-\alpha\|\gamma-\gamma'\|} \frac{\epsilon}{2} \sum_{\gamma'' \in \Z^2} \sum_{\yy =1}^Q  \eu^{\alpha\|\gamma-\gamma''\|}\|\gamma-\gamma''\| \|\gamma''-\gamma'\|  \eu^{\alpha\|\gamma''-\gamma'\|} \left|\overline{(\tau_{0,\gamma}w_{s})(\gamma'',\yy)} (\tau_{0,\gamma'}W_{r})(\gamma'',\yy)\right| \\
&\leq K \, \epsilon \, \eu^{-\alpha'\|\gamma-\gamma'\|} \, .
\end{aligned}
\end{equation}
The same argument works if we substitute $\tau_{\epsilon,\gamma'}W_{r}$ with $\tau_{\epsilon,\gamma'}w_{s'}$ as long as $\left(\gamma',s'\right) \neq \left(\gamma,s\right)$. Thus we can also prove the exponential localization for the third term on the right-hand side of~\eqref{giugno28}:
\begin{align*}
\left| \scal{\tau_{\epsilon,\gamma} w_{s}}{\widetilde{\Pi}_2\tau_{\epsilon,\gamma'} w_{s'}} \right| & \leq \sum_{\gamma''\in\Z^2} \sum_{r=1}^2 \left|\scal{\tau_{\epsilon,\gamma} w_{s}}{\tau_{\epsilon,\gamma''}W_{r}} \right| \left|\scal{\tau_{\epsilon,\gamma''} W_{r}}{\tau_{\epsilon,\gamma'} w_{s'}} \right| \\
&\leq K' \epsilon^2 \eu^{-\alpha'\|\gamma-\gamma'\|} \sum_{\gamma''\in\Z^2} \sum_{r=1}^2 \eu^{-\beta\|\gamma-\gamma''\|}  \leq K \epsilon^2  \eu^{-\alpha''\|\gamma-\gamma'\|} \, ,
\end{align*}
and for the fourth term as well:
\begin{align*}
|\langle \tau_{\epsilon,\gamma} w_{s}, \, & \widetilde{\Pi}_1\tau_{\epsilon,\gamma'} w_{s'} \rangle - \delta_{ss'}\delta_{\gamma \gamma'} | \\
& \leq \delta_{\gamma \gamma'}\delta_{ss'}\sum_{\substack{(\gamma'',s'')\in \Z^2 \times \set{1,\dots,m-1}\\ (\gamma'',s'')\neq (\gamma,s)}} \left| \scal{\tau_{\epsilon,\gamma} w_{s}}{\tau_{\epsilon,\gamma''} w_{s''}} \right| \, \left| \scal{ \tau_{\epsilon,\gamma''} w_{s''}}{\tau_{\epsilon,\gamma'} w_{s'}} \right| \\
&\quad + \sum_{\substack{(\gamma'',s'')\in \Z^2 \times \set{1,\dots,m-1}}} \left| \scal{\tau_{\epsilon,\gamma} w_{s}}{\tau_{\epsilon,\gamma''} w_{s''}} \right| \, \left| \scal{\tau_{\epsilon,\gamma''} w_{s''}}{\tau_{\epsilon,\gamma'} w_{s'}} \right| \\
&\leq K \epsilon^2  \eu^{-\alpha''\|\gamma-\gamma'\|} \, ,
\end{align*}
with $\alpha''<\alpha'$. Hence \eqref{ExpoM} is proved. 

Define now $D_\epsilon:=M_\epsilon-\Id$. The estimate~\eqref{ExpoM} shows that the norm of $D_\epsilon$ is controlled by $\epsilon$, therefore $M_\epsilon^{-1/2}=\left(\Id +D_\epsilon\right)^{-1/2}$ exists and can be expressed as a norm convergent power series around $\epsilon=0$. Moreover, arguing as in the proof of Lemma~\ref{Nagy}, one can show that $M_\epsilon^{-1/2}$ has exponentially localized matrix elements, and that
\begin{equation} \label{M1/2}
|M_\epsilon^{-1/2}(\gamma,s;\gamma',s')- \delta_{ss'}\delta_{\gamma \gamma'}| \le C \, \epsilon \, \eu^{-\rho \norm{\gamma-\gamma'}}
\end{equation}
for some positive $C, \rho>0$.

The series defined in~\eqref{ianuar21} is now an absolutely convergent series and it is straightforward to check that the vectors $V_{\gamma'',s,\epsilon}$ form an orthonormal set. It remains to prove the existence of $m-1$ exponentially localized vectors $w_s^{(\epsilon)}$ such that \eqref{Vepsilon} holds.	By hypothesis $\Pi^{(\epsilon)}$ satisfies~\eqref{MagCommute}, which together with~\eqref{eqn:PropMagTransl} implies that 
\begin{align*}
M_\epsilon(\gamma,s;\gamma',s') & = \scal{\tau_{\epsilon,\gamma} w_s}{\tau_{\epsilon,\gamma'} \Pi^{(\epsilon)} w_{s'}} = \scal{\tau_{\epsilon,\gamma'}^*\tau_{\epsilon,\gamma} w_s}{\Pi^{(\epsilon)} w_{s'}} \\
& = \scal{\ep{\gamma'}{\gamma} \tau_{\epsilon,\gamma-\gamma'} w_s}{\Pi^{(\epsilon)} w_{s'}} = \ep{\gamma}{\gamma'} \scal{\tau_{\epsilon,\gamma-\gamma'} w_s}{\Pi^{(\epsilon)} w_{s'}} \\
& =: \ep{\gamma}{\gamma'} m_\epsilon(\gamma-\gamma';s,s').
\end{align*}
By using the power series expansion for the inverse square root, one can prove (see~\cite{CorneanHerbstNenciu16}) a similar form for the matrix elements of $M^{-1/2}_\epsilon$, namely
\[ [M^{-1/2}_\epsilon](\gamma,s;\gamma',s')=: \ep{\gamma}{\gamma'} m_{\epsilon,-1/2}(\gamma-\gamma';s,s') \;. \]
It then follows, using again~\eqref{eqn:PropMagTransl} and~\eqref{MagCommute}, that
\begin{align*}
V_{\gamma'',s,\epsilon} & = \sum_{\gamma'\in\Z^2} \sum_{s'=1}^{m-1} \ep{\gamma'}{\gamma''} m_{\epsilon,-1/2}(\gamma'-\gamma'';s',s)  [\tau_{\epsilon,\gamma'}\Pi^{(\epsilon)} w_{s'}] \\
& = \sum_{\gamma=\gamma'-\gamma''\in\Z^2} \sum_{s'=1}^{m-1} \ep{\gamma+\gamma''}{\gamma''} m_{\epsilon,-1/2}(\gamma;s',s)  [\tau_{\epsilon,\gamma+\gamma''}\Pi^{(\epsilon)} w_{s'}] \\
& = \sum_{\gamma\in\Z^2} \sum_{s'=1}^{m-1} \ep{\gamma}{\gamma''} m_{\epsilon,-1/2}(\gamma;s',s)  [\ep{\gamma''}{\gamma} \tau_{\epsilon,\gamma''} \tau_{\epsilon,\gamma}\Pi^{(\epsilon)} w_{s'}] \\
& = \tau_{\epsilon,\gamma''} w_{s}^{(\epsilon)}, 
\end{align*}
with 
\begin{equation} \label{wepsilons}
w^{(\epsilon)}_s := \sum_{\gamma\in\Z^2} \sum_{s'=1}^{m-1} m_{\epsilon,-1/2}(\gamma;s',s)  [\tau_{\epsilon,\gamma}\Pi^{(\epsilon)} w_{s'}] = \sum_{\gamma\in\Z^2} \sum_{s'=1}^{m-1} [M_\epsilon^{-1/2}](\gamma,s';0,s) [\tau_{\epsilon,\gamma} \Pi^{(\epsilon)} w_{s'}] \,.
\end{equation}

Due to the exponential localization \eqref{ExpoWF0} of the $w_s$'s and of the matrix elements of $M^{-1/2}_\epsilon$ and $\Pi^{(\epsilon)}$, we easily get that there exist $\beta',C>0$, independent of $\xx$, such that
\begin{equation} \label{ExpoWFeps}
\sup_{\gamma \in \Z^2} \eu^{\beta'\|\gamma\|} |w^{(\epsilon)}_s(\gamma,\xx)| \leq C \, . \qedhere
\end{equation}
\end{proof}

\begin{remark} \label{rmk:weps}
Notice that the functions $w_s^{(\epsilon)}$ defined in~\eqref{wepsilons} satisfy
\begin{align*}
w_s^{(\epsilon)} - w_{s'}  & = w_s^{(\epsilon)} - [\Pi_0 w_{s'}] = w_s^{(\epsilon)} - [\widehat{\Pi}^{(\epsilon)} w_{s'}] + [(\widehat{\Pi}^{(\epsilon)}-\Pi_0) w_{s'}] \\
& = \sum_{\eta\in\Z^2} \sum_{s''=1}^{m-1} [M_\epsilon^{-1/2}](\eta,s'';0,s)  [(\Pi^{(\epsilon)} - \delta_{0\eta} \delta_{s's''} \widehat{\Pi}^{(\epsilon)})  \tau_{\epsilon,\eta}]w_{s''} + [(\widehat{\Pi}^{(\epsilon)}-\Pi_0) w_{s'}]
\end{align*}
hence
\begin{align*}
w_s^{(\epsilon)} & (\gamma,\xx)- w_{s'}(\gamma,\xx) =\\
& = \sum_{\gamma'\in\Z^2} \sum_{\xx'=1}^{Q} [M_\epsilon^{-1/2}](0,s';0,s)  [\Pi^{(\epsilon)}(\gamma,\xx;\gamma',\xx') - \ep{\gamma}{\gamma'}\Pi_0(\gamma,\xx;\gamma',\xx')] w_{s'}(\gamma',\xx') \\
& \quad + \sum_{\eta\ne 0} \sum_{s''\ne s'} \left([M_\epsilon^{-1/2}](\eta,s'';0,s) - \delta_{0\eta}\delta_{s's''} \right) \sum_{\gamma'\in\Z^2}\sum_{\xx'=1}^{Q} \Pi^{(\epsilon)}(\gamma,\xx;\gamma',\xx') \, \ep{\gamma'}{\eta} w_{s''}(\gamma'-\eta,\xx')\\
& \quad + \sum_{\gamma'\in\Z^2} \sum_{\xx'=1}^{Q} (\ep{\gamma}{\gamma'}-1)\Pi_0(\gamma,\xx;\gamma',\xx') w_{s'}(\gamma',\xx') \,.
\end{align*}
Using~\eqref{DiagLoc} to estimate the first sum on the right-hand side of the above, \eqref{M1/2} with~\eqref{ExpLocIntKernel} for the second sum, and the power series of the exponential and~\eqref{ExpLocIntKernel} at $\epsilon = 0$ for the third sum, together with the exponential decay~\eqref{ExpoWF0} of the functions $w_s$, we are able to deduce that for all $s,s' \in \set{1, \ldots, m-1}$
\[
\left| w_s^{(\epsilon)}(\gamma,\xx)- w_{s'}(\gamma,\xx) \right| \le C \, \epsilon \, \eu^{-\sigma \norm{\gamma}}
\]
for some positive constants $C,\sigma > 0$ uniform in $\gamma \in \Z^2$, $\xx \in \set{1, \ldots, Q}$, and $\epsilon$ sufficiently small. Clearly the above implies in turn that
\begin{equation} \label{wepsi}
\left| [\tau_{\epsilon,\eta} w_s^{(\epsilon)}](\gamma,\xx)- [\tau_{\epsilon,\eta}w_{s'}](\gamma,\xx) \right| \le C \, \epsilon \, \eu^{-\sigma \norm{\gamma-\eta}}
\end{equation}
for any $\eta \in \Z^2$.
\end{remark}

\subsection{Construction of $\Pi_2^{(\epsilon)}$}

We now show that the orthogonal projection defined as
\begin{equation}\label{hc1''}
\Pi_2^{(\epsilon)}:=\Pi^{(\epsilon)}_{\phantom{2}}-\Pi_1^{(\epsilon)}
\end{equation} 
can be written as in Proposition \ref{thm:Main3}(\ref{item:Main3_i}). To construct a Parseval frame for it we will mix the space-dimension-doubling method used in the proof of Theorem\ref{thm:Main}\eqref{Main_3} (see Section~\ref{sec:thmMain(ii)}) with magnetic perturbation theory. 

The projection $P_2(\kk)$ in \eqref{apr2} has rank 1, and we introduce $\widetilde{P}_2(\kk):=C P_2(-\kk)C^{-1}$ as in~\eqref{Q=CPC}. Denote by $P_3(\kk):=\widetilde{P}_2(\kk) \oplus P_2(\kk)$ the rank-2 projection acting on $\C^Q\oplus \C^Q$. As it is argued in Section \ref{sec:thmMain(ii)}, its Chern number is zero. We now define the operator 
\[ T^{(\epsilon)}(\gamma,\xx;\gamma',\xx'):=\eu^{\iu \epsilon \phi(\gamma,\gamma')}\int_\Omega \eu^{\iu 2\pi \kk \cdot (\gamma-\gamma')}P_3(\kk)(\xx,\xx')d\kk,\quad \gamma,\gamma' \in \Z^2, \: \xx,\xx'\in \set{1,\dots,2Q}, \]
acting on $\ell^2(\Z^2)\otimes(\C^Q\oplus \C^Q)$. Note the fact that $T^{(\epsilon)}=T_1^{(\epsilon)}\oplus T_2^{(\epsilon)}$ where $T_j^{(\epsilon)}$ acts on $\ell^2(\Z^2)\otimes \C^Q$. We also have
\begin{equation} \label{T2}
\begin{aligned}
T_2^{(\epsilon)}(\gamma,\xx;\gamma',\xx')&:=\eu^{\iu \epsilon \phi(\gamma,\gamma')}\int_\Omega \eu^{\iu 2\pi \kk \cdot (\gamma-\gamma')}P_2(\kk)(\xx,\xx')d\kk\\
&=\eu^{\iu \epsilon \phi(\gamma,\gamma')}\Pi_2(\gamma,\xx;\gamma',\xx')=\widehat{\Pi}^{(\epsilon)}_2(\gamma,\xx;\gamma',\xx') \,,
\end{aligned}
\end{equation}
for $\gamma,\gamma' \in \Z^2$ and $\xx,\xx'\in \set{1,\dots,Q}$, compare~\eqref{hatPi}.

The operators $T_j^{(\epsilon)}$ are almost orthogonal projections, in the sense that 
\[ \Delta_j^{(\epsilon)}:= \left\{ \left(T_j^{(\epsilon)}\right)^2-T_j^{(\epsilon)} \right\}=\mathcal{O}(\epsilon) \] 
in the operator norm; more precisely we want to prove an estimate of the usual type for the matrix elements of $\Delta_j^{(\epsilon)}$, namely
\begin{equation}\label{hc11}
\left| \Delta_j^{(\epsilon)}(\gamma,\xx;\gamma',\xx') \right|\leq C \;\epsilon \; \eu^{-\alpha \|\gamma-\gamma'\|}.
\end{equation}
In order to show this, notice first that $T_j^{(0)}$ is a true projection and hence $\Delta_j^{(0)}=0$. Then, using the magnetic phase composition rule \eqref{ianuar20} to compute the matrix elements of $(T_j^{(\epsilon)})^2$ and noticing that 
\[ \left|T_j^{(\epsilon)}(\gamma;\xx,\xx')\right| \leq C \eu^{-\delta \|\gamma\|}, \] 
we obtain
\[ \left|\Delta_j^{(\epsilon)}(\gamma,\xx;\gamma',\xx')\right|=\left|\Delta_j^{(\epsilon)}(\gamma,\xx;\gamma',\xx')-\Delta_j^{(0)}(\gamma,\xx;\gamma',\xx')\right| \leq C \;\epsilon \; \eu^{-\alpha \|\gamma-\gamma'\|} \]
as wanted.

Now if $\epsilon$ is small enough, we may construct the following explicit orthogonal projections acting on $\ell^2(\Z^2)\otimes \C^Q$ (see \cite{Nenciu02} for more details):
\[ \mathfrak{P}_{j}^{(\epsilon)} := T_j^{(\epsilon)} + \left(T_j^{(\epsilon)} - \frac{1}{2} {\bf 1} \right) \left\{ ({\bf 1}+4\Delta_j^{(\epsilon)})^{-1/2}-{\bf 1} \right\}. \]
Simply using the above formula in each term of the direct sum in the expression for $T^{(\epsilon)}$ we obtain that 
\[ \mathfrak{P}^{(\epsilon)}=\mathfrak{P}_{1}^{(\epsilon)}
\oplus \mathfrak{P}_{2}^{(\epsilon)} \]
is an orthogonal projection acting on the ``doubled'' space and, using~\eqref{hc11} and arguing as in the proof of Lemma~\ref{Nagy}, that
\begin{equation}\label{hc13}
\left| \mathfrak{P}^{(\epsilon)}(\gamma,\xx;\gamma',\xx') -T^{(\epsilon)}(\gamma,\xx;\gamma',\xx') \right| \leq C \;\epsilon \; \eu^{-\alpha \|\gamma-\gamma'\|}.
\end{equation}

Because the non-magnetic projection $P_3(\kk)$ which builds up $T^{(\epsilon=0)}$ is trivial, mimicking the proof of Lemma~\ref{lemaianuar1} we infer that we can construct two exponentially localized Wannier-type vectors $F_r^{(\epsilon)}\in \ell^2(\Z^2)\otimes (\C^Q\oplus\C^Q)$ such that: 
\[ \mathfrak{P}^{(\epsilon)}=\sum_{\gamma\in \Z^2}\sum_{r=1}^2 \ket{(\tau_{\epsilon,\gamma} \oplus \tau_{\epsilon,\gamma}) F_r^{(\epsilon)}} \bra{(\tau_{\epsilon,\gamma} \oplus \tau_{\epsilon,\gamma}) F_r^{(\epsilon)}}, \]
where $\tau_{\epsilon,\gamma} \oplus \tau_{\epsilon,\gamma}$ is the obvious extension of the magnetic translation to the doubled space. 
Restricting ourselves to vectors of the type $0 \oplus \psi$ where $\psi$ is in the range of $\mathfrak{P}_2^{(\epsilon)}$, and denoting by $\pi_2 \colon \ell^2(\Z^2)\otimes (\C^Q\oplus \C^Q)\to 
\ell^2(\Z^2)\otimes \C^Q$ the projection on the second component of the doubled space, we have the identity
\[ \psi = \pi_2( 0\oplus\psi ) =  \sum_{\gamma\in \Z^2}\sum_{r=1}^{2}\scal{\tau_{\epsilon,\gamma}(\pi_2  F_r^{(\epsilon)})}{\psi}_{\ell^2(\Z^2)\otimes \C^Q} \left(\tau_{\epsilon,\gamma} (\pi_2 F_r^{(\epsilon)}) \right). \]
In other words this means that
\begin{equation} \label{frakP2}
\mathfrak{P}_{2}^{(\epsilon)} = \sum_{\gamma\in \Z^2}\sum_{r=1}^2 \ket{\tau_{\epsilon,\gamma}(\pi_2 F_r^{(\epsilon)})}\bra{\tau_{\epsilon,\gamma}(\pi_2 F_r^{(\epsilon)})}.
\end{equation}

The next important step consists of the following estimate that we state as a Lemma. 
\begin{lemma}
There exist constants $\epsilon_0,\alpha, C > 0$ such that for every $0\le\epsilon \leq \epsilon_0$ it holds that
\begin{equation} \label{hc15}
\left| \Pi_2^{(\epsilon)}(\gamma,\xx;\gamma',\xx') - T_2^{(\epsilon)}(\gamma,\xx;\gamma',\xx') \right| \leq C \;\epsilon \; \eu^{-\alpha \|\gamma-\gamma'\|} \,.
\end{equation}
\end{lemma}
\begin{proof}
Considering the fact that $\Pi_2^{(\epsilon)}$ is defined as in \eqref{hc1''} and the equalities~\eqref{hatPi} and~\eqref{T2} hold, we have
\begin{align*}
\Big| \Pi_2^{(\epsilon)}(\gamma,\xx;\gamma',\xx') &-  T_2^{(\epsilon)}(\gamma,\xx;\gamma',\xx') \Big| \\
& \leq \left| \Pi^{(\epsilon)}(\gamma,\xx;\gamma',\xx') - \widehat{\Pi}^{(\epsilon)}(\gamma,\xx;\gamma',\xx') \right| + \left| \Pi_1^{(\epsilon)}(\gamma,\xx;\gamma',\xx') - \widehat{\Pi}_1^{(\epsilon)}(\gamma,\xx;\gamma',\xx') \right|\\
&\leq C \;\epsilon \; \eu^{-\alpha \|\gamma-\gamma'\|} + \left| \Pi_1^{(\epsilon)}(\gamma,\xx;\gamma',\xx') - \widehat{\Pi}_1^{(\epsilon)}(\gamma,\xx;\gamma',\xx') \right| \,,
\end{align*}
which is a consequence of the estimate \eqref{DiagLoc}. Therefore it suffices to prove that the matrix elements of $\Pi_1^{(\epsilon)}-\widehat{\Pi}^{(\epsilon)}_{1}$ are exponentially localized and proportional to $\epsilon$. Since we have proved that this is true for the difference $\widetilde{\Pi}_1^{(\epsilon)}-\widehat{\Pi}^{(\epsilon)}_{1}$, where $\widetilde{\Pi}_1^{(\epsilon)}$ is defined in~\eqref{tildePi} (see~\eqref{dm28}), it suffices to prove the required estimate on the matrix elements of the difference $\Pi_1^{(\epsilon)}-\widetilde{\Pi}^{(\epsilon)}_1$. Since the $(\gamma,\xx;\gamma',\xx')$-matrix element of this difference is provided by a difference of absolutely convergent series, we can estimate
\begin{align*}
\Big| \Pi_1^{(\epsilon)}&(\gamma,\xx;\gamma',\xx')-\widetilde{\Pi}^{(\epsilon)}_1(\gamma,\xx;\gamma',\xx') \Big| \\
& = \left| \sum_{\eta \in \Z^2} \sum_{s=1}^{m-1} [\tau_{\epsilon,\eta} w_{s}^{(\epsilon)}](\gamma,\xx) \, \overline{[\tau_{\epsilon,\eta} w_{s}^{(\epsilon)}](\gamma',\xx')} - [\tau_{\epsilon,\eta} w_{s}](\gamma,\xx) \, \overline{[\tau_{\epsilon,\eta} w_{s}](\gamma',\xx')} \right| \\
& \le \sum_{\eta \in \Z^2} \sum_{s=1}^{m-1} \Big\{ \left| [\tau_{\epsilon,\eta} w_{s}^{(\epsilon)}](\gamma,\xx) - [\tau_{\epsilon,\eta} w_{s}](\gamma,\xx) \right| \, \left|[\tau_{\epsilon,\eta} w_{s}^{(\epsilon)}](\gamma',\xx') \right| \\
& \qquad + \left|[\tau_{\epsilon,\eta} w_{s}](\gamma,\xx) \right| \, \left| [\tau_{\epsilon,\eta} w_{s}^{(\epsilon)}](\gamma',\xx') - [\tau_{\epsilon,\eta} w_{s}](\gamma',\xx') \right| \Big\}.
\end{align*}
In view of~\eqref{wepsi} and of the exponential localization~\eqref{ExpoWF0} and~\eqref{ExpoWFeps} of $w_s$ and $w_{s}^{(\epsilon)}$, the conclusion follows.
\end{proof}

Coupling \eqref{hc15} with \eqref{hc13} we obtain 
\[ \left|\Pi_2^{(\epsilon)}(\gamma,\xx;\gamma',\xx') - \mathfrak{P}_2^{(\epsilon)}(\gamma,\xx;\gamma',\xx')\right|\leq C \;\epsilon \; \eu^{-\alpha \|\gamma-\gamma'\|}. \]
Then, if $\epsilon$ is small enough, the above implies that the two projections $\Pi_2^{(\epsilon)}$ and $\mathfrak{P}_2^{(\epsilon)}$ are unitarily equivalent through a Kato--Nagy unitary $K_\epsilon$ given as in~\eqref{eqn:KatoNagy}, i.e. $\Pi_2^{(\epsilon)}= K_\epsilon \mathfrak{P}_2^{(\epsilon)}K_\epsilon^{-1}$. Therefore we have the following proposition that concludes the construction of $\Pi_2^{(\epsilon)}$.

\begin{proposition}
There exist two exponentially localized functions $W_{r}^{(\epsilon)}$, $r\in \set{1,2}$, such that
\[ \Pi_2^{(\epsilon)}=\sum_{\gamma\in \Z^2}\sum_{r=1}^2\ket{\tau_{\epsilon,\gamma}W_{r}^{(\epsilon)}}\bra{\tau_{\epsilon,\gamma}W_{r}^{(\epsilon)}}.\]
\end{proposition}
\begin{proof}
By hypothesis $\Pi^{(\epsilon)}$ commutes with the magnetic translations and by construction also $\Pi_1^{(\epsilon)}$ does; it follows that so does $\Pi_2^{(\epsilon)}$ by~\eqref{hc1''}.
From~\eqref{frakP2}, it is also clear that $\mathfrak{P}_2^{(\epsilon)}$  commutes with the magnetic translations, and by~\eqref{eqn:KatoNagy} so does the Kato-Nagy unitary $K_\epsilon$. Setting $W_r := K_\epsilon (\pi_2 F_r^{(\epsilon)})$, $r \in \set{1,2}$ (compare~\eqref{frakP2}), the proof is concluded like that of Lemma~\ref{lemmaKN}.
\end{proof}

\begin{remark} \label{ExtensionToContinuous}
The results presented in this Section can be extended with not much effort to continuous families of magnetic Fermi-like projections $\Pi^{(\epsilon)}$ acting in $L^2(\R^2) \otimes \Hi$, where $\Hi = L^2((0,1)^2)$. However, in order to apply the construction of the Parseval frames in the framework of continuous magnetic Schr\"odinger operators with constant magnetic field, it is necessary to prove an analogue of Proposition \ref{ApplicationMPF}. In the continuous case the situation is more complicated and one has to fully exploit magnetic perturbation theory \cite{Nenciu02,Cornean,CorneanMonacoMoscolari18} and use some involved technical results regarding elliptic regularity and Agmon--Combes--Thomas uniform exponential decay estimates \cite{CorneanNenciu}. 
\end{remark}


\appendix

\section{``Black boxes''}

In this Appendix we will provide more details and appropriate references for a number of tools and ``black boxes'' employed in the paper.

\subsection{Smoothing argument} \label{app:Smoothing} 

We start by providing a smoothing argument that allows to produce \emph{real-analytic} Bloch vectors from continuous ones.

\begin{lemma}[Smoothing argument] \label{lemma:smoothing}
Let $\set{P(\kk)}_{\kk \in \R^d}$ be a family of orthogonal projections admitting an analytic, $\Z^d$-periodic analytic extension to a complex strip around $\R^d \subset \C^d$. Assume that there exist continuous, $\Z^d$-periodic, and orthogonal Bloch vectors $\set{\xi_1,..., \xi_m}$ for $\set{P(\kk)}_{\kk \in \R^d}$. Then, there exist also \emph{real-analytic}, $\Z^d$-periodic, and orthogonal Bloch vectors $\set{\widehat{\xi}_1,..., \widehat{\xi}_m}$.

The same holds true if analyticity is replaced by $C^r$-smoothness for some $r \in \N \cup \set{\infty}$.
\end{lemma}
\begin{proof}[Proof (sketch)]
We sketch here the proof: more details can be found in \cite[Sec.~2.3]{CorneanHerbstNenciu16}.

Define
\[ g(\kk) = g(k_1, \ldots, k_d) := \frac{1}{\pi^d} \prod_{j=1}^{d} \frac{1}{1+k_j^2}. \]
The function $g$ is analytic over the strip $\set{\mathbf{z}=(z_1, \ldots, z_d) \in \C^d : \left| \mathrm{Im} \, z_j \right| < 1, \: j \in \set{1,\ldots,d}}$ and obeys $\int_{\R^d} g(\kk) \, \di \kk = 1$. For $\delta > 0$, define $g_\delta(\kk) := \delta^{-d} g(\kk/\delta)$. Set
\[ \psi_a^{(\delta)}(\kk) := \int_{\R^d} g_\delta(\kk - \kk') \, \xi_a(\kk') \, \di \kk', \quad a \in \set{1,...,m}, \: \kk \in \R^d. \]
The above define $\Z^d$-periodic vectors which admit an analytic extension to a strip of half-width $\delta$ around the real axis in $\C^d$, and moreover converge to $\xi_a$ uniformly as $\delta \to 0$. We note here that an alternative way of smoothing has been suggested to us by G.~Panati: he proposed taking the convolution with the Fej{\'e}r kernel, which has the advantage of  integrating on $[-1/2,1/2]^d$ and not on the whole $\R^d$.

Now denote $\phi_a^{(\delta)}(\kk) := P(\kk) \, \psi_a^{(\delta)}(\kk)$, for $a \in \set{1,...,m}$ and $\kk \in \R^d$. Then for any $\epsilon > 0$ there exists $\delta > 0$ such that $\phi_a^{(\delta)}(\kk)$ and $\xi_a(\kk)$ are uniformly at a distance less then $\epsilon$. Moreover, as the $\xi_a$'s are orthogonal, we can make sure that the Gram--Schmidt matrix $S^{(\delta)}(\kk)_{ab} := \scal{\phi_a^{(\delta)}(\kk)}{\phi_b^{(\delta)}(\kk)}$ is close to the identity matrix, uniformly in $\kk$, possibly at the price of choosing an even smaller $\delta$. This implies that $S^{(\delta)}(\kk)^{-1/2}$ is real-analytic and $\Z^d$-periodic, and hence the vectors
\[ \widehat{\xi}_a(\kk) := \sum_{b=1}^{m} \phi_b^{(\delta)}(\kk) \, \left[S^{(\delta)}(\kk)^{-1/2}\right]_{ba} \]
define the required real-analytic, $\Z^d$-periodic, and orthogonal Bloch vectors.
\end{proof}

\subsection{Parallel transport} \label{app:Parallel} 

We recall here the definition of \emph{parallel transport} associated to a smooth and $\Z^d$-periodic family of projections $\set{P(k_1,\ldots,k_d)}_{(k_1,\ldots,k_d) \in \R^d}$ acting on an Hilbert space $\Hi$.

Fix $i \in \set{1,\ldots, d}$. For $(k_1,\ldots,k_d) \in \R^d$, denote by $\kk \in \R^D$, $D=d-1$, the collection of coordinates different from the $i$-th. We use the shorthand notation $(k_1,\ldots,k_d) \equiv (k_i,\kk)$ throughout this Subsection. 

Define
\begin{equation} \label{eqn:generator}
A_{\kk}(k_i) := \iu \left[ \partial_{k_i} P(k_i, \kk), P(k_i, \kk) \right], \quad (k_i,\kk) \in \R^d.
\end{equation}
Then $A_{\kk}(k_i)$ defines a self-adjoint operator on $\Hi$. The solution to the operator-valued Cauchy problem
\begin{equation} \label{eqn:parallel_def}
\iu \, \partial_{k_i} T_{\kk}(k_i, k_i^0) = A_{\kk}(k_i) \, T_{\kk} (k_i, k_i^0), \quad T_{\kk}(k_i^0, k_i^0) = \Id,
\end{equation}
defines a family of unitary operators on $\Hi$, called the \emph{parallel transport unitaries} (along the $i$-th direction). In the following we will fix $k_i^0=0$. This notion coincides with the one in differential geometry of the parallel transport along the straight line from $(0,\kk)$ to $(k_i,\kk)$ associated to the \emph{Berry connection} on the Bloch bundle. The parallel transport unitaries satisfy the properties listed in the following result.

\begin{lemma} \label{prop:parallel}
Let $\set{P(\kk)}_{\kk \in \R^d}$ be a smooth (respectively analytic) and $\Z^d$-periodic family of orthogonal projections acting on an Hilbert space $\Hi$. Then the family of parallel transport unitaries $\set{T_{\kk}(k_i, 0)}_{k_i \in \R, \: \kk \in \R^D}$ defined in \eqref{eqn:parallel_def} satisfies the following properties:
\begin{enumerate}
 \item \label{item:T-smooth} the map $\R^d \ni \kk = (k_i, \kk) \mapsto T_{\kk}(k_i, 0) \in \mathcal{U}(\Hi)$ is smooth (respectively real-analytic);
 \item \label{item:T-periodic} for all $k_i \in \R$ and $\kk \in \R^D$ 
 \[ T_{\kk}(k_i+1, 1) = T_{\kk}(k_i, 0) \]
 and
 \[ T_{\kk + \mathbf{n}}(k_i, 0) = T_{\kk}(k_i, 0) \quad  \text{for } \mathbf{n} \in \Z^D; \]
 \item \label{item:T-intertwine} the intertwining property
 \[ P(k_i, \kk) = T_{\kk}(k_i, 0) \, P(0, \kk) \, T_{\kk}(k_i, 0)^{-1} \]
 holds for all $k_i \in \R$ and $\kk \in \R^D$ .
\end{enumerate}
\end{lemma}

A proof of all these properties can be found for example in \cite{FreundTeufel16} or in \cite[Sec.~2.6]{CorneanHerbstNenciu16}.

\bigskip

In \eqref{alpha}, the parallel transport unitary $\mathcal{T}(\kk) := T_{\kk}(1,0)$ is employed to define the continuous, $\Z^D$-periodic family of unitary matrices $\set{\alpha(\kk)}_{\kk \in \R^D}$. Let $j \in \set{1, \ldots, d}$, $j \ne i$. The integrand in the formula \eqref{deg} for $\deg_j(\det \alpha)$ can be expressed in terms of the parallel transport unitaries as
\[ \tr_{\C^m} \left( \alpha(\kk)^* \partial_{k_j} \alpha(\kk) \right) = \Tr_{\Hi} \left( P(0,\kk) \, \mathcal{T}(\kk)^* \partial_{k_j} \mathcal{T}(\kk) \right) \]
(compare \cite[Lemma~6.1]{CorneanMonacoTeufel17}). Besides, by the Duhamel formula we have
\[ \partial_{k_j} T_{\kk}(k_i,0) = T_{\kk}(k_i,0) \, \int_{0}^{k_i} \di s \, T_{\kk}(s,0)^* \, \partial_{k_j} A_{\kk}(s) \, T_{\kk}(s,0), 	\]
where $A_{\kk}(s)$ is as in \eqref{eqn:generator} (compare \cite[Lemma~6.2]{CorneanMonacoTeufel17}). On the other hand, one can also compute
\[ P(k_i,\kk) \, \partial_{k_j} A_{\kk}(k_i) \, P(k_i,\kk) = P(k_i,\kk) \, [ \partial_{k_i} P(k_i,\kk), \partial_{k_j} P(k_i,\kk) ] \, P(k_i,\kk) \]
so that, denoting $\mathbf{K} := (k_i,\kk) \in \R^d$,
\[ \Tr_{\Hi} \left( P(0,\kk) \, \mathcal{T}(\kk)^* \partial_{k_j} \mathcal{T}(\kk) \right) = \int_{0}^{1} \di k_i \, \Tr_{\Hi} \left( P(\mathbf{K}) \, \left[ \partial_{k_i} P(\mathbf{K}), \partial_{k_j} P(\mathbf{K}) \right] \right) \]
(compare \cite[Eqn.~(6.13)]{CorneanMonacoTeufel17}). Putting all the above equalities together, we conclude that
\[ \deg_j(\det \alpha) = \frac{1}{2 \pi \iu} \int_{0}^{1} \di k_j \int_{0}^{1} \di k_i \, \Tr_{\Hi} \left( P(\mathbf{K}) \, \left[ \partial_{k_i} P(\mathbf{K}), \partial_{k_j} P(\mathbf{K}) \right] \right) = c_1(P)_{ij}, \]
see \eqref{Chern}. The above equality proves Proposition~\ref{thm:alphaChern} as well as Equation~\eqref{alphaChern_ij}.

\subsection{Cayley transform} \label{app:Cayley} 

An essential tool to produce ``good'' logarithms for families of unitary matrices which inherit properties like continuity and ($\gamma$-)periodicity is the \emph{Cayley transform}. We recall here this construction.

\begin{lemma}[Cayley transform] \label{prop:Cayley}
Let $\set{\alpha(\kk)}_{\kk \in \R^{D}}$ be a family of unitary matrices which is continuous and $\Z^{D}$-periodic. Assume that $-1$ lies in the resolvent set of $\alpha(\kk)$ for all $\kk \in \R^{D}$. Then one can \emph{construct} a family $\set{h(\kk)}_{\kk \in \R^{D}}$ of self-adjoint matrices which is continuous, $\Z^{D}$-periodic and such that 
\[ \alpha(\kk) = \eu^{\iu  \, h(\kk)} \quad \text{for all } \kk \in \R^{D}. \]

If $D=2$ and $\set{\alpha(k_2,k_3)}_{(k_2,k_3) \in \R^2}$ is $\gamma$-periodic (in the sense of Definition~\ref{def:gamma}), then the above family of self-adjoint matrices can be chosen to be $\gamma$-periodic as well.
\end{lemma}
\begin{proof}
The proof adapts the one in \cite[Prop.~3.5]{CorneanMonacoTeufel17}. The Cayley transform
\[ s(\kk) := \iu \, \left( \Id - \alpha(\kk) \right) \, \left( \Id + \alpha(\kk) \right)^{-1} \]
is self-adjoint, depends continuously on $\kk$, and is $\Z^D$-periodic (respectively $\gamma$-periodic) if $\alpha$ is as well. One also immediately verifies that
\[ \alpha(\kk) = \left( \Id + \iu \, s(\kk) \right) \, \left( \Id - \iu \, s(\kk) \right)^{-1}. \]

Let $\mathcal{C}$ be a closed, positively-oriented contour in the complex plane which encircles the real spectrum of $s(\kk)$ for all $\kk \in \R^D$. Let $\log(\cdot)$ denote the choice of the complex logarithm corresponding to the branch cut on the negative real semi-axis. Then
\[ h(\kk) := \frac{1}{2 \pi} \oint_\mathcal{C} \log \left( \frac{1 + \iu \, z}{1 - \iu \, z} \right) \left( s(\kk) - z \Id \right)^{-1} \, \di z, \quad \kk \in \R^D, \]
obeys all the required properties.
\end{proof}

\subsection{Generically non-degenerate spectrum of families of unitary matrices} \label{app:NonDegenerate} 

The aim of this Subsection is to prove that
\begin{proposition} \label{prop:generic}
Let $D \le 2$. Consider a continuous and $\Z^D$-periodic family of unitary matrices $\set{\alpha(\kk)}_{\kk \in \R^D}$. Then, one can \emph{construct} a sequence of continuous, $\Z^D$-periodic families of unitary matrices $\set{\alpha_n(\kk)}_{\kk \in \R^D}$, $n \in \N$, such that
\begin{itemize}
 \item $\sup_{\kk \in \R^D} \norm{\alpha_n(\kk) - \alpha(\kk)} \to 0$ as $n \to \infty$, and
 \item the spectrum of $\alpha_n(\kk)$ is completely non-degenerate for all $n \in \N$ and $\kk \in \R^D$.
\end{itemize}

In $D=2$, the same conclusion holds if periodicity and homotopy are replaced by $\gamma$-periodicity and $\gamma$-homotopy, in the sense of Definition~\ref{def:gamma}.
\end{proposition}

The periodic case for $D\leq 2$ has already been treated in  
\cite{CorneanHerbstNenciu16}, \cite{CorneanMonacoTeufel17} and \cite{CorneanMonaco17}, but we will sketch below the main ideas and give details on the new, $\gamma$-periodic situation.  

We will need two technical results, which we state here.

\begin{lemma}[Analytic Approximation Lemma]
Consider a uniformly continuous family of unitary matrices $\alpha(k)$ where $k\in [a,b]\subset \R$. Let $I$ be any compact set completely included in $[a,b]$. 
Then one can \emph{construct} a sequence $\set{\alpha_n(k)}_{k\in I}$, $n \in \N$, of families of unitary matrices which are \emph{real-analytic} on $I$ and such that
\[ \sup_{k\in I} \norm{\alpha_n(k) - \alpha(k)} \to 0 \quad \text{as } n \to \infty. \]
If $\alpha$ is continuous and $\mathbb{Z}$-periodic, the same is true for $\alpha_n$ and the approximation is uniform on $\R$. This last statement can be extended to any $D\geq 1$.
\end{lemma}
\begin{proof}[Proof (sketch)]
The proof proceeds in the same spirit of Lemma~\ref{lemma:smoothing} above. First, we take the convolution with a real-analytic kernel and obtain a smooth matrix $\beta(k)$ which is close in norm to $\alpha(k)$. Thus $\kappa:=\beta^*\beta$ must be close to the identity matrix, it is self-adjoint and real-analytic, and the same holds true for $\kappa^{1/2}$. Finally, we restore unitarity by writing $\alpha':=\beta \kappa^{1/2}$ and checking that $(\alpha')^*\alpha'=\mathbf{1}$. More details can be found in \cite[Lemma~A.2]{CorneanMonacoTeufel17}.
\end{proof}

\begin{lemma}[Local Splitting Lemma]
For $R > 0$ and $\kk_0 \in \R^D$, denote by $B_R(\kk_0)$ the open ball of radius $R$ around $\kk_0$. Let $\set{\alpha(\kk)}_{\kk \in B_R(\kk_0)}$ be a continuous family of unitary matrices. Then, for some $R' \le R$, one can \emph{construct} a sequence $\set{\alpha_n(\kk)}_{\kk \in B_{R'}(\kk_0)}$, $n \in \N$, of continuous families of unitary matrices such that
\begin{itemize}
 \item $\sup_{\kk \in B_R(\kk_0)} \norm{\alpha_n(\kk) - \alpha(\kk)} \to 0$ as  $n \to \infty$, and
 \item the spectrum of $\alpha_n(\kk)$ is completely non-degenerate for all $\kk \in B_{R'}(\kk_0)$. 
\end{itemize}
\end{lemma}

The proof of the above Lemma can be found in \cite[Lemma~A.1]{CorneanMonacoTeufel17} for $D=1$ and in \cite[Lemma~5.1]{CorneanMonaco17} for $D=2$.

\begin{proof}[Proof of Proposition~\ref{prop:generic}]
The main idea is to lift all the spectral degeneracies of $\alpha$ within the unit interval $[0,1]$ or the unit square $[0,1] \times [0,1]$, and then extend the approximants with non-degenerate spectrum to the whole $\R^D$ by either periodicity or $\gamma$-periodicity.

\medskip

We start with $D=1$. By the Analytic Approximation Lemma we can find an approximant $\alpha^{(1)}$ of $\alpha$ which depends analytically on $k$. If $\alpha^{(1)}$ has degenerate eigenvalues, then they either cross at isolated points (a finite number of them in the compact interval $[0,1]$) or they stay degenerate for all $k \in [0,1]$. Pick a point in $[0,1]$ which is not an isolated degenerate point. Applying the Local Splitting Lemma, find a continuous approximant $\alpha^{(2)}$ of $\alpha^{(1)}$ for which the second option is ruled out, so that its eigenvalues cannot be constantly degenerate.

Let now $\alpha^{(3)}$ be an analytic approximation of $\alpha^{(2)}$, obtained by means of the Analytic Approximation Lemma. The eigenvalues of $\alpha^{(3)}$ can only be degenerate at a finite number of points $\set{0 < k_1 < \cdots < k_S < 1}$ (we assume without loss of generality that no eigenvalue intersections occur at $k=0$: this can be achieved by means of small shift of the coordinate). By applying the Local Splitting Lemma to balls of radius $1/n$ around each such point (starting from a large enough $n_0$), and extending the definition of the approximants from $[0,1]$ to $\R$ by periodicity, we obtain the required continuous and periodic approximants $\alpha_n$ with completely non-degenerate spectrum. Notice that, under the assumption of null-homotopy of $\alpha$, the rest of the argument of Theorem~\ref{thm:2-step} applies: in particular, for $n$ sufficiently large $\alpha_n$ admits a continuous and periodic logarithm, namely $\alpha_n(k) = \eu^{\iu  \, h_n(k)}$.

\medskip
 
Now we continue with $D=2$. We will only treat the $\gamma$-periodic setting, since the periodic case for $D\leq 2$ has been already analyzed in  
\cite{CorneanHerbstNenciu16}, \cite{CorneanMonacoTeufel17} and \cite{CorneanMonaco17}. 

We start by considering the strip $[0,1]\times \R$. The matrix $\alpha(0,k_3)$ is periodic, hence  we may find a smooth approximation  $\alpha_0(k_3)$ which is always non-degenerate and periodic. 

The matrix $\alpha(k_2,k_3) \alpha(0,k_3)^{-1}$ is close to the identity near $k_2=0$, and so is $\alpha(k_2,k_3) \alpha_0(k_3)^{-1}$. Hence if $k_2$ is close to $0$ we can write (using the Cayley transform)
$$\alpha(k_2,k_3)=\eu^{\iu  H_0(k_2,k_3)}\alpha_0(k_3)$$
where $H_0(k_2,k_3)$ is continuous, periodic in $k_3$, and uniformly close to zero. Due to the $\gamma$-periodicity of $\alpha$, we have that 
$\alpha(1,k_3)$ and  $\gamma(k_3)\alpha_0(k_3)\gamma(k_3)^{-1}$ are also close in norm. Reasoning in the same way as near $k_2=0$ we can write
$$\alpha(k_2,k_3)=\eu^{\iu  H_1(k_2,k_3)}\gamma(k_3)\alpha_0(k_3)\gamma(k_3)^{-1}$$
where $H_1(k_2,k_3)$ is  continuous, periodic in $k_3$, and uniformly close to zero near $k_2=1$. 

Let $\delta<1/10$. Choose a smooth function $0\leq g_\delta\leq 1$ such that
\[ g_\delta(k_2)= \begin{cases}
1  & \text{if }  k_2\in [0,\delta]\cup [1-\delta,1],\\
0  & \text{if } 2\delta\leq  k_2\leq 1-2\delta.
\end{cases} \]
For $0\leq k_2\leq 1$ and $k_3\in \R$, define the matrix $\alpha_\delta(k_2,k_3)$  in the following way: 
\[ \alpha_\delta(k_2,k_3):= \begin{cases}
\eu^{\iu  (1-g_\delta(k_2))H_0(k_2,k_3)} \, \alpha_0(k_3)  & \text{if } 0\leq k_2\leq 3\delta,\\
\alpha(k_2,k_3) & \text{if } 3\delta<k_2< 1-3\delta, \\
\eu^{\iu  (1-g_\delta(k_2))H_1(k_2,k_3)} \, \gamma(k_3) \, \alpha_0(k_3)\, \gamma(k_3)^{-1} & \text{if } 1-3\delta\leq  k_2\leq 1.
\end{cases}\]

We notice that $\alpha_\delta$ is continuous, periodic in $k_3$ and converges in norm to $\alpha$ when $\delta$ goes to zero. Moreover, 
$$\alpha_\delta(1,k_3)=\gamma(k_3)\,\alpha_\delta(0,k_3)\,\gamma(k_3)^{-1},$$
which is a crucial ingredient if we want to continuously extend it by $\gamma$-periodicity to $\R^2$. 

We also note that  $\alpha_\delta(k_2,k_3)$ is completely non-degenerate when $k_2$ is either $0$ or $1$, hence by continuity it must remain completely non-degenerate when $k_2\in[0,\epsilon]\cup[1-\epsilon,1]$ if $\epsilon$ is small enough. 

Following \cite{CorneanMonaco17}, we will explain how to produce an approximation $\alpha'(k_2,k_3)$ of $\alpha_\delta(k_2,k_3)$ with the following properties: 
\begin{itemize}
\item it coincides with $\alpha_\delta(k_2,k_3)$ if $k_2\in[0,\epsilon]\cup[1-\epsilon,1]$, 
\item it is continuous on $[0,1]\times \R$ and periodic in $k_3$,  
\item it is completely non-degenerate on the strip $[0,1]\times \R$. 
\end{itemize}
Assuming for now that all this holds true, let us investigate the consequences. Because it coincides with $\alpha_\delta$ near $k_2=0$ and $k_2=1$, we also have:
$$\alpha'(1,k_3)=\gamma(k_3)\alpha'(0,k_3)\gamma(k_3)^{-1}.$$
If $k_2>0$ we define recursively
$$\alpha'(k_2+1,k_3)=\gamma(k_3)\alpha'(k_2,k_3)\gamma(k_3)^{-1}$$
and if $k_2<0$ 
$$\alpha'(k_2,k_3)=\gamma(k_3)^{-1}\alpha'(k_2+1,k_3)\gamma(k_3).$$
Then $\alpha'$ has all the properties required in the statement, and the proof is complete.

\medskip

Finally let us sketch the main ideas borrowed from \cite{CorneanMonaco17} which are behind the proof of the three properties of $\alpha'$ listed above. 

First, the construction of $\alpha'$ is based on continuously patching  non-degenerate local logarithms, which is why the already non-degenerate region 
$k_2\in [0,\epsilon]\cup [1-\epsilon,1]$ is left unchanged. 

Second, let us consider the finite segment defined by $k_2\in [\epsilon,1-\epsilon]$ and $k_3=0$. The family of matrices $\set{\alpha_\delta(k_2,0)}$ is $1$-dimensional, with a spectrum which is completely non-degenerate near $k_2=\epsilon$ and $k_2=1-\epsilon$. Reasoning as in the case $D=1$ we can find a continuous  approximation $\alpha_2(k_2)$ which is completely non-degenerate on the whole interval $k_2\in [\epsilon,1-\epsilon]$. The matrix $\alpha_\delta(k_2,k_3)\alpha_2(k_2)^{-1}$ is close to the identity matrix if $|k_3|\ll 1$, hence we may locally perturb $\alpha_\delta$ near the segment $(\epsilon,1-\epsilon)\times \{0\}$ so that the new $\alpha_\delta'$  is completely non-degenerate on a small tubular neighborhood of the boundary of the segment $(\epsilon,1-\epsilon)\times \{0\}$. This perturbation must be taken small enough not to destroy the initial non-degeneracy near $k_2=\epsilon$ and $k_2=1-\epsilon$. 

Third, since $\alpha_\delta$ is periodic in $k_3$, the local perturbation around the strip  $(\epsilon,1-\epsilon)\times \{0\}$ can be repeated near all the strips $(\epsilon,1-\epsilon)\times \mathbb{Z}$. The new matrix, $\alpha_\delta''$, will be non-degenerate near a small tubular neighborhood of any unit square of the type $[0,1]\times [p,p+1]$, with $p\in \mathbb{Z}$. The final step is to locally perturb $\alpha_\delta''$ inside these squares, like in \cite[Prop.~5.11]{CorneanMonaco17}. The splitting method relies in an essential way on the condition $D\leq 2$, since it uses the fact that a smooth map between $\R^D$ and $\R^3$ cannot have regular values.
\end{proof}

\subsection{Resolvent estimates}
\label{app:KernelEstimates}

In this final Appendix we will prove the estimates on the matrix elements of the resolvent of the Hamiltonian $\HH_\epsilon$ that we used in Section~\ref{sec:Hofstadter}.
\begin{proposition}[Combes--Thomas type estimate]
	\label{Combes-Thomas}
	Consider an operator $H_0$ in $\ell^2(\Z^2) \otimes \C^Q$ such that its matrix elements are localized along the diagonal, that is,
	\begin{equation*}
	\left |H_0(\gamma,\xx;\gamma',\xx')\right |\leq C\eu^{-\beta_0\|\gamma-\gamma'\|} \qquad \forall \gamma,\gamma' \in \Z^2\, , \quad \xx, \xx' \in \set{1,\ldots,Q}
	\end{equation*}
	for some positive constants $C$ and $\beta_0$. Moreover fix a compact set $K\subset \rho(\HH_0)$.  Then, there exist two constants $C'$ and $0<\beta<\beta_0$ such that 
	\begin{equation*}
	\sup_{z\in K} \left\| (H_0-z)^{-1}(\gamma,\xx;\gamma',\xx')  \right\| \leq C' \eu^{-\beta\|\gamma-\gamma'\|} \, , \quad \forall \: \gamma, \gamma'\in \Z^2, \: \xx, \xx' \in \set{1,\ldots,Q} . 
	\end{equation*}
\end{proposition}
\begin{proof}
Take $\gamma_0 \in \Z^2$. Consider the operator $H_\beta^{(\gamma_0)}$ defined by the following matrix elements:
\begin{equation}
\label{Halpha}
H^{(\gamma_0)}_\beta(\gamma,\xx;\gamma',\xx'):=\eu^{\beta\|\gamma-\gamma_0\|}H_0(\gamma,\xx;\gamma',\xx')\eu^{-\beta\|\gamma'-\gamma_0\|} \, , 
\end{equation}
for all $\gamma, \gamma'\in \Z^2$, and $\: \xx, \xx' \in \set{1,\ldots,Q}$. Using the inequality $|\eu^{x}-1|\leq |x| \eu^{|x|}$, which holds for all $x\in\R$, together with the triangle inequality we have
\begin{equation}
\label{HDiff}
\sup_{\gamma_0 \in \Z^2} \left | H^{(\gamma_0)}_\beta(\gamma,\xx;\gamma',\xx')- H_0(\gamma,\xx;\gamma',\xx')  \right | \leq C \beta \|\gamma-\gamma'\|\,\eu^{-(\beta_0-\beta)\|\gamma-\gamma'\|} \, .
\end{equation}

Using a Schur--Holmgren estimate, as soon as $\beta<\beta_0$ we get from \eqref{HDiff} that $\|H_\beta^{(\gamma_0)}-H_0\|\leq \beta C$ for all $\gamma_0$.  If $z \in K\subset\rho(H_0)$, we can choose a $\beta$ small enough such that the operator
 $$\left(\Id +\left( H^{(\gamma_0)}_\beta-H_0 \right)\left(H_0-z\right)^{-1}\right)$$  
is invertible uniformly in $z$ and $\gamma_0$. Thus we obtain that
\begin{equation*}
\big(H^{(\gamma_0)}_\beta-z\big)^{-1}=\left(H_0-z\right)^{-1}\left(\Id + \left(H^{(\gamma_0)}_\beta-H_0\right) \left(H_0-z\right)^{-1}\right)^{-1},
\end{equation*}
which implies 
\begin{equation}
\label{RHalphanorm}
\sup_{\gamma_0\in \Z^2}\sup_{z \in K } \norm{ \big(H_\beta^{(\gamma_0)}-z\big)^{-1}} =: A<\infty \, . 
\end{equation} 
Also, $\beta$ only depends on the minimal distance between $z$ and the spectrum of $H_0$.
 
We are now ready to prove the exponential localization of the resolvent of $H_0$. From the definition~\eqref{Halpha} of $H^{(\gamma_0)}_\beta$ we obtain that $\eu^{-\beta\|\cdot -\gamma_0\|}H^{(\gamma_0)}_\beta=H_0 \eu^{-\beta\|\cdot-\gamma_0\|}$. From this identity and from \eqref{RHalphanorm} we get that for every $z \in K$
\begin{equation}
\label{CTAux1}
\left(H_0-z\right)^{-1}  \eu^{-\beta\|\cdot-\gamma_0\|} =   \eu^{-\beta\|\cdot-\gamma_0\|} \big(H^{(\gamma_0)}_\beta-z\big)^{-1} \, .
\end{equation} 
Hence $\eqref{CTAux1}$ shows that $\left(H_0-z\right)^{-1}  \eu^{-\beta\|\cdot-\gamma_0\|}$ maps in the domain of the unbounded multiplication operator $\eu^{\beta\|\cdot-\gamma_0\|}$. Finally, considering the vector $\delta_{\gamma_0,\xx'}$ that is equal to 1 only in $(\gamma_0,\xx')$, and using the fact that in the discrete setting the $\ell^\infty$ norm is bounded by the $\ell^2$ norm, \eqref{RHalphanorm} implies
\begin{equation*}
\begin{aligned}
 \left| \eu^{\beta\|\gamma-\gamma_0\|}\left(H_0-z\right)^{-1} (\gamma,\xx;\gamma_0,\xx' )\right|=\left|\Big( \big(H_\beta^{(\gamma_0)}-z\big)^{-1}  \delta_{\gamma_0,\xx'} \Big) (\gamma,\xx)\right| \leq  A \, 
\end{aligned}
\end{equation*}
which concludes the proof.
\end{proof}

For the next statement, recall that $\HH_\epsilon$ was defined in~\eqref{HHepsilon}.

\begin{proposition}
	\label{prop:DiagLoc}
	Fix a compact set $K\subset \rho(\HH_0)$. Then there exist $\epsilon_0>0$, $\alpha<\infty$ and $C<\infty$ such that for every $0\leq \epsilon\leq \epsilon_0$ we have that $K\subset \rho(\HH_\epsilon)$ and:
	\begin{equation}
	\label{DiagLocProp}
	\sup_{z\in K}\left \vert(\HH_\epsilon-z)^{-1}(\gamma,\xx;\gamma',\xx')-\eu^{\iu \epsilon \phi(\gamma,\gamma')}(\HH_0-z)^{-1}(\gamma,\xx;\gamma',\xx')\right \vert\leq C\;\epsilon \;\eu^{-\alpha \|\gamma-\gamma'\|}.
	\end{equation}
\end{proposition}
\begin{proof}
	By hypothesis we know that $|\HH_0(\gamma,\xx;\gamma',\xx')| \leq C' \eu^{-\beta\|\gamma-\gamma'\|}$ and hence Proposition~\ref{Combes-Thomas} gives us that also $\left| (\HH_0-z)^{-1}(\gamma,\xx;\gamma',\xx')\right| \leq C'' \eu^{-\beta\|\gamma-\gamma'\|}$, uniformly for every $z \in K$. Consider the operator $S_z^{(\epsilon)}$ defined by the following matrix elements:
	\begin{equation*}
	S^{(\epsilon)}_z(\gamma,\xx;\gamma',\xx'):=\ep{\gamma}{\gamma'} (\HH_0-z)^{-1}(\gamma,\xx;\gamma',\xx') \, ,
	\end{equation*}
	for all $\gamma, \gamma'\in \Z^2$, and $\: \xx, \xx' \in \set{1,\ldots,Q}$. Then consider the matrix of $(\HH_{\epsilon}-z)S^{(\epsilon)}_z$. Exploiting the magnetic phase composition rule \eqref{ianuar20} and the fact that $\ep{\gamma}{\gamma}=1$ we get
	\begin{equation}
	\label{ApproxResolvent}
	(\HH_{\epsilon}-z)S^{(\epsilon)}_z=:{\bf 1} + T_z^{(\epsilon)} \, ,
	\end{equation}
	where $ T_z^{(\epsilon)} $ is the operator associated with the matrix elements
	\begin{equation}
	\label{KerTz}
	\ep{\gamma}{\gamma'}\sum_{\tgamma\in\Z^2}\sum_{\tx=1}^{Q}  \left(\ep{\gamma-\tgamma}{\tgamma-\gamma'}-1\right) \HH_0(\gamma,\xx;\tgamma,\tx) (\HH_0-z)^{-1}(\tgamma,\tx;\gamma',\xx') \, ,
	\end{equation}
	for all $\gamma, \gamma'\in \Z^2$, and $\: \xx, \xx' \in \set{1,\ldots,Q}$. Now note that
	\begin{equation*} 
	\left|\ep{\gamma-\tgamma}{\tgamma-\gamma'}-1\right| \leq \frac{\epsilon}{2} \|\gamma-\tgamma\| \|\tgamma-\gamma'\| \,.
	\end{equation*}
	Considering the exponential localization of $\HH_0$ and $(\HH_0-z)^{-1}$, a simple computation shows that, for every $\alpha<\beta$,
	\begin{equation}
	\label{KerEstimate}
	\eu^{\alpha\|\gamma-\gamma'\|} \left|T_z^{(\epsilon)}(\gamma,\xx;\gamma',\xx') \right| \leq \widetilde{C} \epsilon \, , 
	\end{equation}
	where $\widetilde{C}$ is some constant independent of $z$. Hence a Schur--Holmgren estimate now proves that $\|T_z^{(\epsilon)}\| \leq \tilde{C}\epsilon $. So, fix an $\epsilon_0$ such that the norm of $T_z^{(\epsilon)}$ is less than 1, then for every $\epsilon\leq \epsilon_0$  we can invert the operator ${\bf 1}+T_z^{(\epsilon)}$. Due to the selfadjointness of $\HH_{\epsilon}$ we know a priori that $(\HH_\epsilon-z)$ is invertible for every $z$ such that $\mathrm{Im} \,z\neq0$. So, from \eqref{ApproxResolvent} we obtain that
	\begin{equation*}
	(\HH_\epsilon-z)^{-1}=S_z^{(\epsilon)}\left({	\bf 1}+T_z^{(\epsilon)}\right)^{-1}
	\end{equation*}
	and
	$$\left\|(\HH_\epsilon-z)^{-1} \vphantom{S_z^{(\epsilon)}}\right\| \leq \left\| S_z^{(\epsilon)}\right\| < C \, , $$
	where $C$ is a constant that depends only on $K$ and does not depend on the imaginary part of $z$. So we can conclude that $K$ is also in the resolvent set of $\HH_{\epsilon}$ whenever $\epsilon\leq \epsilon_0$. Finally, from \eqref{ApproxResolvent} we have that $ S_z^{(\epsilon)}-(\HH_\epsilon-z)^{-1} = (\HH_\epsilon-z)^{-1} T_z^{(\epsilon)}$. Since $K$ is in the resolvent set of $\HH_\epsilon$, using Proposition~\ref{Combes-Thomas} we infer that $ (\HH_\epsilon-z)^{-1}$ has matrix elements localized around the diagonal, hence \eqref{DiagLocProp} follows taking into account \eqref{KerEstimate}.
\end{proof}



\bigskip \bigskip

{\footnotesize

\begin{tabular}{rl}
(H.D. Cornean) & \textsc{Department of Mathematical Sciences, Aalborg University} \\
 &  Skjernvej 4A, 9220 Aalborg, Denmark \\
 &  \textsl{E-mail address}: \href{mailto:cornean@math.aau.dk}{\texttt{cornean@math.aau.dk}} \\
 \\
(D. Monaco) & \textsc{Dipartimento di Matematica e Fisica, Universit\`{a} degli Studi Roma Tre} \\
 &  Largo San Leonardo Murialdo 1, 00146 Roma, Italy \\
 &  \textsl{E-mail address}: \href{mailto:dmonaco@mat.uniroma3.it}{\texttt{dmonaco@mat.uniroma3.it}} \\
\\
(M. Moscolari) &  \textsc{Dipartimento di Matematica, ``La Sapienza'' Universit\`{a} di Roma} \\
&   Piazzale Aldo Moro 2, 00185 Roma, Italy \\
&   \textsl{and}  \\
&   \textsc{Department of Mathematical Sciences, Aalborg University} \\
&   Skjernvej 4A, 9220 Aalborg, Denmark \\
&  \textsl{E-mail address}:  \href{mailto:moscolari@mat.uniroma1.it}{\texttt{moscolari@mat.uniroma1.it}} \\
\\
\end{tabular}

}

\end{document}